\documentclass[11pt,reqno]{article}

\usepackage{authblk}
\usepackage{amsmath,amssymb,amsfonts,amsthm,subfloat,algorithm,url,float,epsf,psfrag,bm}
\usepackage{mathtools}

\usepackage[pdftex]{color,graphicx}
\usepackage[colorlinks,citecolor=blue,urlcolor=blue]{hyperref}

\usepackage[noabbrev,capitalize]{cleveref}

\usepackage{algorithm, algorithmic, verbatim, float}
\usepackage{marvosym}
\usepackage{titlesec}

\usepackage[pagewise]{lineno}

\usepackage[numbers]{natbib}
\usepackage{setspace,caption}
\usepackage{verbatim}
\usepackage{array, float}
\usepackage{xcolor}
\usepackage{fontenc}
\usepackage[toc,page]{appendix}
\usepackage[nottoc]{tocbibind}
\usepackage[english]{babel}
\usepackage{relsize,bm}
\usepackage{cellspace}
\usepackage{booktabs}
\usepackage{blindtext}
\usepackage{subcaption}
\usepackage[utf8]{inputenc}

\usepackage{mathrsfs}
\usepackage{amsfonts}
\usepackage{enumerate}
\usepackage{latexsym}
\usepackage{multirow}
\allowdisplaybreaks
\newtheoremstyle{exampstyle}
{8pt} 
{8pt} 
{\it} 
{} 
{\bfseries} 
{.} 
{.5em} 
{} 

\theoremstyle{exampstyle}
\usepackage{bbm}
\newtheorem{theorem}{Theorem}
\newtheorem{example}{Example}
\newtheorem{lemma}{Lemma}
\newtheorem{corollary}{Corollary}
\newtheorem{remark}{Remark}

\newtheorem{prop}{Proposition}

\newtheorem{definition}{Definition}
\newtheorem{assumption}{Assumption}

\numberwithin{equation}{section}
\numberwithin{example}{section}
\numberwithin{theorem}{section}
\numberwithin{lemma}{section}
\numberwithin{corollary}{section}
\numberwithin{prop}{section}
\numberwithin{definition}{section}
\numberwithin{remark}{section}

\newcommand{\eat}[1]{}

\DeclareMathOperator*{\argmin}{\arg\!\min}
\DeclareMathOperator*{\argmax}{\arg\!\max}

\renewcommand{\bar}[1]{\overline{#1}}
\renewcommand{\hat}[1]{\widehat{#1}}
\renewcommand{\tilde}[1]{\widetilde{#1}}
\newcommand{\E}{\mathbb{E}}

\newcommand{\R}{\mathbb{R}}
\newcommand{\G}{\mathcal{G}}
\newcommand{\indep}{\mathop{\rotatebox[origin=c]{90}{$\models$}}}

\newcommand{\B}{\boldsymbol}

\definecolor{LightCyan}{rgb}{0.88,1,1}
\definecolor{Gray}{gray}{0.9}

\usepackage{xr}
\usepackage{enumerate}

\begin{document}

\title{Multivariate Symmetry: Distribution-Free Testing via Optimal Transport}


\author{Zhen Huang\vspace{-0.2cm}\\
    Department of Statistics, Columbia University\\
  e-mail: {\url{zh2395@columbia.edu}}\\
    and \vspace{0.2cm}\\
    Bodhisattva Sen\thanks{Supported by NSF grant DMS-2015376.}\vspace{0.15cm}\\
    Department of Statistics, Columbia University\\
    e-mail: \url{bodhi@stat.columbia.edu}}

\maketitle


\abstract{The sign test (Arbuthnott, 1710) and the Wilcoxon
signed-rank test (Wilcoxon, 1945) are among the first examples of a
nonparametric test. These procedures --- based on signs, (absolute)
ranks and signed-ranks --- yield distribution-free tests for symmetry
in one-dimension. In
this paper we propose a novel and unified framework for distribution-free testing of multivariate symmetry (that includes central symmetry, sign symmetry, spherical symmetry, etc.) based on the theory of
optimal transport. Our approach leads to notions of distribution-free generalized
multivariate signs, ranks and signed-ranks. As a consequence, we develop analogues of the sign and Wilcoxon signed-rank tests that share many
of the appealing properties of their one-dimensional counterparts. In particular, the proposed tests are exactly distribution-free in finite
samples with an asymptotic normal limit, and adapt to various
notions of multivariate symmetry. We study the consistency of the
proposed tests and their behavior under local alternatives, and show
that the proposed generalized Wilcoxon signed-rank (GWSR) test is
particularly powerful against location shift alternatives. We show
that in a large class of such models, our GWSR test suffers from no loss in (asymptotic) efficiency, when compared to Hotelling's $T^2$ test, despite being nonparametric and exactly
distribution-free. An appropriately score transformed version of the GWSR statistic leads to a locally asymptotically optimal test. Further, our method can be readily used to construct distribution-free confidence sets for the center of symmetry. }

\section{Introduction}\label{sec:Intro}
Symmetry is everywhere, from the bilateral symmetry in Greek sculpture, to the description of crystalline structure in nature by modern group theory \citep{Weyl1952Symmetry,Serfling2014}. The study of symmetry in statistics also has an old and rich history. Perhaps the first published report of a nonparametric test goes back to the landmark  paper~\cite{Arbuthnott1710sign} by John Arbuthnott where he essentially performed a statistical hypothesis test for univariate {\it symmetry}, computing a $p$-value via the {\it sign test}; see e.g.,~\cite{bellhouse2001john, Gibbons2011}. Some statisticians also date the beginning of (classical) nonparametric statistics, in particular rank-based methods, to the appearance of Wilcoxon's 1945 paper \cite{wilcoxon1945individual}, which introduced the {\it Wilcoxon signed-rank test} (WSR) for testing univariate symmetry. Both these classical tests are {\it distribution-free} and have had an enormous impact in the development of nonparametric statistics.


A univariate random variable $X$ is {\it symmetric} (around 0) if $X$ has the same distribution as $-X$.
Given $n$ observations $X_1,\ldots,X_n$ from the distribution of $X$,
the sign test rejects the null hypothesis ${\rm H}_0:X\overset{d}{=}-X$ when $\left|\sum_{i=1}^n {\rm sign}(X_i) \right|$ is large; here ${\rm sign}(X_i) =1$ if $X_i \ge 0$ and is $-1$ otherwise.
In contrast, the Wilcoxon signed-rank test first sorts $|X_1|,\ldots,|X_n|$ to obtain their (absolute) ranks $\{{\rm rank}(|X_i|)\}_{i=1}^n$, and  rejects ${\rm H}_0$ for large values of $\left|\sum_{i=1}^n {\rm sign}(X_i){\rm rank}(|X_i|)\right|$.
When $X$ has a continuous distribution, the classical sign and signed-rank tests enjoy several properties that make them powerful and useful in practice, including: (i) {\it distribution-freeness}, (ii) {\it independence} of signs and absolute ranks (under ${\rm H}_0$), (iii) simple asymptotic normal limiting distributions, and (iv) high (relative) efficiency compared to the $t$-test~\citep{Hodges1956efficiency,chernoff1958}; see~\citep[Chapter 15]{vanderVaart1998} for a detailed discussion of these results.

It is natural to ask if generalizations of these tests, that satisfy the above mentioned desirable properties, exist in multi-dimensions. In this paper, we answer this question in the affirmative by proposing a novel framework for testing {\it multivariate symmetry}, using the theory of optimal transport~\citep{Villani03,villani2009OT}. 

While symmetry in one dimension is unambiguous (i.e., $X\overset{d}{=}-X$), in the multivariate setting there are various notions of symmetry. The most prominent ones include: (a)  {\it central symmetry}~\citep{PuriSen1967,Blough1989PP,Ngatchou2009symmetry,Oja2010nonp,EG2016}, (b) {\it sign symmetry}~\citep{efron1969,Pinelis1994T2,Oja2010nonp}, and (c) {\it spherical symmetry}~\citep{Baringhaus1991,Ghosh1992charac,Henze2014sph,alb2020sph}; a thorough review of the topic can be found in~\citet{Serfling2014}.


In this paper, we consider the following general setting for testing multivariate symmetry. Let $\G$ be a compact subgroup of the orthogonal group ${\rm  O}(p)$ (for $p \ge 1$) --- the group of all $p \times p$ real orthogonal matrices, where the group operation is the usual matrix multiplication.
For a random vector
${\bf X}  \sim \mu \in \mathcal{P}_{\rm ac}(\mathbb{R}^p)$, where $\mathcal{P}_{\rm ac}(\mathbb{R}^p)$ denotes the class of all Borel probability measures on $\R^p$ with a Lebesgue density, we are interested in testing the hypothesis of {\it $\G$-symmetry}:
\begin{equation}\label{eq:hypo}
	\textrm{H}_0: {\bf X}  \stackrel{d}{=}  Q  {\bf X}{\rm\ \ for\ all\ \ } Q\in \G \qquad\qquad \textrm{versus}\qquad\qquad\textrm{H}_1: \textrm{not H}_0.
\end{equation}
In particular, the prominent examples of symmetry mentioned earlier correspond to the following groups $\G$.
\begin{enumerate}
\item[(a)] {\it Central symmetry}: Here $\G = \{I, -I\}$, where $I\equiv I_p$ is the identity matrix of order $p$.

Central symmetry (or ``reflective'' or ``diagonal'' or ``simple'' or ``antipodal'' symmetry \citep{PuriSen1967,Blough1989PP,Ngatchou2009symmetry,Oja2010nonp,EG2016}) arises naturally from paired data 
(see e.g., \cite[Section 5]{Gibbons2003nonpara}):
Given a sample of $n$ pairs $\{(\mathbf{Y}_i,\mathbf{Z}_i)\}_{i=1}^n$,
it is often of interest to study the differences
$\mathbf{X}_i:=\mathbf{Y}_i-\mathbf{Z}_i$, $i=1,\ldots,n$.
Under the null hypothesis that $\mathbf{Y}_i$ and $\mathbf{Z}_i$ are exchangeable, $\mathbf{X}_i$ will have a centrally symmetric distribution \citep{Lehmann75, Gibbons2003nonpara}.
Extensions of central symmetry include {\it angular symmetry} \citep{Randles1989sign},
which corresponds to $\frac{{\bf X}}{\|{\bf X}\|} \stackrel{d}{=}  - \frac{{\bf X} }{\|{\bf X} \|}$; here $\|\cdot\|$ denotes the usual Euclidean norm.

\item[(b)] {\it Sign symmetry}: Here $\G = \{Q \in \R^{p \times p}: Q {\rm\ is\ a\ diagonal\ orthogonal\ matrix}\}$; thus each diagonal entry of $Q$ is $\pm 1$.

Sign symmetry is also called {\it orthant symmetry} or {\it marginal symmetry} \citep{efron1969,Pinelis1994T2,Oja2010nonp}. Such sign-change invariance is often assumed in the {\it location-scatter} model, which provides tools for robust estimation of the regular mean vector and covariance matrix \citep{Oja2010nonp}; also see the book by Puri and Sen~\citep{Puri-Sen-1971} which gives a comprehensive account of multivariate analysis methods based on marginal signs and ranks.

\item[(c)] {\it Spherical symmetry}: Here $\G = {\rm O}(p)$.

Spherically symmetric distributions are natural generalizations of the Gaussian distribution, and have drawn substantial interest~\citep{Fang1990symm,alb2020sph}.
It is also known that many estimation results are robust when departures from the normality assumption are in the direction of the spherically symmetric family, with applications in the general linear model \citep{Jammalamadaka1987linear,Berk1989opt,Fourdrinier1995sph}, sufficient dimension reduction \citep{Li2005contour,Zeng2010central}, etc.
Tests for spherical symmetry have been considered in \citep{Baringhaus1991,Ghosh1992charac,Henze2014sph,alb2020sph}.

\end{enumerate}
Clearly, other subgroups of ${\rm O}(p)$, such as the {\it signed permutation group}\footnote{The signed permutation group consists of all matrices having exactly one nonzero entry in each row and exactly one nonzero entry in each column, and the nonzero entries are either 1 or $-1$.} \citep{fire2006statistics}, also correspond to a notion of symmetry.
The $\G$-symmetric model described by ${\rm H}_0$ in \eqref{eq:hypo} is motivated by the work of~\citet{beran1997multivariate}, and it provides a general framework for testing many notions of multivariate symmetries in a unified fashion.

To develop distribution-free tests for the general hypothesis of $\G$-symmetry in~\eqref{eq:hypo} we first observe (in~\cref{sec:OT_signs_ranks}) that the one-dimensional signs, ranks and signed-ranks --- which are the building blocks for distribution-free inference in the univariate setting --- can be obtained as solutions to an optimization problem which can be generalized to multi-dimensions. This resulting optimization is an optimal transport (OT) problem~\cite{Villani03, villani2009OT} with a specific cost function that can be solved efficiently (see \cref{sec:comp_efficiency}). 

Specifically, given observations $\mathbf{X}_1,\ldots,\mathbf{X}_n$ from $\mu\in \mathcal{P}_{\rm ac}(\R^p)$ and (known) candidate rank vectors $\mathbf{h}_1,\ldots,\mathbf{h}_n\in \mathbb{R}^p$, the {\it generalized sign} $S_n(\mathbf{X}_i)$ is defined as an element of $\G$ and the collection of {\it generalized ranks} $(R_n(\mathbf{X}_1),\ldots, R_n(\mathbf{X}_n))$ is a permutation of $(\mathbf{h}_1,\ldots,\mathbf{h}_n)$; see \cref{sec:OT} for the formal definitions. Our approach leads to notions of ($\G$-symmetric specific) generalized multivariate signs, ranks and signed-ranks that are also distribution-free (under ${\rm H}_0$). In the one-dimensional case, if $(\mathbf{h}_1,\ldots,\mathbf{h}_n)$ is chosen as $(\frac{1}{n},\ldots,\frac{n}{n})$, then we recover the usual notion of signs and ranks. Further, mirroring the one-dimensional setting, we define the {\it generalized sign} test statistic as $\sum_{i=1}^n S_n(\mathbf{X}_i)$, and the statistic for the {\it generalized Wilcoxon signed-rank} (GWSR) test will be defined as
$\sum_{i=1}^n S_n(\mathbf{X}_i)R_n(\mathbf{X}_i)$; see~\cref{sec:ARE} for the details.

Over the past few years, multivariate ranks~\citep{hallin2017distribution,Nabarun2021rank}, defined via the theory of OT, have been effectively used to design distribution-free testing procedures in two-sample problems \citep{boeckel2018multivariate,Nabarun2021rank}, for independence testing \citep{Nabarun2021rank,Shi2022indep,shi2022universally,shi2023semiparametrically}, in multivariate linear models~\citep{hallin2020fully,Hallin2022VARMA}, for directional data \citep{hallin2022nonparametric}, etc. Our framework, although inspired by these recent developments, is novel and more general and includes the usual notion of OT based multivariate ranks~\cite{hallin2017distribution,Nabarun2021rank} as a special case. We achieve this by considering a more general cost function in the OT problem that naturally arises from the hypothesis testing problem~\eqref{eq:hypo} and reduces to the usual squared Euclidean distance cost when $\G$ is the trivial group $\{I\}$. Further, in Section~\ref{sec:popu_signs_ranks}, using the theory of OT, we provide characterizations of the population analogues of the empirical generalized signs, ranks and signed-ranks (see~\cref{thm:population_rank}). Moreover, we show the consistency of these sample estimates to their population analogues in~\cref{cor:conv_popu}. Our treatment does not need any moment assumptions on the data distribution $\mu$ as we appeal to geometric characterizations of OT (as pioneered by Robert McCann~\cite{McCann1995}).

Among the important reasons for the  popularity of the classical Wilcoxon signed-rank test are
the striking results of \citet{Hodges1956efficiency} and \citet{chernoff1958}, where the authors show that the asymptotic (Pitman) relative efficiency of Wilcoxon’s test with respect
to (w.r.t.) the one-sample Student’s $t$-test, under location-shift alternatives, never falls below 0.864 (with the identity score function) and 1 (with the Gaussian score function) respectively, despite the former being exactly distribution-free for all sample sizes. Motivated by these results, we study in \cref{subsec:efficiency} the (asymptotic relative) efficiency, in the sense of Pitman \citep[see e.g.,][Section 14.3]{vanderVaart1998}, of our proposed multivariate generalization of the Wilcoxon signed-rank test compared to Hotelling's $T^2$ test --- the multivariate generalization of the usual $t$-test. In particular, we show that the proposed GWSR test satisfies Hodges-Lehmann and Chernoff-Savage-type efficiency lower bounds over natural sub-families of multivariate distributions, despite being entirely agnostic to the underlying data generating mechanism; see Theorems~\ref{thm:indepcomp}-\ref{thm:blind}. 

Let us now briefly highlight the main results in this paper.
\begin{enumerate}
\item Distribution-freeness: under ${\rm H}_0$,
\begin{enumerate}
\item $S_n(\mathbf{X}_1),\ldots,S_n(\mathbf{X}_n)$ are i.i.d.\ following the uniform distribution over $\G$.
\item The ranks $(R_n(\mathbf{X}_1),\ldots,R_n(\mathbf{X}_n))$ are uniformly distributed over all $n!$ permutations of $(\mathbf{h}_1,\ldots,\mathbf{h}_n)$; see \cref{prop:properties_of_generalized_sign_rank}.
\end{enumerate}
\item Independence: $(S_n(\mathbf{X}_1),\ldots,S_n(\mathbf{X}_n))$ is independent of $(R_n(\mathbf{X}_1),\ldots,R_n(\mathbf{X}_n))$ under $\textrm{H}_0$ (\cref{prop:properties_of_generalized_sign_rank}).
\item Asymptotic normality: Both $\frac{1}{\sqrt{n}}\sum_{i=1}^n S_n(\mathbf{X}_i)$ and $\frac{1}{\sqrt{n}}\sum_{i=1}^n S_n(\mathbf{X}_i)R_n(\mathbf{X}_i)$ converge in distribution to multivariate normal distributions, under ${\rm H}_0$ (\cref{thm:CLT_wilcoxon}).
\item Relative efficiency against location shift alternatives for suitable sub-classes of multivariate distributions (e.g., product distributions, elliptically symmetric distributions, etc.):
\begin{enumerate}
\item (Hodges-Lehmann phenomenon~\citep{Hodges1956efficiency}) The relative efficiency of the GWSR test w.r.t.\ Hotelling's $T^2$ test never falls below 0.864, if $\mathbf{h}_1,\ldots,\mathbf{h}_n$ and $\G$ determine a uniform {\it effective reference distribution} (see Theorem~\ref{thm:indepcomp} for the precise statement).
\item (Chernoff-Savage phenomenon~\citep{chernoff1958}) The relative efficiency of the GWSR test against Hotelling's $T^2$ test is lower bounded by 1, if $\mathbf{h}_1,\ldots,\mathbf{h}_n$ and $\G$ determine a Gaussian {\it effective reference distribution} (see Theorems~\ref{thm:indepcomp}-\ref{thm:blind}).
\end{enumerate}
\item Efficiency (against location shift alternatives): The family of score transformed GWSR tests (with statistic $\frac{1}{\sqrt{n}}\sum_{i=1}^n S_n(\mathbf{X}_i)J(R_n(\mathbf{X}_i))$ for some score function $J:\R^p\to \R^p$) contains a {\it locally asymptotically optimal test} (see \cref{thm:efficiency}), i.e., a test which has maximum (asymptotic) local power among all possible tests with a fixed Type I error. This is in parallel with the analogous result in univariate setting, cf.~\citet[Corollary 15.10]{vanderVaart1998}.
\end{enumerate}

We say that $\mathbf{X}$ is $\G$-symmetric around ${\B \theta} \in \R^p$ (usually referred to as the {\it center of $\G$-symmetry}) if $(\mathbf{X} - {\B \theta}) \overset{d}{=} Q(\mathbf{X} - {\B \theta})$, for all $Q \in \G$. Our framework can be immediately extended to test the hypothesis of $\G$-symmetry around any fixed and known ${\B \theta}$. To test symmetry around a general center ${\B \theta} \in \R^p$, one may just translate the data $\{\mathbf{X}_i\}_{i=1}^n$ to $\{\mathbf{X}_i-{\B \theta}\}_{i=1}^n$, thereby reducing the problem to~\eqref{eq:hypo}. Further, if ${\B \theta}$ is unknown, our proposed testing framework can be used to construct {\it distribution-free confidence sets} for ${\B \theta}$ (see~\cref{subsec:conf_set}). Further, in Section~\ref{sec:simulations} we demonstrate the superior performance of our proposed tests (especially for location shift models) via simulation studies, illustrate the implications of our efficiency results for data sets with small to moderate sample sizes, and compare our proposals with existing methods that test for different notions of symmetry.

From the above properties, we believe that our generalized multivariate signs and ranks (defined via OT) are natural extensions of the classical concepts. As far as we are aware, our work is the first to propose notions of multivariate distribution-free signs and ranks (relative to any group $\G$) and develop distribution-free testing for~\eqref{eq:hypo} in multi-dimensions with the above efficiency lower bounds, analogous to the classical Wilcoxon's signed-rank test.

\subsection{Literature Review}
Multivariate symmetry is an important concept in nonparametric statistics \citep{ZS2000}. Testing for multivariate symmetry has been studied extensively in the literature, see, e.g.,~\citet{Baringhaus1991},~\citet{DH1999} and \citet{Serfling2014}.

\textbf{Extensions of sign in multi-dimension.}
One of the earliest and most popular tests for symmetry is possibly the sign test \citep{Arbuthnott1710sign}.
Since then, several notions of multivariate signs have been proposed in the literature. The component-wise multivariate sign, which is a popular generalization,
dates back to \citet{Bennett1962sign,Chatterjee1966bivariate}. 
The asymptotic normality and Pitman efficiency of the component-wise sign test against Hotelling's $T^2$ was derived by \citet{Bickel1965competitors}.
Component-wise ranks and signed-ranks have also been used for testing \citep{Puri-Sen-1971,Marie1984sym}.
One drawback of the component-wise sign is that its distribution-freeness requires $\mathbf{X}$ to be sign symmetric (it is not distribution-free under central symmetry).
We will show that our generalized sign (see~\eqref{eq:Sign-Rank}) reduces to the component-wise sign when $\G$ is the group corresponding to sign symmetry (see \cref{rk:comp_sign}).

A multivariate generalization of the sign test based on {\it interdirections} was proposed in \citet{Randles1989sign}, generalizing the two-sided univariate sign test and Blumen's bivariate sign test \citep{blumen1958new}. This test statistic is invariant under nonsingular linear transformations and it has a finite-sample distribution-free property over distributions with elliptical directions \citep{Randles1989sign}.
\citet{CS1993} expressed Randles's statistic in terms of a data-dependent sign, and further unified a number of previous affine invariant generalized sign tests \citep{Hodges1955sign,Watson1961circle,Ajne1968circular}.
Such transformation re-transformation technique was later investigated to define affine equivariant multivariate median and quantiles
\citep{Chakraborty1996TR,CCO1998,Chakraborty2001quantiles,Serfling2010Equivariance}, and affine invariant multivariate sign test \citep{Randles2000}.
Locally asymptotically maximin tests based on interdirections were derived in \citet{Hallin2002optimal} under elliptical symmetry.

\citet{Mottonen1995sign} defined the multivariate sign as the unit vector in the direction of a given point, which is called the {\it spatial sign}. The Pitman efficiencies of the corresponding tests and a generalization of the Hodges-Lehmann estimator were further analyzed in \citet{Mottonen1997efficiency,Mottonen2005Multivariate}. However, none of the above notions of signs are distribution-free beyond spherical symmetry.

\citet{hallin2017distribution} proposed a novel center-outward definition of ranks and signs based on measure transportation with attractive properties, similar to our framework with $\G= \{I\}$. However, their center-outward signs differ from the classical signs and our proposal in multiple ways; e.g., the center-outward signs are not i.i.d.\ uniform over $\{\pm 1\}$ for symmetric distributions in one dimension. 


\textbf{Tests for central symmetry.}
Going beyond the sign test, there are tests for central symmetry which are consistent against all alternatives.
In the univariate setting, \citet{McWilliams1990runs} presented a simple test based on a runs statistic. 
\citet{Dyckerhoff2015depth} extended McWilliams's procedure to test for bivariate central symmetry using the notion of data depth. 
In~\citet{EG2016}, omnibus tests for central symmetry of a bivariate probability distribution are proposed that are asymptotically distribution-free. 
Another class of approaches explores the properties of the characteristic function of a centrally symmetric distribution and uses the
empirical characteristic function to construct test statistics, including both in the univariate setting \citep{csorgo1987symmetry,Feuerverger1977characteristic} and in multi-dimensional scenarios \citep{Ghosh1992charac,Heathcote1995symmetry,Henze2003inv,Ngatchou2009symmetry,Henze2020news}. 
Affine invariant tests for central symmetry 
are proposed in \citet{Henze2003inv} using kernel density estimators; also see~\citep{Henze2020news}.
However, these tests are in general not asymptotically distribution-free, and the limiting distributions need to be estimated from the data for inference \citep{Ngatchou2009symmetry,Henze2020news}.

\textbf{Tests for spherical symmetry.}
A direct extension of McWilliams's procedure for testing spherical symmetry in multi-dimension is given in \citet[Section 2.4]{marden1999multivariate}, which has a simple asymptotic $\chi^2_1$ distribution under the null. In~\citet{Baringhaus1991}, rotationally invariant tests based on test statistics of the von Mises type are proposed, which give rise to tests that are consistent against all alternatives. 
The characteristic function has also been used for testing
spherical symmetry in multi-dimensional scenarios \citep{Henze2014sph,Ghosh1992charac}.~\citet{Kariya1977sph} consider two classes of alternatives to test spherical symmetry, and provide uniformly most powerful tests. A new testing procedure based on the fact that $\mathbf{X}$ is spherically symmetric if and only if $\E[\mathbf{u}^\top\mathbf{X}|\mathbf{v}^\top\mathbf{X}]=0$ for any two perpendicular vectors $\mathbf{u}$ and $\mathbf{v}$ is given in \citet{alb2020sph}, with its rejection region being found via bootstrap.


\textbf{Tests for sign symmetry.}
The analysis of the classical Hotelling's $T^2$ test has also been extended to sign symmetric distributions by
\citet{efron1969,Pinelis1994T2}, etc.
In \citet{Oja2010nonp}, a book-length treatment on tests based on spatial signs and ranks is given;
sign-change versions of the spatial sign test and spatial signed-rank test are proposed when the observations are sign symmetric.
These tests also have good efficiency properties with respect to Hotelling's $T^2$ in the multivariate $t$-distribution settings \citep[Chapter 8]{Oja2010nonp}. However, Hodges-Lehmann~\citep{Hodges1956efficiency} and Chernoff-Savage~\citep{chernoff1958} type lower bounds are not available for these methods.

\subsection{Organization of the Paper}
In \cref{sec:OT_signs_ranks}, we define multivariate signs, ranks and signed-ranks, and provide their population analogues.
In \cref{sec:ARE}, we define multivariate analogues of the sign test and the Wilcoxon signed-rank test. We further investigate the statistical properties of these proposed tests such as distribution-freeness, asymptotic normality, consistency, and efficiency.
Finite sample performance of the proposed tests and comparisons with competing methods are given in \cref{sec:simulations}. A distribution-free confidence set for the unknown center of symmetry of a $\G$-symmetric distribution is constructed in~\cref{subsec:conf_set}. All proofs of our results and further discussions are relegated to Appendices~\ref{sec:discussions}-\ref{sec:Appendix-C}.

\section{Multivariate Generalized Signs, Ranks and Signed-ranks}\label{sec:OT_signs_ranks}
\subsection{Preliminaries} 
Let ${\rm  O}(p)$, the orthogonal group in dimension $p$, denote the group of all $\R^{p \times p}$ orthogonal matrices under the group operation of matrix multiplication. Let $\G$ be a compact subgroup of the orthogonal group ${\rm  O}(p)$.
In particular, $\G$ satisfies the following group axioms: (i) $I\in\G$; here $I\equiv I_p$ is the identity matrix of order $p$; (ii) if $Q\in \G$, then $Q^{-1}\in\G$; (iii) if $Q_1$, $Q_2\in\G$, then $Q_1Q_2\in \G$; (iv) if $Q_1$, $Q_2$, $Q_3\in\G$, then $(Q_1Q_2)Q_3=Q_1(Q_2Q_3)$.
The compactness assumption ensures that there exists a unique uniform distribution (also called the Haar measure) over $\G$; see e.g., \citet{Eaton1989group}.

\begin{definition}[Orbit]\label{defn:Orbit}
The {\it orbit} of an element $\mathbf{x} \in \R^p$ is the set of elements (in $\R^p$) to which $\mathbf{x}$ can be moved by the elements of $\G$, i.e., $\{Q \mathbf{x}: Q \in \G\}$.
\end{definition}

Note that, if $\mathbf{y}$ is in the orbit of $\mathbf{x}$, then $\mathbf{x}$ is also in the orbit of $\mathbf{y}$. Thus, 
$\R^p$ can be partitioned into a disjoint union of orbits.
Let $B \subset \R^p$. We define the set 
\begin{equation}\label{eq:defGB}
    \G B := \{Q \mathbf{x}: Q \in \G, \mathbf{x} \in B\}  \subset \R^p
\end{equation} as the image of the group action of $\G$ on $B$.
\begin{definition}[Group action]\label{defn:Group-Action}
A (left) group action of $\G$ on $\G B$ is a function $\cdot: \G \times \G B \to \G B$ which satisfies the following two axioms: (i) $I \cdot \mathbf{x}= \mathbf{x}$, and (ii) $Q_1 \cdot (Q_2 \cdot \mathbf{x}) = (Q_1 Q_2) \cdot \mathbf{x}$, for all $\mathbf{x} \in \G B$, and $Q_1, Q_2 \in  \G $. In our setting we take the group action $\cdot$ to be the usual matrix multiplication operation. 
\end{definition}

\begin{definition}[Free group action]\label{defn:Free-Group-Action}
We say that $\G$ acts {\it freely} on $\G B$ if for $\mathbf{y}\in B$ and $Q\in \G$, $Q \mathbf{y}=\mathbf{y}$ implies $Q= I$ (i.e., for any $\mathbf{x}$ in $\G B$, only the identity in $\G$ leaves $\mathbf{x}$ fixed).
\end{definition}
Intuitively, a free group action helps identify the unique element in $\G$ that maps a point to another point (i.e., $Q_1\mathbf{x}=Q_2\mathbf{x}$ for $Q_1,Q_2\in\G$ and $\mathbf{x}\in\G B$ implies that $Q_1=Q_2$).
Let $\mathcal{P}_{\rm ac}(\mathbb{R}^p)$ denote the class of all Borel probability measures on $\R^p$ with a Lebesgue density.
We say that a probability measure $\mu$ is concentrated on a set $B$ if $\mu(B)=1$.
If $\mu$ is a Borel measure, the support of $\mu$ is the smallest closed set on which $\mu$ is concentrated.

\begin{definition}[Push forward distribution]
Let $F:\R^p\to\R^k$ (for $k \ge 1$) be a function and $\mu$ be a distribution on $\R^p$. The push forward distribution $F\# \mu$ is defined as the distribution of $F(\mathbf{X})$, where $\mathbf{X}\sim \mu$.
\end{definition}

\begin{definition}[Cyclical monotonicity]\label{def:c-cyc}
Let $c:\R^p\times\R^p\to(-\infty,+\infty]$ be a function. A subset $\Gamma\subset \R^p\times \R^p$ is said to be $c$-cyclically monotone if for any $N\in\mathbb{N}$ and any family $(\mathbf{x}_1,\mathbf{y}_1),\ldots,(\mathbf{x}_N,\mathbf{y}_N)$ of points in $\Gamma$, the following inequality holds:
$$\sum_{i=1}^Nc(\mathbf{x}_i,\mathbf{y}_i)\leq \sum_{i=1}^Nc(\mathbf{x}_i,\mathbf{y}_{i+1}),$$
with the convention $\mathbf{y}_{N+1}=\mathbf{y}_1$.
\end{definition}

\subsection{Connection to Optimal Transport}\label{sec:OT}
For a random vector ${\bf X} \sim \mu \in \mathcal{P}_{\rm ac}(\mathbb{R}^p)$ and a compact subgroup $\G$ of the orthogonal group ${\rm  O}(p)$, we are interested in distribution-free testing of the hypothesis of $\G$-symmetry~\eqref{eq:hypo}. In this subsection we define new distribution-free (in finite samples) multivariate generalizations of signs and ranks that will form the building blocks for distribution-free inference. Further, these generalizations of signs and ranks --- which reduce to the usual classical notions of sign and rank when $p=1$ and $\G = \{I ,-I\}$ --- adapt to different notions of multivariate symmetry, and share many properties similar to their classical counterparts. 
We borrow ideas from 
the theory of optimal transport (OT) to define 
our multivariate generalizations.
Although our approach is inspired by the notion of multivariate ranks defined via OT~\citep{hallin2017distribution,Nabarun2021rank}, our framework is novel and more general.

Let $\nu $ be some known and fixed Borel measure on $\mathbb{R}^p$, which will be referred to as the {\it reference distribution}; see \cref{rk:ref_mea} for concrete examples. 
For each $n$, we consider a quasi-Monte Carlo \citep{halton1964algorithm,Kuipers1974unif,Harald1992mc} or Monte Carlo \citep{Hasting1970monte,Robert2004monte} sequence $\{\mathbf{h}_{n,j} \in\mathbb{R}^p:\, j = 1,2,\ldots,n\}$ assumed to be known and fixed, such that the empirical distribution over $\{\mathbf{h}_{n,j}\}_{j=1}^n$, denoted by $\nu_n$, converges in distribution to $\nu$ as $n\to \infty$. So $\nu_n$ provides a discretization of $\nu$.
For notational brevity, we hide $n$ and write $\mathbf{h}_{j}\equiv \mathbf{h}_{n,j} $.
For instance, $\{\mathbf{h}_j\}_{j=1}^n$ can be taken as $n$ i.i.d.\ observations from $\nu$ and then fixed. If $\nu={\rm Uniform}(0,1)^p$, we can also let $\{\mathbf{h}_j\}_{j=1}^n$ be the first $n$ points in the Halton sequence \citep{halton1964algorithm}; other choices are also available (see e.g., \citet[Appendix D.3]{Nabarun2021rank}).
As we will see later, to recover the classical one-dimensional sign and signed-rank, one can take $\nu ={\rm Uniform}(0,1)$ and $\mathbf{h}_{j} = \frac{j}{n}$, for $j = 1,2,\ldots,n$.

Let $\mathcal{S}_n$ be the set of all $n!$ permutations of $\{1,\ldots ,n\}$. Consider the following optimization problem:
\begin{equation}\label{eq:OTproblem}
\min \left\{\sum_{i=1}^n \lVert Q_i^\top \mathbf{X}_{\sigma(i)}-\mathbf{h}_i\rVert^2:Q_i \in \G, \sigma\in \mathcal{S}_n\right\}.
\end{equation}
Let $(\{\hat{Q}_i\}_{i=1}^n,\hat{\sigma})$ be a minimizer of \eqref{eq:OTproblem}. 
In fact, under mild assumptions, $\hat{\sigma}$ is actually a.s.\ unique, and $\{\hat{Q}_i\}_{i=1}^n$ is also a.s.\ unique if $\G$ induces a free group action (see \cref{prop:sample_uniq}). Define the cost function
\begin{equation}\label{eq:cost}
c(\mathbf{x},\mathbf{h}):=\min_{Q\in \G}\|Q^\top \mathbf{x}-\mathbf{h}\|^2 , \qquad \mbox{for} \;\; \mathbf{x},\mathbf{h} \in \R^p,
\end{equation}
so that~\eqref{eq:OTproblem} reduces to $\min_{\sigma \in \mathcal{S}_n} \sum_{i=1}^n c(\mathbf{X}_{\sigma(i)},\mathbf{h}_i)$ --- the optimal transport (OT) problem from $\{\mathbf{X}_1,\ldots, \mathbf{X}_n\}$ to $\{\mathbf{h}_1,\ldots, \mathbf{h}_n\}$ under the cost function $c(\cdot,\cdot)$. Observe that the cost function satisfies $c(\mathbf{x},\mathbf{h}) = c(\mathbf{x},Q\mathbf{h})=c(Q\mathbf{x},\mathbf{h})$ for any $Q \in \G$, i.e., it is invariant under transformations in $\G$.

Define the {\it generalized sign} $S_n:\{\mathbf{X}_1,\ldots,\mathbf{X}_n\}\to \G$ and {\it generalized rank} $R_n:\{\mathbf{X}_1,\ldots,\mathbf{X}_n\}\to\{\mathbf{h}_1,\ldots,\mathbf{h}_n\}$ as:
\begin{equation}\label{eq:Sign-Rank}
S_n(\mathbf{X}_i) = \hat{Q}_{\hat{\sigma}^{-1}(i)},\qquad \mbox{and}\qquad R_n(\mathbf{X}_i) = \mathbf{h}_{\hat{\sigma}^{-1}(i)}, \qquad \mbox{for} \;\; i = 1,\ldots, n.
\end{equation}
The {\it generalized signed-rank} of the $i$-th observation $\mathbf{X}_i$ is then defined as $S_n(\mathbf{X}_i)R_n(\mathbf{X}_i) $.

\begin{remark}[When $p=1$]\label{rem:p=1}
Note that when $p=1$  and $\G = \{1,-1\}$, the generalized signs and ranks (in~\eqref{eq:Sign-Rank}) reduce to the classical one-dimensional concepts. In fact, problem~\eqref{eq:OTproblem} directly yields the usual signs and ranks in this case. To see this observe that, for $x \in \R$ and $h \in [0,1]$, from~\eqref{eq:cost}, $c(x,h) = \min \{(x -h)^2, (-x - h)^2\}$, which (due to the linear ordering of the real line) reduces to $c(x,h) = (|x| -h)^2$. It is now easy to see that~\eqref{eq:OTproblem}, with $h_j = \frac{j}{n}$ for all $j$, simplifies to $$\min_{\sigma \in \mathcal{S}_n} \sum_{i=1}^n \left(|X_i| - \frac{\sigma(j)}{n}\right)^2 = \max_{\sigma \in \mathcal{S}_n} \sum_{i=1}^n |X_i|\frac{\sigma(j)}{n},$$ which, by the rearrangement inequality (see e.g.,~\cite[ Theorem 368]{Hardy1952}), yields that $\hat{\sigma}(i)$ will equal the usual (absolute) rank of $|X_i|$, for $i = 1,\ldots, n$. Thus, $X_i$ with the $j$-th largest absolute value will be mapped to $j/n$, and the sign of $X_i$ will be 1 if $X_i \ge 0$ and $-1$ if $X_i<0$.
\end{remark}

When $p > 1$, due to the absence of a canonical ordering in $\R^p$, optimization problem~\eqref{eq:OTproblem} does not have an explicit solution. However, the approach is similar; optimization problem~\eqref{eq:OTproblem} assigns the data points
$\{\mathbf{X}_i\}_{i=1}^n$ to the canonical multivariate rank vectors
$\{\mathbf{h}_j\}_{j=1}^n$ in a bijective fashion such that the objective
function in~\eqref{eq:OTproblem} is minimized. To develop some intuition for the generalized signs and signed-ranks (as in~\eqref{eq:Sign-Rank}), maybe it is easier to express the cost function as $c(\mathbf{x},\mathbf{h}) =
\min_{Q \in \G} \|\mathbf{x} - Q\mathbf{h}\|^2$. Thus, if $R_n(\mathbf{X}_i) = \mathbf{h}_j$, then
$c(\mathbf{X}_i,\mathbf{h}_j) = \min_{Q \in \G} \|\mathbf{X}_i - Q\mathbf{h}_j\|^2$ and hence the signed-rank
$S_n(\mathbf{X}_i) R_n(\mathbf{X}_i) \equiv S_n(\mathbf{X}_i) \mathbf{h}_j $ is the closest element (in Euclidean norm) to
$\mathbf{X}_i$ in the orbit of $\mathbf{h}_j$. Moreover, $S_n(\mathbf{X}_i)$, the sign of $\mathbf{X}_i$, is the argmin of the cost $\|\mathbf{X}_i - Q\mathbf{h}_j\|^2$ over $\G$ --- it is the orthogonal matrix in $\G$ that transforms the rank vector $R_n(\mathbf{X}_i) \equiv \mathbf{h}_j$ to
bring it closest to $\mathbf{X}_i$. 
\begin{figure}
     \centering
     \hspace{-0.32in}
     \begin{subfigure}[b]{0.45\textwidth}
         \centering
         \includegraphics[width=\textwidth]{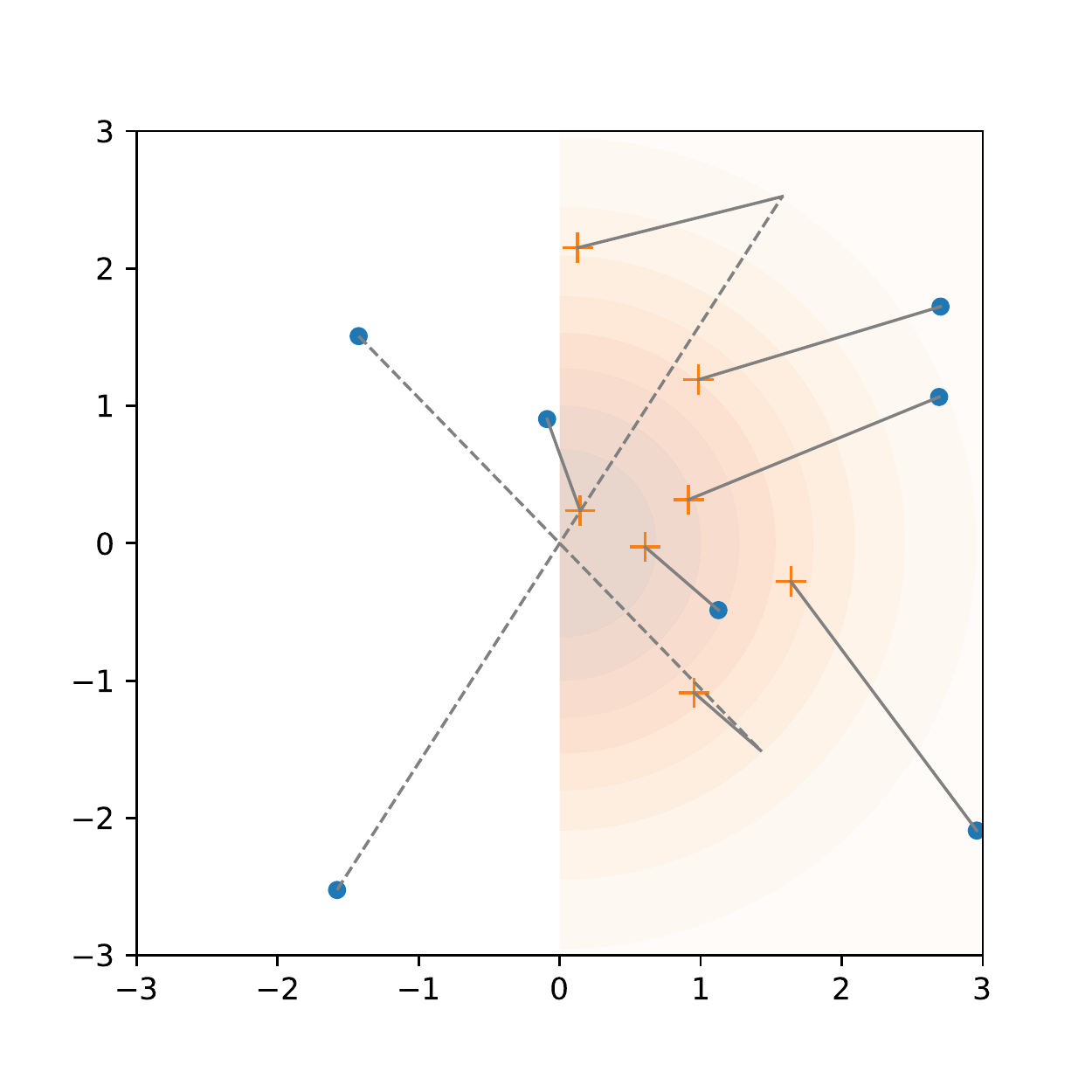}
     \end{subfigure}
     \hspace{0.2in}
     \hspace{-0.32in}
     \begin{subfigure}[b]{0.42\textwidth}
         \centering
         \includegraphics[width=\textwidth]{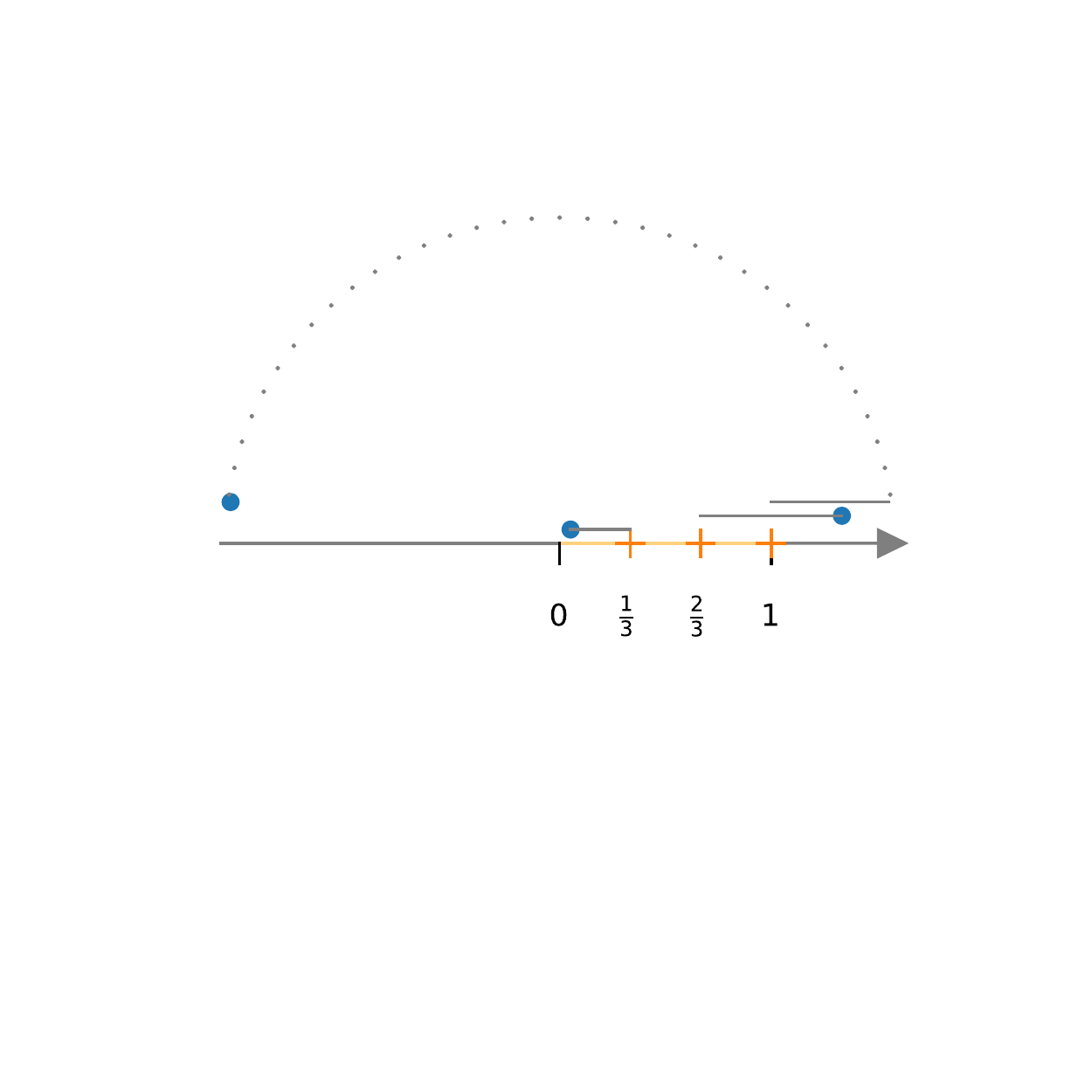}
     \end{subfigure}
        \caption{\small {\bf Left}: The data points $\mathbf{X}_1,\ldots,\mathbf{X}_n$ (marked by ``{\color{blue} $\bullet$}") with $n=7$ drawn from $N(\mathbf{0}_2,4I_2)$.
        The rank vectors $\mathbf{h}_1,\ldots,\mathbf{h}_n$ (marked by ``{\color{orange}+}") are taken as $n$ i.i.d.\ draws from the distribution of $(|Z_1|,Z_2)$, where $Z_1,Z_2$ are i.i.d.\ from $N(0,1)$.
        {\bf Right} (classical signs and ranks): $n=3$ data points (marked by ``{\color{blue} $\bullet$}") are drawn from $N(0,4)$. The ranks (marked by ``{\color{orange}+}") are $h_j = \frac{j}{n}$, $j=1,\ldots,n$.
        The squared length of each solid line is the cost of transporting each data point. The dashed (dotted) lines indicate that $\mathbf{X}_i \mapsto -\mathbf{X}_i$ (i.e., $S_n(\mathbf{X}_i) = -I$) before being transported to the $\mathbf{h}_j$'s.}
        \label{fig:ot_vis}
\end{figure}

Figure~\ref{fig:ot_vis} illustrates this for $p=1$ and $p=2$ with $\G = \{-I,I\}$. The left panel of Figure~\ref{fig:ot_vis} shows an example consisting of $n=7$ data points $\mathbf{X}_1,\ldots,\mathbf{X}_n$ (marked by ``{\color{blue} $\bullet$}") from $N(\mathbf{0}_2,4I_2)$ and their generalized ranks when $\G=\{I,-I\}$ (corresponding to central symmetry).
The rank vectors $\mathbf{h}_1,\ldots,\mathbf{h}_n$ (marked by ``{\color{orange}+}") are taken as $n$ i.i.d.\ draws from the distribution of $(|Z_1|,Z_2)$, where $Z_1,Z_2$ are i.i.d.\ from $ N(0,1)$.
The squared length of each solid line in Figure~\ref{fig:ot_vis} is the cost of transporting each data point, i.e., $c(\mathbf{X}_i,R_n(\mathbf{X}_i))  \equiv \|\mathbf{X}_i - S_n(\mathbf{X}_i) R_n(\mathbf{X}_i)\|^2$, for $i=1,\ldots, n$.
For some data points $\mathbf{X}_i$, $\mathbf{X}_i$ is first transported to $-\mathbf{X}_i$ without incurring any cost if $\min_{Q\in \{-I,I\}}\|Q^\top \mathbf{X}_i - R_n(\mathbf{X}_i)\|^2=\|-\mathbf{X}_i-R_n(\mathbf{X}_i)\|^2$ (as shown by the dashed line). 
In such a case, $S_n(\mathbf{X}_i)$, the sign of $\mathbf{X}_i$, is $-I$.
If directly mapping $\mathbf{X}_i$ to $R_n(\mathbf{X}_i)$ has a smaller squared distance compared to that between $-\mathbf{X}_i$ and $R_n(\mathbf{X}_i)$, then $S_n(\mathbf{X}_i)=I$.
When $p=1$ (see the right plot in Figure~\ref{fig:ot_vis}), similar observations can be made. It is just that in this case optimization problem~\eqref{eq:OTproblem} can be explicitly solved (as shown in~\cref{rem:p=1}).

The following proposition (proved in~\cref{pf:uniq_sample}) shows that under mild assumptions, the generalized ranks and signed-ranks are a.s.\ unique. Although the generalized signs may not be unique, they are a.s.\ unique when the group action of $\G$ is free.

\begin{prop}[Uniqueness of sample ranks, signs, and signed-ranks]\label{prop:sample_uniq}
Assume that ${\bf X}  \sim \mu \in \mathcal{P}_{\rm ac}(\mathbb{R}^p)$, and suppose that no two $\mathbf{h}_j$'s lie on a same orbit of $\G$.
\begin{enumerate}
    \item Then the multivariate rank $R_n(\cdot)$ and signed-rank $S_n(\cdot)R_n(\cdot)$ are a.s.\ unique.
    \item If $\mathbf{h}_1,\ldots,\mathbf{h}_n\in B\subset\R^p$, and $\G$ acts freely on $\G B$ (see \eqref{eq:defGB}), then
    $S_n(\cdot)$ is a.s.\ unique.
\end{enumerate}
\end{prop}
A free group action (see Definition~\ref{defn:Free-Group-Action}) is easily available for central symmetry and sign symmetry by taking, for example, $ B=(0,\infty)\times\R^{p-1}$ and
$ B=(0,\infty)^{p}$ respectively. Thus, for these groups the generalized sign $S_n(\cdot)$ is unique; see \cref{fig:signed_rank_cen} (central symmetry) and \cref{fig:signed_rank_sign} (sign symmetry).

However, for spherical symmetry, we do not have a free group action when $p>1$:
In \cref{fig:signed_rank_sph} (spherical symmetry), $\G = {\rm O}(2)$ and the generalized sign is not unique --- rotation by 90 degrees and reflection along the line $y=x$ both map the rank vector to the signed-rank. However, the signed-rank is unique --- it is the point in $\{\mathbf{x}\in\R^2:\|\mathbf{x}\|=2\}$ (as the rank vector $\mathbf{h}$, indicated by  ``{\color{orange}+}" in \cref{fig:signed_rank_sph}, satisfies $\|\mathbf{h}\|=2$) that is closest to the data point.

\begin{remark}[When the minimizer of~\eqref{eq:OTproblem} is not unique] 
In cases where the minimizer of~\eqref{eq:OTproblem} is not unique, perform a random draw from all possible minimizers as follows:
If $\sigma_1,\ldots,\sigma_k\in \mathcal{S}_n$ all attain the minimal cost
$\sum_{i=1}^n c(\mathbf{X}_{\sigma(i)},\mathbf{h}_i)$, set $\hat{\sigma}$ to be $\sigma_j$ with probability $\mathbb{P}(\hat{\sigma}=\sigma_j)=\frac{1}{k}$, $j=1,\ldots,k$.
After determining $\hat{\sigma}$, find the minimizer $\hat{Q}_i:=\argmin_{Q\in\G} \|Q_i^\top \mathbf{X}_{\hat{\sigma}(i)}-\mathbf{h}_i\|^2$. If this minimizer is not unique, introduce another random variable/vector (independent of everything else) to choose $\hat{Q}_i$ from the uniform distribution over $\G$ conditioning on that it minimizes $\lVert Q_i^\top \mathbf{X}_{\hat{\sigma}(i)}-\mathbf{h}_i\rVert^2$; see \cref{subsec:consist_efficiency} on how to sample from this conditional distribution.
\end{remark}

\begin{figure}
     \centering
     \hspace{-0.32in}
     \begin{subfigure}[b]{0.35\textwidth}
         \centering
         \includegraphics[width=\textwidth]{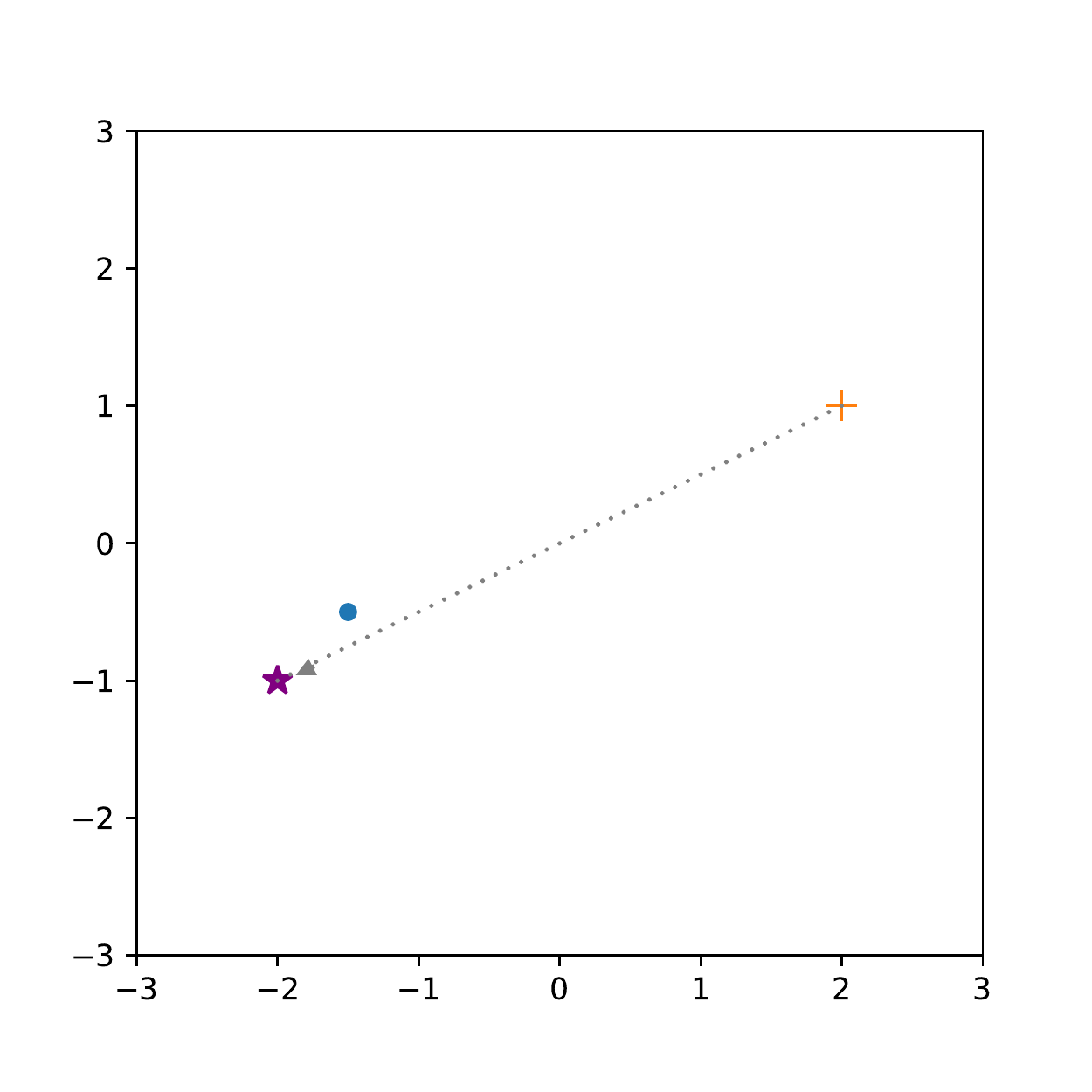}
         \captionsetup{justification=centering}
         \caption{$\G$: central symmetry.}
         \label{fig:signed_rank_cen}
     \end{subfigure}
     \hspace{-0.32in}
     \begin{subfigure}[b]{0.35\textwidth}
         \centering
         \includegraphics[width=\textwidth]{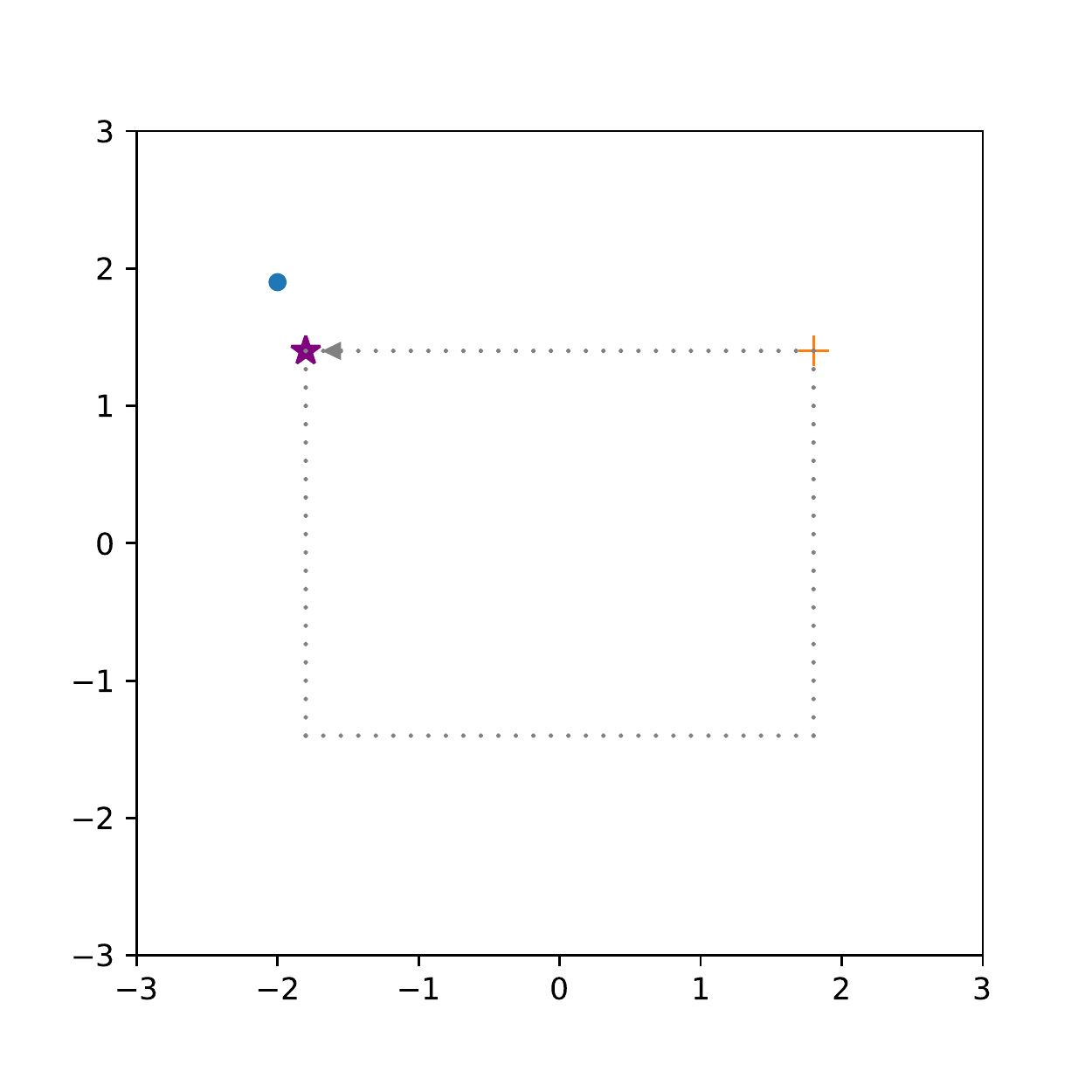}
         \captionsetup{justification=centering}
         \caption{$\G$: sign symmetry.}
         \label{fig:signed_rank_sign}
     \end{subfigure}
     \hspace{-0.32in}
     \begin{subfigure}[b]{0.35\textwidth}
         \centering
         \includegraphics[width=\textwidth]{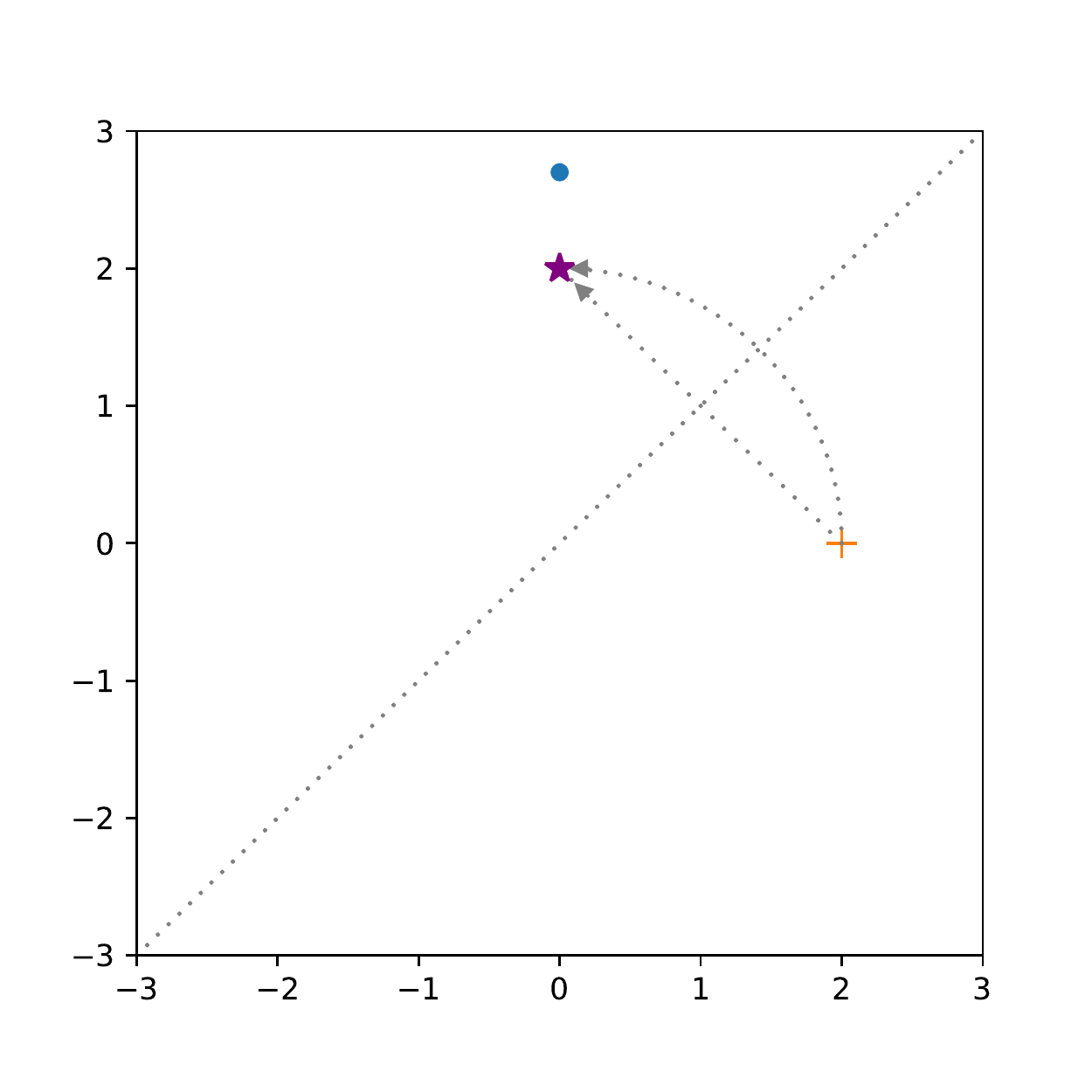}
         \captionsetup{justification=centering}
         \caption{$\G$: spherical symmetry.}
         \label{fig:signed_rank_sph}
     \end{subfigure}
        \caption{\small Illustration of a data point (marked by ``{\color{blue} $\bullet$}"), its rank vector (marked by ``{\color{orange}+}"), and its signed-rank (marked by ``{\color{purple}$\star$}"). The signed-rank is the unique point in the orbit of the rank vector that is closest to the data point.
        For (c), where $\G$ corresponds to spherical symmetry, the sign is not unique as rotation by 90 degrees and reflection along the line $y=x$ both map the rank vector to the signed-rank.}
        \label{fig:signed_rank}
\end{figure}


The following proposition (see \cref{sec:pf_properties_of_generalized_sign_rank} for a proof) shows that the generalized signs and signed-ranks, possibly chosen via randomization as described in Remark~\ref{eq:OTproblem} if they are not unique, extend naturally the properties of their classical counterparts (when $p=1$). In particular, under the null hypothesis of $\G$-symmetry, the generalized signs and generalized ranks are independent, and any statistic based on them is exactly distribution-free in finite samples.
\begin{prop}\label{prop:properties_of_generalized_sign_rank}
Consider the generalized signs, ranks and signed-ranks as defined in~\eqref{eq:Sign-Rank}. If they are not unique (cf.~Proposition~\ref{prop:sample_uniq}), we assume that they are chosen via randomization as described in Remark~\ref{eq:OTproblem}. Then, the following properties hold:
\begin{enumerate}
\item The generalized rank $(R_n(\mathbf{X}_1),\ldots,R_n(\mathbf{X}_n))$ is independent of any order statistics\footnote{An {\it order statistic} is an un-ordered version of the data $\{\mathbf{X}_1,\ldots,\mathbf{X}_n\}$ \citep{hallin2017distribution}.} of $\{\mathbf{X}_1,\ldots,\mathbf{X}_n\}$, and is uniformly distributed over the set of all $n!$ permutations of $(\mathbf{h}_1,\ldots,\mathbf{h}_n)$.
\item $(R_n(\mathbf{X}_1),\ldots,R_n(\mathbf{X}_n))$ and the generalized signs $(S_n(\mathbf{X}_1),\ldots,S_n(\mathbf{X}_n))$ are independent under ${\rm H}_0:\mathbf{X}\overset{d}{=}Q\mathbf{X}$, for all $Q\, \in \G$.
\item $S_n(\mathbf{X}_1),\ldots,S_n(\mathbf{X}_n)$ are i.i.d.\ following the uniform distribution over $\G$, under ${\rm H}_0$.
\end{enumerate}
\end{prop}

\begin{remark}[Unsigned multivariate rank]
If $\G=\{I\}$, then $S_n(\mathbf{X}_i) = I$ and $R_n(\mathbf{X}_i)$ reduces to the (unsigned) multivariate rank defined in \citet{Nabarun2021rank},
and includes the center-outward ranks proposed in \citet{hallin2017distribution}.
\end{remark}

\begin{remark}[Component sign]\label{rk:comp_sign}
For sign-symmetry, when $\nu_n$ is concentrated on $(0,\infty)^p$, $S_n(\mathbf{X}_i)$ is a diagonal matrix with diagonal entries being exactly equal to the component signs of $\mathbf{X}_i$ \citep{Bennett1962sign,Bickel1965competitors,Chatterjee1966bivariate,Puri-Sen-1971}.
\end{remark}

\subsection{Computational Complexity}\label{sec:comp_efficiency}
The optimization problem \eqref{eq:OTproblem} can be solved efficiently.
Recall that $c(\mathbf{x},\mathbf{h})=\min_{Q\in \G}\|Q^\top \mathbf{x}-\mathbf{h}\|^2$ is the cost function (see \eqref{eq:cost}), which can be computed in $O(1)$ time if $\G$ is finite. Then
the ranks $R_n(\mathbf{X}_i)$ can be found by solving the assignment problem \citep{Munkres1957assign,Bertsekas1988auction} of $\mathbf{X}_1,\ldots,\mathbf{X}_n$ to $\mathbf{h}_1,\ldots,\mathbf{h}_n$ under the cost $c(\cdot,\cdot)$,
for which algorithms with worst case complexity $O(n^3)$ are available \citep{Jonker1987algorithm}.
After obtaining the ranks, $S_n(\mathbf{X}_i) = \argmin_{Q\in \G}\|Q^\top \mathbf{X}_i-R_n(\mathbf{X}_i)\|^2$ can be obtained (see the discussion in~\cref{subsec:consist_efficiency} and~\cref{rem:Sp-Sym} on how to obtain this minimizer for the three main examples of symmetry considered in this paper).

\begin{remark}[Fast computation for spherical symmetry]
When $\G={\rm O}(p)$, the computation of the ranks (and the signed-ranks) is even faster (i.e., in $O(n\log n)$ time) since $c(\mathbf{x},\mathbf{h})=\min_{Q\in \G}\|Q^\top \mathbf{x}-\mathbf{h}\|^2 =\min_{Q\in \G}\left\{ \|\mathbf{x}\|^2 - 2 \mathbf{x}^\top Q\mathbf{h} + \|\mathbf{h}\|^2 \right\} =   \big(\|\mathbf{x}\| - \|\mathbf{h}\|\big)^2$, where the last step follows from the fact that there exists $Q\in \G$ that maps $\mathbf{h}$ to the direction of $\mathbf{x}$, e.g., a Householder reflection \citep{Householder1958}. Thus, if $\mathbf{X}_i$ has the $j$-th largest Euclidean norm among $\mathbf{X}_1,\ldots,\mathbf{X}_n$ and $\|\mathbf{h}_1\| < \ldots<\|\mathbf{h}_n\|$, then $\mathbf{X}_i$ will have $\mathbf{h}_j$ as its rank.
The signed-rank of $\mathbf{X}_i$ is simply the vector in the direction of $\mathbf{X}_i$ with length $\|R_n(\mathbf{X}_i)\|$, i.e., $S_n(\mathbf{X}_i) R_n(\mathbf{X}_i) = \|R_n(\mathbf{X}_i)\| \frac{\mathbf{X}_i}{\|\mathbf{X}_i\|}$.

\end{remark}

\subsection{The Population Analogues}\label{sec:popu_signs_ranks}
So far, we have defined the sample version of the multivariate generalized signs and ranks, $S_n(\cdot)$ and $R_n(\cdot)$, via the optimization problem \eqref{eq:OTproblem}.
We will now show that, under finite second order moment assumptions on $\mu$ and $\nu$, the sample rank map $R_n(\cdot)$ converges to a population rank map $R(\cdot)$ such that $(\mathbf{X},R(\mathbf{X}))$ solves the following OT problem (aka Kantorovich relaxation \cite{villani2009OT}):
\begin{equation}\label{eq:popu_ot_problem}
\inf_{(\mathbf{X},\mathbf{H}):(\mathbf{X},\mathbf{H})\in \Pi(\mu,\nu)} \mathbb{E} \left[c(\mathbf{X},\mathbf{H})\right],\quad {\rm with}\ c(\mathbf{x},\mathbf{h})=\min_{Q\in \G}\|Q^\top \mathbf{x}-\mathbf{h}\|^2.
\end{equation}
Here, $\Pi(\mu,\nu)$ consists of all joint Borel distributions $(\mathbf{X},\mathbf{H})$ with marginals $\mathbf{X}\sim \mu$ and $\mathbf{H}\sim \nu$. It can be seen that the optimization problem \eqref{eq:popu_ot_problem} is a natural population analogue of \eqref{eq:OTproblem}.

To establish the convergence of $S_n(\cdot)$ and $R_n(\cdot)$, we need the reference distribution $\nu$ and the group $\G$ to satisfy certain natural compatibility assumptions. In particular, we assume that: 

\begin{assumption}\label{assump:nu-G}
There exists a Borel set $B\subset\R^p$ with $\nu(B)= 1$ such that for any $\mathbf{x}\in\mathbb{R}^p$, the orbit $\{Q\mathbf{x}: Q\in \G\}$ of $\mathbf{x}$ intersects $B$ at one point at most.
\end{assumption}
This assumption requires that only one point at most should be taken as a rank vector (which belongs to $B$) from any orbit of $\G$.
This is because transporting $\mathbf{x}$ to any point in an orbit of $\G$ has the same cost due to the form of the cost function $c(\cdot,\cdot)$ (see \eqref{eq:cost}),
and there is no need to include two representative points from an orbit, which may break the uniqueness of the population OT map.\footnote{Let $\mu = \nu = N(\mathbf{0},I)$. Then $R(\mathbf{x})= \mathbf{x}$ and $R(\mathbf{x})=- \mathbf{x}$ are both OT maps with 0 loss (see \eqref{eq:popu_ot_problem}), provided $-I\in \G$.}

To formally state our result, we introduce some notation.

\begin{definition}[Quotient map $q$]
Under \cref{assump:nu-G} we define the quotient map $q:\G B\to B$ by $q(Q\mathbf{x})=\mathbf{x}$ for $\mathbf{x}\in B$ and $Q\in \G$.
Observe that \cref{assump:nu-G} ensures that $q$ is well-defined.
\end{definition}
Note that the above definition is slightly different from the usual way a quotient map is defined, e.g., as a map that takes a point to an equivalence class (i.e., an orbit here); see e.g., \citet[Chapter 2.22]{Munkres2000topology}. Instead, our quotient map $q$ explicitly maps all points in the orbit $\{Q\mathbf{x}: Q\in \G\}$ to a representative point $\mathbf{x}$ in $ B$.

\begin{figure}
     \centering
     \begin{subfigure}[b]{0.32\textwidth}
         \centering
         \includegraphics[width=\textwidth]{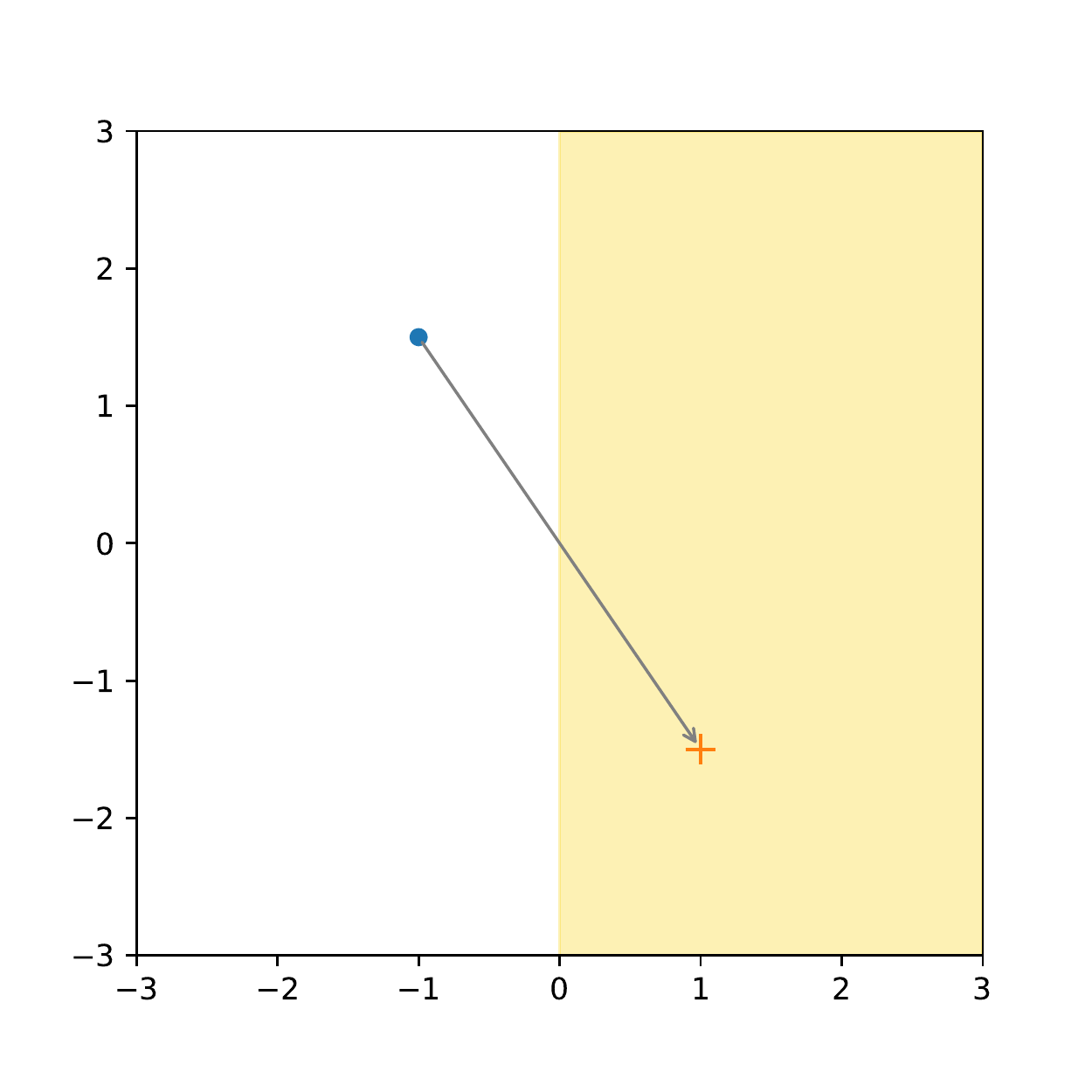}
         \captionsetup{justification=centering}
         \caption{$\G$: central symmetry.}
         \label{fig:q_cen}
     \end{subfigure}
     \begin{subfigure}[b]{0.32\textwidth}
         \centering
         \includegraphics[width=\textwidth]{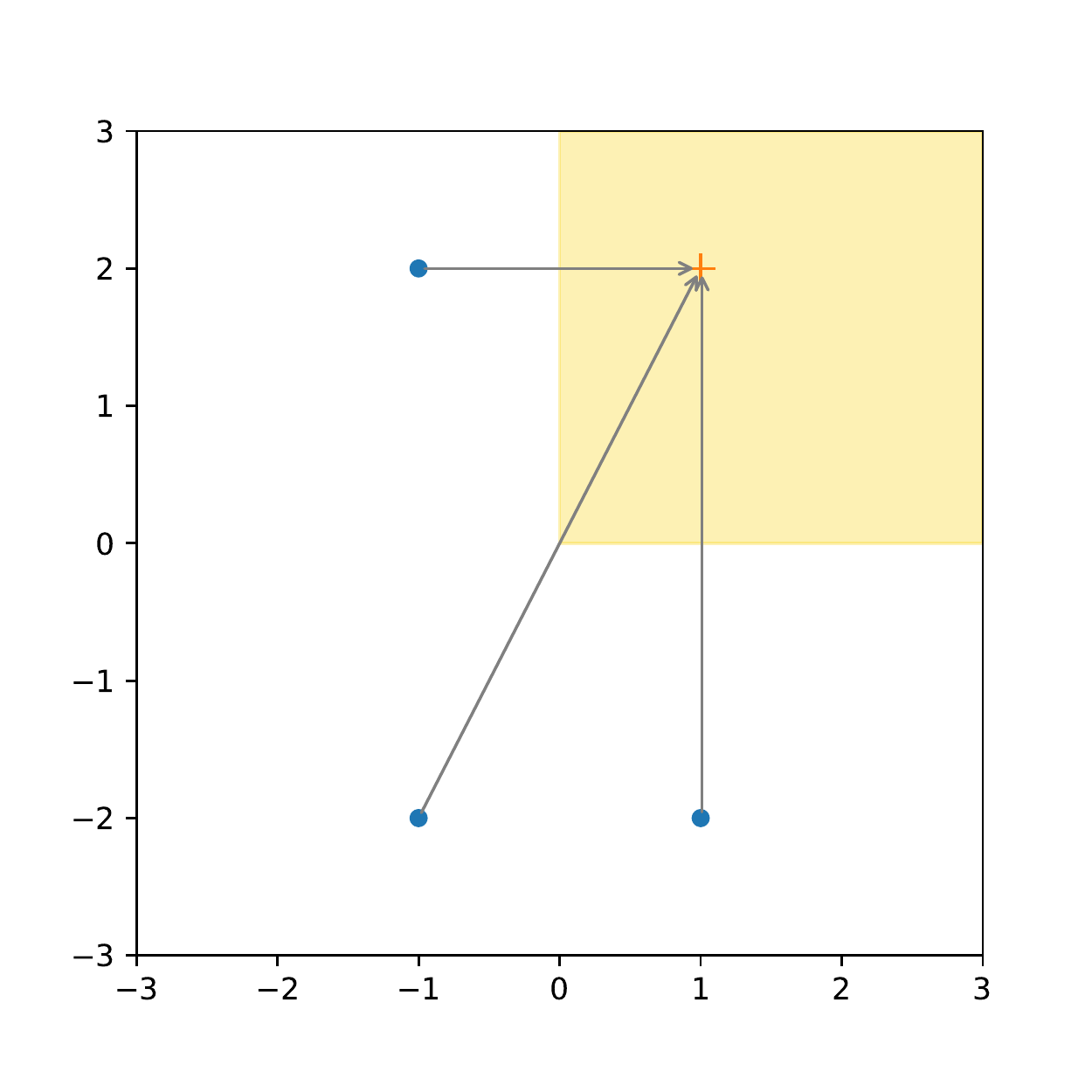}
         \captionsetup{justification=centering}
         \caption{$\G$: sign symmetry.}
         \label{fig:q_sign}
     \end{subfigure}
     \begin{subfigure}[b]{0.32\textwidth}
         \centering
         \includegraphics[width=\textwidth]{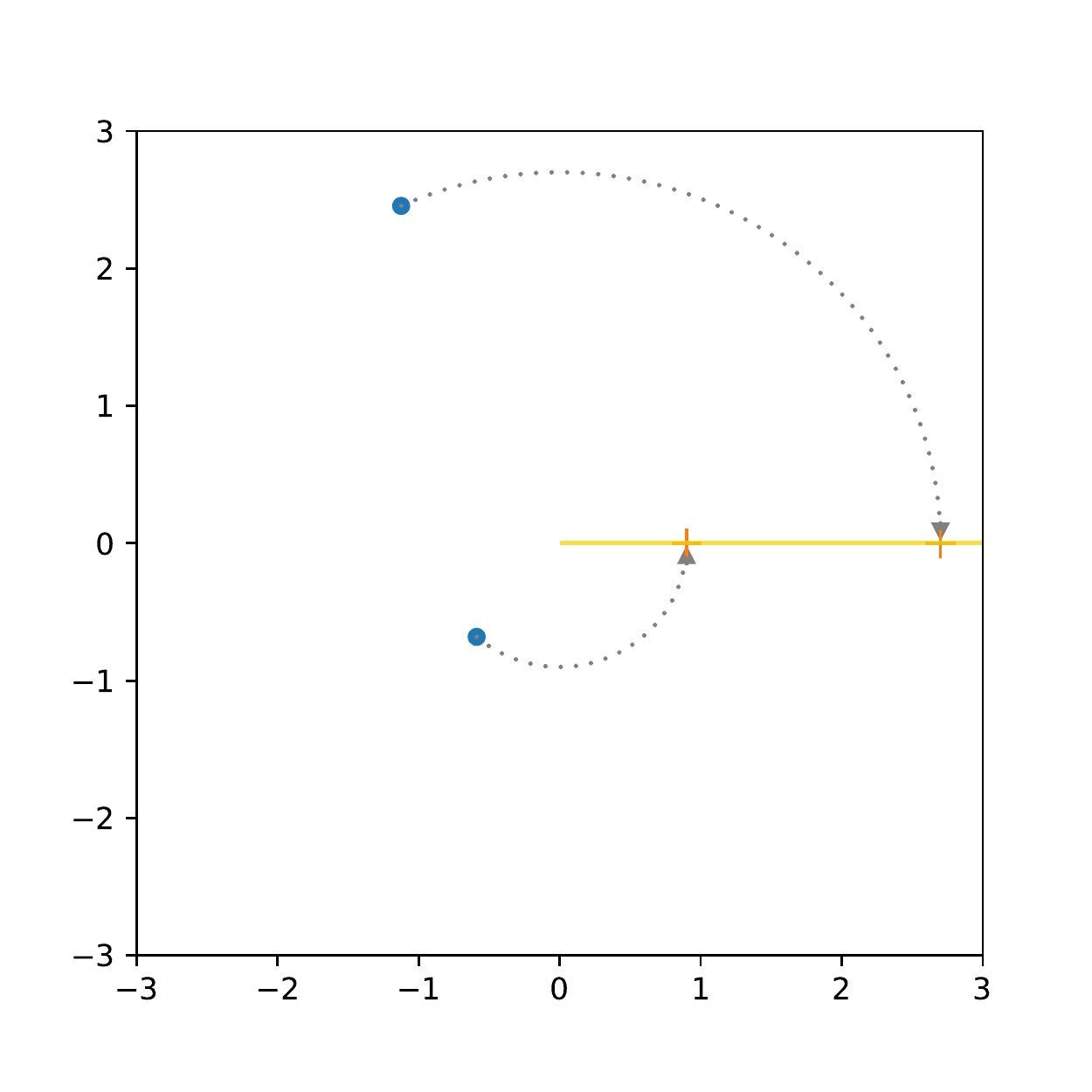}
         \captionsetup{justification=centering}
         \caption{$\G$: spherical symmetry.}
         \label{fig:q_sph}
     \end{subfigure}
        \caption{Illustrations of the quotient map $q(\cdot)$ for different groups $\G$, when $p=2$. Here, each blue point (marked by ``{\color{blue} $\bullet$}") is mapped to a representative point in its orbit (marked by ``{\color{orange}+}") that lies in the set $B$ (shaded in yellow).
        }
        \label{fig:quotient_maps}
\end{figure}

\cref{fig:quotient_maps} shows examples of quotient maps. In \cref{fig:q_cen}, $B=(0,\infty)\times \R$ and $\G = \{-I,I\}$. Here $q(\mathbf{x})=-\mathbf{x}$ if $x_1 < 0$ and $q(\mathbf{x})=\mathbf{x}$ if $x_1 > 0$, where $\mathbf{x}=(x_1,x_2)$.
\cref{fig:q_sign} shows the actions of the quotient map $q$ on the three points $(-1,2)$, $(-1,-2)$ and $(1,-2)$, for $\G$ being the group corresponding to sign symmetry with $ B=(0,\infty)^2$.
All these three points are in the same orbit of $\G$ and are mapped to the same point $(1,2)$.
In this case, $q(\mathbf{x})=(|x_1|,|x_2|)$.
\cref{fig:q_sph} shows the quotient map for $\G={\rm O}(2)$, corresponding to spherical symmetry, with $ B=(0,\infty)\times\{0\}$.
In this case, the quotient map $q$ can be written as $q(\mathbf{x})=(\|\mathbf{x}\|,0)$.


The following result (proved in \cref{sec:pf_consistency}) states the existence and uniqueness of the population rank map $R:\R^p\to\R^p$ such that $({\rm identity}, R)\#\mu$ (i.e., the distribution of $(\mathbf{X},R(\mathbf{X}))$ where $\mathbf{X} \sim \mu$) solves \eqref{eq:popu_ot_problem}, if the optimal cost is finite. Even if the optimal cost is infinite, $R(\cdot)$ can be uniquely characterized using the notion of $c$-cyclical monotonicity (see \cref{def:c-cyc}).

\begin{theorem}[Population rank and signed-rank maps]\label{thm:population_rank}
Suppose \cref{assump:nu-G} holds. Denote the distribution of $\mathbf{X}$ by $\mu\in \mathcal{P}_{\rm ac}(\mathbb{R}^p)$. 
Then there exists a $\mu$-a.e.\ unique Borel measurable map $R:\R^p\to\R^p$ such that $(\mathbf{X},R(\mathbf{X}))$ has the unique distribution in $\Pi(\mu,\nu)$ with a $c$-cyclically monotone support. Moreover, $R$ has the following properties:
\begin{enumerate}
\item[(i)]
Let $q$ be the quotient map from $\G B$ to $ B$. Then, there exists a lower semicontinuous convex function $\psi:\mathbb{R}^p\to (-\infty,+\infty]$ such that
$$R(\mathbf{x}) = q(\nabla \psi(\mathbf{x})) \quad (\mu\mbox{-a.e.}~\mathbf{x}).$$
\item[(ii)] Let $S\sim {\rm Uniform}(\G)$ be independent of $\mathbf{X}\sim \mu$ and $\mathbf{H}\sim\nu$. Let $\mu_S$ and $\nu_S$ denote the distributions of $S\mathbf{X}$ and $S\mathbf{H}$, respectively. Then
$\nabla \psi(\cdot)$ is the $\mu_S$-a.e.\ unique gradient of convex function that pushes $\mu_S$ to $\nu_S$.
\item[(iii)] $\nabla \psi$ is equivariant under the group action of $\G$, i.e.,
$\nabla \psi(Q\mathbf{x}) = Q\nabla \psi(\mathbf{x})$ for $Q\in \G$ and $\mathbf{x} \in \R^p$ being a differentiable point of $\psi$.
\item[(iv)] $\nabla \psi(\mathbf{X}) \overset{a.s.}{=} S(\mathbf{X},R( \mathbf{X})) R(\mathbf{X})$ with $S(\mathbf{x},\mathbf{h}) := \argmin_{Q\in \G} \|Q^\top\mathbf{x} - \mathbf{h}\|^2$. Hence, $\nabla \psi(\cdot)$ can be viewed as the population signed-rank map.
\item[(v)] If $\mu$ and $\nu$ have finite second moments, then
$({\rm identity}, R)\#\mu$ is the unique solution to the population OT problem \eqref{eq:popu_ot_problem}.
\end{enumerate}
\end{theorem}

Let us discuss the conclusions in $(i)$-$(v)$ of the above result: 
$(i)$ shows that the population rank map is essentially the gradient of a convex function $\psi(\cdot)$ (cf.~McCann's geometric characterization of OT maps~\citep{McCann1995} under the usual squared loss), up to a composition with the quotient map $q(\cdot)$ which brings it to $B$. The convex function $\psi(\cdot)$ is characterized in $(ii)$ --- its gradient pushes the symmetrized data distribution $\mu_S$ to the symmetrized reference measure $\nu_S$, and it is unique  (which follows from the celebrated result in \citet{McCann1995}). Although the result in~\cite{McCann1995} would imply that $\nabla\psi(\cdot)$ is $\mu_S$-a.e.\ uniquely defined,
\cref{lem:abs_cont} in the Appendix shows that $\mu$ is absolutely continuous w.r.t.\ $\mu_S$, and thus
$\nabla\psi(\cdot)$ is also $\mu$-a.e.\ uniquely defined. In 
$(iii)$ and $(iv)$ of Theorem~\ref{thm:population_rank} we illustrate some of the  `nice' properties of $\nabla\psi(\cdot)$: it respects the symmetry of $\G$, and turns out to be the population signed-rank map, which further motivates the study of the generalized (sample) signed-ranks $S_n(\cdot)R_n(\cdot)$.
Note that $(i)$-$(iv)$ in the above result do not assume any moment conditions on $\mu$ and $\nu$, and provides a geometric characterization of the rank map using $c$-cyclical monotonicity.
When second moments of $\mu$ and $\nu$ do exist, $(v)$ shows that the population rank map is indeed the OT map in the sense of minimizing~\eqref{eq:popu_ot_problem}.

\begin{remark}[On the proof of Theorem~\ref{thm:population_rank}]
The proof of the above result proceeds via checking the uniqueness of the $c$-subdifferential set \citep[Theorem 5.30]{villani2009OT} to show that the $c$-cyclically monotone measure in $\Pi(\mu,\nu)$ is unique, and is given by a Monge map. It crucially uses the structure of the cost function $c(\cdot,\cdot)$, e.g., a $c$-convex function (defined in \cref{sec:pf_consistency}) is still convex in the usual sense, with its (almost everywhere) gradient being $\G$-equivariant.
\end{remark}

Note that the minimizer $S(\mathbf{x},\mathbf{h}) = \argmin_{Q\in \G} \|Q^\top\mathbf{x} - \mathbf{h}\|^2$ may not be unique in general. For example, when $\G = {\rm O}(p)$, there are many orthogonal matrices that map a given vector to a given direction (see e.g., \cref{fig:signed_rank_sph}). However, in the proof of \cref{thm:population_rank}  we show that the signed-rank
$S(\mathbf{X},R( \mathbf{X})) R(\mathbf{X})$ is $\mu$-a.s.~unique --- it is the point in the orbit of $R( \mathbf{X})$, i.e.,
$\{QR( \mathbf{X}): Q\in \G\}$, that is closest to $\mathbf{X}$.

With the population rank map $R(\cdot)$ and the signed-rank map $\nabla\psi(\cdot)$ defined in \cref{thm:population_rank}, we state below a convergence theorem for the sample generalized signs, ranks and signed-ranks. We assume the following condition:
\begin{assumption}[On weak convergence of $\nu_n$]\label{assump:Weak-nu_n}
The empirical measure $\nu_n := \frac{1}{n} \sum_{i=1}^n \delta_{\mathbf{h}_i}$ converges weakly to the reference distribution $\nu$.
\end{assumption}
The result below states that these sample quantities converge to their population counterparts in the averaged $\rho$-loss. Here, we let $\rho(\cdot,\cdot)$ be any continuous nonnegative function on $\R^p \times \R^p$ satisfying:
(i) $\rho (\mathbf{x},\mathbf{x}) = 0$, $\forall \; \mathbf{x}\,\in\,\R^p$;  (ii) there exists $C_1 > 0$ such that $\rho(\mathbf{x},\mathbf{y})\leq C_1\left(\rho(\mathbf{x},\mathbf{0})+\rho(\mathbf{y},\mathbf{0})\right)$, $\forall \; \mathbf{x}$, $\mathbf{y}\in\R^p$; (iii) there exists $C_2 >0$ such that $\rho(Q\mathbf{x},\mathbf{0})\leq C_2 \left( \rho(\mathbf{x},\mathbf{0})+1\right)$,  $\forall \; \mathbf{x}\,\in\,\R^p$, $Q\in\G$.
Examples of $\rho(\cdot,\cdot)$ include the usual Euclidean distance $\|\cdot - \cdot\|$, any norm on $\R^p$, and $\rho(\mathbf{x},\mathbf{y})=\|\mathbf{x} - \mathbf{y}\|^k$ for any $k > 0 $.

\begin{theorem}[Convergence to population generalized ranks, signs and signed-ranks]\label{cor:conv_popu}
Suppose Assumptions~\ref{assump:nu-G} and \ref{assump:Weak-nu_n} hold. We have:
\begin{enumerate}
\item (Convergence of ranks) If the reference distribution $\nu$, which concentrates on $B$, is chosen such that the quotient map $q:\G  B\to B$ defined as  $q(Q\mathbf{x})=\mathbf{x}$ (for $\mathbf{x}\in B$, $Q\in \G$) is continuous w.r.t.\ the usual Euclidean topology, then
$R(\cdot) = q(\nabla \psi(\cdot))$ is $\mu$-a.e.\ continuous. If furthermore, $\E \rho(\mathbf{H}_{n},\mathbf{0}) \to \E \rho(\mathbf{H},\mathbf{0})<\infty$ as $n\to\infty$, where $\mathbf{H}_{n}\sim \nu_n$ and $\mathbf{H}\sim \nu$, then
$$\frac{1}{n}\sum_{i=1}^n \rho(R_n(\mathbf{X}_i),R(\mathbf{X}_i) )\overset{a.s.}{\longrightarrow} 0,\quad {\rm as}\ n\to\infty.$$
\item (Convergence of signs) If $\G$ acts freely on $\G B$,
then $S(\mathbf{x},R( \mathbf{x})) := \argmin_{Q\in \G} \|Q^\top\mathbf{x} - R(\mathbf{x})\|^2$ is uniquely defined for $\mu$-a.e.\ $\mathbf{x}$, and for any continuous function $\tilde{\rho}(\cdot,\cdot)$ on $\R^{p^2}\times \R^{p^2}$ such that $\tilde{\rho}(Q,Q)=0$, for all $Q\,\in\R^{p^2}$, we have
$$\frac{1}{n}\sum_{i=1}^n \tilde{\rho}(S_n(\mathbf{X}_i) , S(\mathbf{X}_i,R(\mathbf{X}_i)))\overset{a.s.}{\longrightarrow} 0,\quad {\rm as}\ n\to\infty.$$
\item (Convergence of signed-ranks) 
If $\E \rho(\mathbf{H}_{n},\mathbf{0}) \to \E \rho(\mathbf{H},\mathbf{0})<\infty$ where $\mathbf{H}_{n}\sim \nu_n$, $\mathbf{H}\sim \nu$, then
$$\frac{1}{n}\sum_{i=1}^n \rho(S_n(\mathbf{X}_i)R_n(\mathbf{X}_i) , \nabla\psi(\mathbf{X}_i))\overset{a.s.}{\longrightarrow} 0,\quad {\rm as}\ n\to\infty.$$
\end{enumerate}
\end{theorem}
A proof of the above result is given in \cref{sec:pf_consistency}.
To establish the convergence of ranks, the continuity of $q$ is needed, 
while to establish the convergence of $S_n(\mathbf{X}_i)$, we need to restrict ourselves to $\G$ with a free group action  (which includes central symmetry and sign symmetry) to make sure that $S_n(\mathbf{X}_i)$ and $S(\mathbf{X}_i,R(\mathbf{X}_i))$ are a.s.\ uniquely defined.
The convergence of the signed-ranks requires the weakest assumption, again motivating that test statistics based on the signed-ranks may be preferable. This intuition will be confirmed in \cref{sec:ARE}. Note that the above result allows for the use of any `reasonable' loss function while comparing the sample and population signs, ranks and signed-ranks.

\section{Distribution-Free Tests for Symmetry}\label{sec:ARE}
In this section, we will define multivariate generalizations of the sign test and the Wilcoxon signed-rank test that share analogous properties with their classical counterparts (Sections~\ref{subsec:sign}-\ref{subsec:rank}).
We then study the consistency of our proposed tests (\cref{subsec:consist_efficiency}) and state results on their relative efficiency w.r.t.\ Hotelling's $T^2$ test (\cref{subsec:efficiency}). We describe the locally asymptotically optimal property of our generalized Wilcoxon signed-rank (GWSR) test in~\cref{subsec:efficiency}.

\subsection{Generalized Sign Test}\label{subsec:sign}
Recall the hypothesis of $\G$-symmetry in \eqref{eq:hypo}, and that
under ${\rm H}_0$, the generalized signs $S_n(\mathbf{X}_1),\ldots,S_n(\mathbf{X}_n)$ are i.i.d.\ from ${\rm Uniform}(\G)$.
Hence, any test of uniformity over $\G$ can be applied with the generalized signs to define a multivariate analogue of the sign  test.
\begin{example}[Chi-squared test]
If $\G = \{g_1,\ldots,g_k\}$ is a finite group of size $k$, then we can count the observed frequencies of each element in $\G$, i.e., $Y_i:=\sum_{j=1}^n I(S_n(\mathbf{X}_j)=g_i)$, $i=1,\ldots,k$.
Under ${\rm H}_0$, $(Y_1,\ldots,Y_k)$ follows
${\rm Multinomial}\left(n,\frac{1}{k}\mathbf{1}_k\right)$. Further, we can perform an asymptotic Pearson's chi-squared test:
$$\tilde{T}_n := \sum_{i=1}^k \frac{(Y_i - n/k)^2}{n/k}\overset{d}{\longrightarrow}\chi^2_{k-1},\quad {\rm as\ }n\to\infty.$$
${\rm H}_0 $ can be rejected for a large value of $\tilde{T}_n$.
\end{example}
However, for a group with large size $k$ (e.g., $\G$ corresponding to sign symmetry has size $2^p$, which can be very large for even moderate $p$), the convergence of the chi-squared test statistic $\tilde{T}_n$ to its limiting distribution can be slow.
Hence, in this paper 
we propose 
the following {\it generalized sign test}: 
\begin{equation}\label{eq:Sign-Test}
T_n:=\frac{1}{\sqrt{n}}\sum_{i=1}^n S_n(\mathbf{X}_i).
\end{equation}
We reject the null hypothesis ${\rm H}_0$ for large values of $\|T_n\|^2_F$ where $\|\cdot \|_F$ denotes the matrix Frobenius norm. The generalized sign test is distribution-free under ${\rm H}_0$ as the signs $S_n(\mathbf{X}_i)$ are i.i.d.\ from ${\rm Uniform}(\G)$. We give below a simple result on the asymptotic behavior of $\|T_n\|^2_F$ under ${\rm H}_0$. Note that this test statistic is also a direct generalization of the classical sign test.



\begin{prop}[Generalized sign test]\label{prop:Sign-Test}
We have the following consequences:
\begin{enumerate}
\item Central symmetry: $S_n(\mathbf{X}_i)$ are i.i.d.\ uniform over $\G= \{I, -I\}$ under ${\rm H}_0$, and $T_n$ (in~\eqref{eq:Sign-Test}) is thus equivalent to a binomial random variable (up to a linear transformation) in finite samples. The asymptotic distribution of $T_n$ follows from the usual one-dimensional CLT:
$$\frac{1}{p}\|T_n\|^2_F \overset{d}{\longrightarrow} \chi^2_1,\quad {\rm as}\ n\to\infty.$$
\item Sign symmetry: $S_n(\mathbf{X}_i)$ are i.i.d.\ uniform over diagonal orthogonal matrices under the null. Equivalently, $S_n(\mathbf{X}_i)$ can be viewed as being i.i.d.\ uniform over $\{1, -1\}^p$. In such a case,
$$\|T_n\|^2_F \overset{d}{\longrightarrow} \chi^2_p,\quad {\rm as}\ n\to\infty.$$
\item Spherical symmetry: $S_n(\mathbf{X}_i)$ are i.i.d.\ uniform over the orthogonal group ${\rm O}(p)$ under ${\rm H}_0$. In this case, our generalized sign test behaves as:
$${p}\|T_n\|^2_F \overset{d}{\longrightarrow} \chi^2_{p^2},\quad {\rm as}\ n\to\infty.$$
\end{enumerate}
\end{prop}
Possibly the simplest test of uniformity over ${\rm O}(p)$ is given by the Rayleigh test \citep[Section 13.2.2]{Mardia2000directional}, which rejects the uniformity assumption when $pn {\rm Tr}\left(\bar{S}^\top \bar{S}\right)$ is large, where 
$\bar{S}:=\frac{1}{n}\sum_{i=1}^n S_n(\mathbf{X}_i)$. Note that this is exactly the generalized sign test given in~\cref{prop:Sign-Test}-3. We also  observe that all the above three cases reduce to the classical two-sided sign test when $p=1$.

\subsection{Generalized Wilcoxon Signed-Rank (GWSR)  Test }\label{subsec:rank}
The GWSR statistic can naturally be defined as
$\frac{1}{\sqrt{n}} \sum_{i=1}^n S_n(\mathbf{X}_i)  R_n(\mathbf{X}_i)$.
We also consider a more general score transformed version: for a score function $J:B \to \R^p$ define
\begin{equation}\label{eq:score_Wn}\mathbf{W}_n:= \frac{1}{\sqrt{n}}\sum_{i=1}^n S_n(\mathbf{X}_i)  J(R_n(\mathbf{X}_i)).
\end{equation}

Note that $\mathbf{W}_n$ is distribution-free under $\textrm{H}_0$ (see~\eqref{eq:hypo}) following~\cref{prop:properties_of_generalized_sign_rank}.
Apart from using its exact distribution in finite samples, one can also use its asymptotic distribution as $n\to\infty$, which yields a computationally simple test.
\begin{definition}[Effective reference distribution]
 We define the {\it effective reference distribution} (ERD), obtained with score function $J(\cdot)$ and reference distribution $\nu$, as the distribution of $SJ(\mathbf{H})$, where $S\sim {\rm Uniform}(\G)$ is independent of $\mathbf{H}\sim\nu$.
\end{definition}
The ERD is $\G$-symmetric, 
i.e., $QSJ(\mathbf{H})$ has the same distribution as $SJ(\mathbf{H})$, for all $Q\in\G$.
Let $\Sigma_{\rm ERD}$ be the covariance matrix of the ERD.
The GWSR test rejects the null hypothesis ${\rm H}_0$ for large values of \begin{equation}\label{eq:Wilcox-Test}
\mathbf{W}_n^\top \Sigma_{\rm ERD}^{-1}\mathbf{W}_n.
\end{equation}
In the next few subsections we study the behavior of the GWSR test under the null and under (local) contiguous alternatives. For this, we need a subset of the following assumptions, which we list below for the convenience of the reader.

We say $J:\G B\to \R^p$ is {\it equivariant} under the group action of $\G$ if $J(Q\mathbf{x})=Q J(\mathbf{x})$ for all $Q\in\G$ and $\mathbf{x}\in \G B$.
If $J:B\to\R^p$ can be extended to an equivariant function on $\G B$,
then $\mathbf{W}_n$ will be a statistic based on the signed-ranks.

\begin{assumption}[On the score function $J$]\label{assump:equivarJ}
$J:\G B\to \R^p$ is $\nu_S$-a.e.\ continuous and equivariant under the group action of $\G$.
\end{assumption}
\begin{assumption}[Moment convergence of $\nu_n$ to $\nu$]\label{assump:nu_n-nu} Let $\mathbf{H}_n\sim\nu_n = \frac{1}{n} \sum_{i=1}^n \delta_{\mathbf{h}_i}$ and  $\mathbf{H}\sim\nu$.
\begin{enumerate}
    \item[(i)] (1st moment) $\E \|J(\mathbf{H}_n)\| \to \E \|J(\mathbf{H})\|<\infty $ as $n \to \infty$.
    \item[(ii)] (2nd moment) $\mathbb{E}\|J(\mathbf{H}_{n})\|^2 \to \mathbb{E}\|J(\mathbf{H})\|^2 <\infty$ as $n \to \infty$.
    \item[(iii)] (2nd moments)  $\mathbb{E}|J(\mathbf{H}_{n})_i J(\mathbf{H}_{n})_j|\to \mathbb{E}|J(\mathbf{H})_{i} J(\mathbf{H})_{j}|<\infty$ as $n \to \infty$, for $i,j\in\{1,\ldots,p\}$.
\end{enumerate}
\end{assumption}

\begin{assumption}[On the group $\G$]\label{assump:G}
\begin{enumerate}
    \item[(i)] $\mathbb{E} [{\rm Uniform}(\G)] = \mathbf{0}_{p\times p}$; \hspace{0.1in}(ii) $-I_p \in \G$.
\end{enumerate}
\end{assumption}
Note that, under the condition that $\nu_n$ converges weakly to $\nu$ (see~\cref{assump:Weak-nu_n}), Assumption~\ref{assump:nu_n-nu}-$(iii)$ implies Assumption~\ref{assump:nu_n-nu}-$(ii)$, and Assumption~\ref{assump:nu_n-nu}-$(ii)$ implies Assumption~\ref{assump:nu_n-nu}-$(i)$.
By the strong law of large numbers, if $J\# \nu$ has finite second moments, and $\nu_n$ is obtained from i.i.d.\ sampling from $\nu$, then $\nu_n$ satisfies Assumption~\ref{assump:nu_n-nu}-$(iii)$.
If $J(\cdot)$ is bounded and continuous, then Assumption~\ref{assump:nu_n-nu}-$(iii)$  also holds by the dominated convergence theorem, as long as $\nu_n$ converges weakly to $\nu$.
Similarly, Assumption~\ref{assump:G}-$(ii)$ implies Assumption~\ref{assump:G}-$(i)$, and central symmetry, sign symmetry, spherical symmetry all satisfy Assumption~\ref{assump:G}-$(ii)$.

The following result (proved in~\cref{sec:pf_Wil_CLT}) shows that $\mathbf{W}_n$ has an asymptotic normal distribution, mirroring the classical Wilcoxon signed-rank statistic.

\begin{theorem}[Asymptotic normality of $\mathbf{W}_n$]\label{thm:CLT_wilcoxon}
Suppose Assumptions~\ref{assump:Weak-nu_n},~\ref{assump:nu_n-nu}-(iii) and \ref{assump:G}-(i) hold,
and that the score function $J(\cdot)$ used in defining $\mathbf{W}_n$ (see \eqref{eq:score_Wn}) is $\nu$-a.e.\ continuous.
Then, under ${\rm H}_0$,
$$\mathbf{W}_n\overset{d}{\longrightarrow}N\left(\mathbf{0}_p,\Sigma_{\rm ERD}\right).$$
\end{theorem}

\subsection{Consistency of the Tests}\label{subsec:consist_efficiency}
If $\mathbf{X}$ is $\G$-symmetric, then by the equivariance of $\psi$ (see \cref{thm:population_rank}), $Q\nabla \psi(\mathbf{X}){=}\nabla\psi(Q\mathbf{X})\overset{d}{=}\nabla\psi(\mathbf{X})$ for all $Q\in\G$. Hence $\nabla\psi(\mathbf{X})$ is also $\G$-symmetric, and $\E[ J(\nabla\psi(\mathbf{X}))] = \mathbf{0}$ if~\cref{assump:G}-$(i)$ holds and $J(\cdot)$ is equivariant.\footnote{For $S \sim {\rm Uniform}(\G)$ independent of $\mathbf{H} \sim \nu$, we have $\E[ J(\nabla\psi(\mathbf{X}))] = \E[ J(S \mathbf{H})] = \E[S J(\mathbf{H})]  = \E[S] \E[ J(\mathbf{H})] = \mathbf{0}$.} On the other hand, if $\mathbf{X}$ is not $\G$-symmetric, then $\E[ J(\nabla\psi(\mathbf{X}))]$ may not be $\mathbf{0}$. The following result (\cref{prop:consistency_signed_rank}; see \cref{sec:pf_consistency} for a proof) states that under this assumption our GWSR test will be consistent.

\begin{theorem}[Consistency of GWSR test]\label{prop:consistency_signed_rank}
Suppose Assumptions~\ref{assump:nu-G},~\ref{assump:Weak-nu_n}, \ref{assump:equivarJ} and \ref{assump:nu_n-nu}-(i) hold.
Then the distribution-free GWSR test (see~\eqref{eq:Wilcox-Test}) is consistent against all alternatives for which
$\mathbb{E} [J(\nabla\psi(\mathbf{X}))]\neq \mathbf{0}$.
\end{theorem}
It is natural to ask if the condition $\mathbb{E} [J(\nabla\psi(\mathbf{X}))]\neq \mathbf{0}$ holds for location shift alternatives. The following result (see \cref{sec:pf_location_shift} for a proof) shows that it is indeed the case, thereby showing that the GWSR test~\eqref{eq:Wilcox-Test} for $\G$-symmetry (if $-I \in \G$) is consistent against location shift alternatives of a centrally symmetric distribution. 

\begin{corollary}[Consistency in location shift model]\label{cor:location_shift}
Suppose Assumptions~\ref{assump:nu-G},~\ref{assump:Weak-nu_n},~\ref{assump:nu_n-nu}-(i) and \ref{assump:G}-(ii) hold.
Suppose $J(\mathbf{x})=\mathbf{x}$, and $\mathbf{X}$ is a centrally symmetric distribution shifted by $\mathbf{\Delta}\neq \mathbf{0}$, i.e.,
$$\mathbf{X} - \mathbf{\Delta} \stackrel{d}{=}  \mathbf{\Delta} - \mathbf{X}.$$
If $\psi$ (defined in \cref{thm:population_rank}) is strictly convex\footnote{$\psi$ is said to be strictly convex on $U$ if for any $\mathbf{x}$, $\mathbf{y}\in U$ and $\alpha\in(0,1)$, $\psi(\alpha\mathbf{x}+(1-\alpha)y)<\alpha\psi(\mathbf{x})+(1-\alpha)\psi(\mathbf{y})$.} on an open set of positive $\mu$-measure,
then the GWSR test is consistent.
\end{corollary}

Let us now discuss the consistency properties of the generalized sign test. Recall that $S_n(\mathbf{X}_i) \equiv S(\mathbf{X}_i, R_n(\mathbf{X}_i))$ where $S(\mathbf{x},\mathbf{h}) = \argmin_{Q\in \G} \|Q^\top\mathbf{x} - \mathbf{h}\|^2$.
When we have a free group action, $S(\mathbf{x},\mathbf{h})$ will be $({\rm identity}, R)\#\mu$-a.e.~unique (see~\cref{cor:conv_popu}). For example:
\begin{enumerate}
\item (Central symmetry) $S(\mathbf{x},\mathbf{h}) = {\rm sign}(\mathbf{x}^\top \mathbf{h} ) I_p$ is unique provided that $\mathbf{x}^\top \mathbf{h}\neq 0$.
\item (Sign symmetry) $S(\mathbf{x},\mathbf{h}) = {\rm diag}\left(\textbf{sign}(\mathbf{x}\odot\mathbf{h} )\right)$ is unique provided $\mathbf{x}\odot\mathbf{h}\neq \mathbf{0}$, where
for $\mathbf{a},\mathbf{b}\in \R^p$, $\mathbf{a}\odot \mathbf{b}$ denotes $(a_1b_1, \ldots ,a_pb_p)$, and
$\textbf{sign}(\cdot)$ is the component sign.
\end{enumerate}
\begin{remark}[Spherical symmetry]\label{rem:Sp-Sym}
For spherical symmetry, the minimizer $S(\mathbf{x},\mathbf{h})$ is not unique and needs to be chosen uniformly at random from the set of all minimizers.
We can select the minimizer in the following way, by introducing an independent Gaussian vector $\varepsilon \sim N(\mathbf{0}, I_{p(p-1)})$:
Suppose $\mathbf{h},\mathbf{x}\neq \mathbf{0}$. Let $\mathbf{v} = \frac{\mathbf{h}}{\|\mathbf{h}\|}$, $\mathbf{w}=\frac{\mathbf{x}}{\|\mathbf{x}\|}$, and $V$ and $W$ be $p\times (p-1)$ matrices such that $V^\top V=W^\top W=I_{p-1}$, $V^\top \mathbf{v}=W^\top \mathbf{w} = \mathbf{0}$.
Generate $S(\mathbf{x},\mathbf{h}) \equiv S(\mathbf{x},\mathbf{h},\varepsilon)$ as a function of $\mathbf{x},\mathbf{h},\varepsilon$:
\begin{enumerate}
\item[(a)] Set $\mathbf{v}^\top$ as the first row and extend it to a $p\times p $ orthogonal matrix $[\mathbf{v}\ V]^\top \in {\rm O}(p)$.
\item[(b)] Set $\mathbf{w}$ as the first column and add $p(p-1)$ i.i.d.\ standard normal random variables $\varepsilon\in\R^{p(p-1)}$ to form a $p\times p $ matrix $[\mathbf{w}\ \varepsilon]$. Apply
Gram-Schmidt orthogonalization to the columns of $[\mathbf{w}\ \varepsilon]$ (in the order of the columns) to obtain ${\rm GS}([\mathbf{w}\ \varepsilon]) \in {\rm O}(p)$.
\item[(c)] Set $S(\mathbf{x},\mathbf{h},\varepsilon) = {\rm GS}([\mathbf{w}\ \varepsilon])[\mathbf{v}\ V]^\top$.
\end{enumerate}
The validity of the above procedure is shown in~\cref{sec:gen_sign_sph}.
\end{remark}
The following result (proved in \cref{sec:pf_consistency}) gives sufficient conditions for the consistency of the generalized sign test for the three prime examples of $\G$ considered in this paper, corresponding to central, sign and spherical symmetry.


\begin{theorem}[Consistency of the generalized sign test]\label{prop:consistency_sign}
Suppose Assumptions~\ref{assump:nu-G} and~\ref{assump:Weak-nu_n} hold.
Then, the generalized sign test, which rejects $\mathrm{H}_0$ for large values of $\|T_n\|_F$, is consistent if 
\begin{itemize}
\item (Central/Sign symmetry) $\G$ acts freely on $\G B$, and $\mathbb{E} [S(\mathbf{X},R( \mathbf{X}))]\neq \mathbf{0}$.
\item (Spherical symmetry) $\G= {\rm O}(p)$, $\nu(\{\mathbf{0}\})=0$, and $\mathbb{E} [S(\mathbf{X},R( \mathbf{X}),\varepsilon)]\neq \mathbf{0}$, where $S(\mathbf{x},\mathbf{h},\varepsilon)$ is given in the above construction (which requires $\mathbf{h}\neq \mathbf{0}$).
\end{itemize}
\end{theorem}
Note that even in the univariate setting, the classical sign test may not be consistent against location shift alternatives if the symmetric distribution has no mass in a neighborhood of the origin.\footnote{For example, if $X\sim {\rm Uniform}\left([-2,-1]\cup [1,2]\right)$, then the sign test will be powerless against a location shift by $\delta < 1$.} In fact, the condition $\mathbb{E} [S(\mathbf{X},R( \mathbf{X}))]\neq \mathbf{0}$ for consistency mentioned in~\cref{prop:consistency_sign} directly generalizes its one-dimensional counterpart; see remark 
below.

\begin{remark}[Consistency of the proposed tests when $p=1$]\label{rk:classical_Wilcoxon}
If $p=1$ and $\nu$ is concentrated on $(0,\infty)$,
then by Theorem~\ref{prop:consistency_sign} the generalized sign test is consistent against alternatives where $\mathbb{E}[{\rm sign}(X)]\neq 0$. This coincides with the condition needed for the consistency of the classical two-sided sign test~\citep[see e.g.,][Section 2.3]{kolassa2020introduction}.
If $\nu={\rm Uniform}(0,1)$, then
$R(x)=F_{|X|}(|x|)$, where $F_{|X|}(\cdot)$ is the cumulative distribution function of $|X|$, and $\nabla \psi(x) = {\rm sign}(x) F_{|X|}(|x|)$.
\cref{prop:consistency_signed_rank} implies GWSR test with identity score function is consistent against
$\mathbb{E}[{\rm sign}(X)F_{|X|}(|X|)]\neq 0$, which is equivalent to the condition for the consistency of the classical Wilcoxon signed-rank test, i.e.,
$\mathbb{P}(X_1+X_2 > 0) \neq \frac{1}{2}$, where $X_1,X_2$ are i.i.d.\ from $\mu$; see \cref{sec:pf_classical_Wilcoxon} for a proof. 
\end{remark}

\subsection{Efficiency of the GWSR Test}\label{subsec:efficiency}
In this subsection we study the asymptotic relative efficiency (ARE) of our GWSR test w.r.t.~the Hotelling's $T^2$ test. We first introduce the following lemma (see \cref{sec:pf_lem:population_signed_rank} for a proof), which states that under ${\rm H}_0$, we can replace the empirical signed-ranks by their population counterparts,
while keeping the limiting distribution of the GWSR statistic unchanged. Results of this type are sometimes referred to as H{\' a}jek representation (cf.~\citet[Theorem 1, Section 6.1.7]{hajek1999represent}) and have appeared in a number of recent works on applications of multivariate ranks defined via OT \citep{shi2022universally,hallin2020fully}.

\begin{lemma}[H{\' a}jek representation]\label{lem:population_signed_rank}
Suppose Assumptions~\ref{assump:nu-G}, \ref{assump:Weak-nu_n}, \ref{assump:equivarJ}, \ref{assump:nu_n-nu}-(ii) and \ref{assump:G}-(i) hold.
Then, under $\textrm{H}_0$,
$$\frac{1}{\sqrt{n}}\sum_{i=1}^n S_n(\mathbf{X}_i)  J(R_n(\mathbf{X}_i))
- \frac{1}{\sqrt{n}}\sum_{i=1}^n S(\mathbf{X}_i,R(\mathbf{X}_i))  J(R(\mathbf{X}_i)) \overset{L^2}{\longrightarrow}0.$$
\end{lemma}
The above lemma together with Le Cam's third lemma implies the following CLT under contiguous alternatives; see \cref{sec:pf_thm:asymptotics_local_alternatives} for a proof.
\begin{theorem}[Asymptotic distribution under local alternatives]\label{thm:asymptotics_local_alternatives}
Let $\mathbf{X}_1,\ldots,\mathbf{X}_n$ be i.i.d.~with Lebesgue density $f(\mathbf{\cdot} -{\B \theta})$ on $\R^p$, where $f$ is the density of a $\G$-symmetric distribution and $\B \theta \in \R^p$. Consider testing:
\begin{equation}\label{eq:local_test}
{\rm H}_0: {\B \theta} = \mathbf{0}_p \qquad\textrm{versus}\qquad {\rm H}_1: {\B \theta} = \frac{{\B \xi}_n}{\sqrt{n}},
\end{equation}
where ${\B \xi}_n\to {\B \xi}\in\R^p$ as $n \to \infty$.
Suppose Assumptions~\ref{assump:nu-G}, \ref{assump:Weak-nu_n}, \ref{assump:equivarJ}, \ref{assump:nu_n-nu}-(ii) and \ref{assump:G}-(i) hold.
Under the assumption of differentiability under quadratic mean of the parametric family $\{f(\cdot -{\B \theta})\}_{{\B \theta} \in \R^p}$ (see~\cref{sec:regu} for the details), we have under ${\rm H}_1:{\B \theta} = \frac{{\B \xi}_n}{\sqrt{n}}$,
\begin{equation*}\label{eq:Gamma}
\mathbf{W}_n\overset{d}{\longrightarrow}  N\left({\B \gamma},\Sigma_{\rm ERD}\right),\quad \qquad \mbox{where} \;\;
{\B \gamma} := -\mathbb{E}_{{\rm H}_0} \left[J(\nabla\psi(\mathbf{X}))\frac{{\B \xi}^\top \nabla f (\mathbf{X})}{f(\mathbf{X})}\right]\in\R^p.
\end{equation*}
\end{theorem}
The above theorem implies that
$$\mathbf{W}_n^\top \Sigma_{\rm ERD}^{-1} \mathbf{W}_n\overset{d}{\longrightarrow}
\left\|\Sigma_{\rm ERD}^{-1/2}{\B \gamma} + N(\mathbf{0},I_p)  \right\|^2,$$ which is a non-central $\chi^2$-distribution with $p$ degrees of freedom and {\it non-centrality} parameter 
$\|\Sigma_{\rm ERD}^{-1/2}{\B \gamma}\|^2$.  Note that this limiting distribution does not depend on $\G$ directly; it only depends on $J(\cdot)$, $\nabla \psi(\cdot)$ (which further depends on $\nu_S$; see~\cref{thm:population_rank}-$(ii)$), ${\B \xi} \in \R^p$, and the distribution of $\mathbf{X}$ under ${\rm H}_0$. We compare this test with the Hotelling's $T^2$ statistic given by
\begin{equation}\label{eq:Hotelling-T2}
T^2=n\bar{\mathbf{X}}_n^\top S_n^{-1} \bar{\mathbf{X}}_n,
\end{equation}
where $S_n := \frac{1}{n-1}\sum_{i=1}^n (\mathbf{X}_i-\bar{\mathbf{X}}_n)(\mathbf{X}_i-\bar{\mathbf{X}}_n)^\top$ is the sample covariance matrix. Observe that as $n\to\infty$, $S_n\overset{p}{\longrightarrow}\Sigma_\mathbf{X} := \mathbb{E}[\left(\mathbf{X} - \mathbb{E}\mathbf{X}\right)\left(\mathbf{X} - \mathbb{E}\mathbf{X}\right)^\top]$, if the latter is finite. The absolute continuity of $\mathbf{X}\sim \mu$ implies that $\Sigma_\mathbf{X}$ is positive definite.
Under the sequence of contiguous alternatives ${\B \theta}_n \equiv {\B \xi}_n/\sqrt{n}$, we have
$$\begin{aligned}
\sqrt{n} S_n^{-1/2}\bar{\mathbf{X}}_n - \Sigma_\mathbf{X}^{-1/2}{\B \xi}=\sqrt{n}S_n^{-1/2} \left( \bar{\mathbf{X}}_n - {\B \xi}_n/\sqrt{n}\right) + (S_n^{-1/2}{\B \xi}_n  - \Sigma_\mathbf{X}^{-1/2}{\B \xi})\overset{d}{\longrightarrow}N(\mathbf{0},I_p).
\end{aligned}$$
Hence, under this contiguous alternative, the Hotelling's $T^2$ statistic has non-central $\chi^2$-distribution with $p$ degrees of freedom and {\it non-centrality} parameter $ \|\Sigma_\mathbf{X}^{-1/2}{\B \xi}\|^2$.

\begin{definition}[Asymptotic relative efficiency]\label{defn:ARE}
For ${\B \xi}\in\R^p\backslash\{\mathbf{0}\}$, consider testing
$${\rm H}_0: {\B \theta} = \mathbf{0}_p \qquad\textrm{versus}\qquad{\rm H}_1: {\B \theta} = \varepsilon_n {\B \xi}_n,\quad  {\rm where\ }\varepsilon_n\to 0,\ {\B \xi}_n\to{\B \xi}$$
in the location shift model $f(\mathbf{\cdot} -{\B \theta})$, where $f$ is the density of a $\G$-symmetric distribution on $\R^p$.
For $0<\alpha <\beta<1$, let $N_1(\alpha,\beta,{\B \xi})$ be the minimum number of samples needed for the (asymptotic) level $\alpha$ test based on $\mathbf{W}_n$ to achieve power $\beta$ when ${\B \theta}={\B \xi}$, i.e.,
$$N_1\left(\alpha,\beta,{\B \xi}\right):=\min\left\{K\geq 1: \mathbb{P}_{{\B \theta} ={\B \xi}} \left( \mathbf{W}_K^\top \Sigma_{\rm ERD}^{-1} \mathbf{W}_K \geq \chi_p^2(\alpha) \right)\geq \beta \right\},$$
where $\chi_p^2(\alpha)\in\R$ satisfies $\mathbb{P}\left(\chi_p^2 \geq \chi_p^2(\alpha) \right) = \alpha$.
Similarly, let $N_2(\alpha,\beta,{\B \xi})$ denote the minimum number of samples needed for the (asymptotic) level $\alpha$ test based on the Hotelling's $T^2$ statistic to achieve power $\beta$ when ${\B \theta} = {\B \xi}$.
The asymptotic relative efficiency (ARE) of the GWSR test w.r.t.~Hotelling's $T^2$, in the direction ${\B \xi}$, is defined as:
 $${\rm ARE}(\mathbf{W}_n,\bar{\mathbf{X}}_n;{\B \xi})=\lim_{n\to\infty}\frac{N_2\left(\alpha,\beta,\varepsilon_n {\B \xi}_n\right)}{N_1\left(\alpha,\beta,\varepsilon_n {\B \xi}_n\right)}.$$
\end{definition}
The ARE, in general can depend on $\alpha$, $\beta$, $\{\varepsilon_n\}$ and $\{{\B \xi}_n\}$. However, if the tests have asymptotically non-central chi-squared distributions with the same degrees of freedom, as in our case here (such that~\cref{thm:asymptotics_local_alternatives} holds),
then the ARE is free of $\alpha, \beta, \{\varepsilon_n\}$, and is the squared ratio of the corresponding non-centrality parameters 
(see \cref{sec:pf:defnARE} for a proof):
$${\rm ARE}(\mathbf{W}_n,\bar{\mathbf{X}}_n;{\B \xi})=\frac{\|\Sigma_{\rm ERD}^{-1/2}{\B \gamma}\|^2}{\| \Sigma_\mathbf{X}^{-1/2}{\B \xi} \|^2}.$$
The above ARE turns out to be lower-bounded for some sub-families of multivariate distributions. 

In the one-dimensional case, when $f(\cdot)$ is the density of $N(0,\sigma^2)$ for some unknown $\sigma>0$, the relative efficiency of the Wilcoxon signed-rank test against the $t$-test is $3/\pi\approx 0.95$ \citep{Hodges1956efficiency}. The following proposition (see \cref{sec:pf_prop:Gaussian} for a proof) generalizes this result.
A Chernoff-Savage-type result is also provided, which illustrates that the GWSR test can be no worse than the Hotelling's $T^2$ test, in terms of ARE, if the Gaussian reference distribution (with identity score function) is used.


\begin{prop}[Gaussian case]\label{prop:Gaussian}
Suppose Assumptions~\ref{assump:nu-G}, \ref{assump:Weak-nu_n}, \ref{assump:nu_n-nu}-(ii) and \ref{assump:G}-(i) hold.
Suppose $f(\cdot)$ is the density of a multivariate normal distribution $N(\mathbf{0}_p,\Sigma)$ for some unknown positive definite covariance matrix $\Sigma$.
\begin{enumerate}
\item If ${\rm ERD}=N(\mathbf{0},I)$, with $J(\mathbf{x})=\mathbf{x}$, then for all ${\B \xi}\neq \mathbf{0} \in \R^p$, ${\rm ARE}(\mathbf{W}_n,\bar{\mathbf{X}}_n;{\B \xi})=1$.

\item If ${\rm ERD} = {\rm Uniform}(-1,1)^p$, obtained from $J(\mathbf{x})=(2\Phi(x_1)-1,\ldots,2\Phi(x_p)-1)^\top$, where $\Phi$ is the cumulative distribution function (c.d.f.) of the standard normal distribution\footnote{Note that this ERD excludes testing for spherical symmetry as $J(\cdot)$ is not equivariant under the group action of ${\rm O}(p)$.}, then for all ${\B \xi}\neq \mathbf{0}$, ${\rm ARE}(\mathbf{W}_n,\bar{\mathbf{X}}_n;{\B \xi})={3}/{\pi}\approx 0.95$.

\item If ${\rm ERD}$ is the spherical uniform distribution\footnote{That is, the product of the uniform distribution over the unit sphere with the uniform distribution over the unit interval of distances to the origin \citep{hallin2017distribution,Shi2022indep,Hallin2022VARMA,hallin2020fully}.}, with $J(\mathbf{x})=\mathbf{x}$, then for all ${\B \xi}\neq \mathbf{0}$,
$${\rm ARE}(\mathbf{W}_n,\bar{\mathbf{X}}_n;{\B \xi})=\kappa_p \geq \left\{\begin{aligned}
& 0.95,\quad\quad {\rm for\ }p<5,\\
&0.648,\quad\ \ {\rm for\ all\ }p,
\end{aligned}\right.
$$
(see \cref{sec:pf_prop:Gaussian} for the exact value of $\kappa_p$).
In particular, $\kappa_1 = 3/\pi$ reduces to the classical ARE of the Wilcoxon's signed-rank test against the $t$-test \citep{Hodges1956efficiency}.
\end{enumerate}
\end{prop}
Note that the above result is valid for all possible choices of $\G$ as long as the ERD is of the specified form. We will observe this phenomenon for all results in this section. A remarkable result of the classical Wilcoxon signed-rank test in one-dimension is the Hodges-Lehmann phenomenon, which states that the relative efficiency of the classical Wilcoxon signed-rank test w.r.t.\ the $t$-test for location shift alternatives never falls below 0.864 \citep{Hodges1956efficiency}.
The following theorem (see \cref{sec:pf_thm:indepcomp} for a proof) provides a multivariate analogue to this result, albiet in a more restrictive setting --- where the components of the multivariate vector are assumed to be independent.

\begin{theorem}[Independent components case]\label{thm:indepcomp}
Under the assumptions of \cref{thm:asymptotics_local_alternatives}, 
suppose further that $f(\mathbf{x})=\prod_{i=1}^p \tilde{f}_i(x_i)$, for $\mathbf{x} = (x_1,\ldots, x_p) \in \R^p$, where $\tilde{f}_1,\ldots,\tilde{f}_p$ are univariate Lebesgue densities with finite non-zero variances.
\begin{enumerate}
\item If the ERD is ${\rm Uniform}(-1,1)^p$, with $J(\mathbf{x})=\mathbf{x}$, then ${\rm ARE}(\mathbf{W}_n,\bar{\mathbf{X}}_n;{\B \xi})\geq  \frac{108}{125} \approx 0.864$.

Further, equality in the ARE lower bound holds if and only if for all $\xi_i\neq 0$, $\tilde{f}_i$ has density:
$$\tilde{f}_i(x)=\frac{3}{20\sqrt{5}\sigma_i^3}(5\sigma_i^2 - x^2) I(|x|\leq \sqrt{5}\sigma_i),\quad {\rm for\ some\ }\sigma_i>0.$$
\item If the ERD is $N(\mathbf{0},I)$, with $J(\mathbf{x})=\mathbf{x}$, then ${\rm ARE}(\mathbf{W}_n,\bar{\mathbf{X}}_n;{\B \xi})\geq  1$.

Further, equality holds if and only if for all $\xi_i\neq 0$, $\tilde{f}_i$ is the density of a normal distribution.
\end{enumerate}
\end{theorem}

The ARE of the GWSR test w.r.t.\ the Hotelling's $T^2$ test can also be lower-bounded for another rich class of multivariate probability distributions, namely the class of elliptically symmetric distributions; this is stated in the theorem below (see \cref{sec:pf_ellip} for a proof).
\begin{theorem}[Elliptically symmetric case]\label{thm:ellip}
Under the assumptions of \cref{thm:asymptotics_local_alternatives}, 
suppose further that $\mathbf{X}$ has finite second moments under ${\rm H}_0$ and
$$f(\mathbf{x})\propto ({\rm det}(\Sigma))^{-\frac{1}{2}}\underline{f}\left(\mathbf{x}^\top \Sigma^{-1}\mathbf{x}  \right), \qquad \mbox{ for } \mathbf{x}\in \R^p,$$
for some unknown positive definite matrix $\Sigma$, and some unknown function $\underline{f}:\R^+\to \R^+$.
Under standard regularity conditions on the smoothness and integrability of $\underline{f}$
\citep{deb2021efficiency,hallin2017distribution} (see \cref{sec:pf_ellip} for details), we have
\begin{enumerate}
\item If the ERD is the spherical uniform distribution with  $J(\mathbf{x})=\mathbf{x}$, then:
$$ {\rm ARE}(\mathbf{W}_n,\bar{\mathbf{X}}_n;{\B \xi}) \ge  \frac{81}{500}\cdot \frac{(\sqrt{2p-1}+1)^5}{p^2(\sqrt{2p-1}+5)}\geq 0.648.$$
See \cref{sec:pf_ellip} for the distribution that attains the ARE lower bound above.
\item If the ERD is $N(\mathbf{0},I)$ with $J(\mathbf{x})=\mathbf{x}$, then ${\rm ARE}(\mathbf{W}_n,\bar{\mathbf{X}}_n;{\B \xi})\geq 1$.

Equality above is attained if and only if $f(\cdot)$ is a multivariate normal distribution.
\end{enumerate}
\end{theorem}
In the following, we consider a class of generative models used in blind source separation \citep{Cardoso1998blind,comon2010handbook,shlens2014tutorial,deb2021efficiency} and in independent component analysis \citep{Hyvarunen2000ICA,Samarov2004ica,Bach2003KICA} and derive similar lower bounds on the ARE (see~\cref{sec:pf_blind} for a proof).
\begin{theorem}[Generative model for blind source mixing]\label{thm:blind}
Under the assumptions of \cref{thm:asymptotics_local_alternatives}, 
suppose further that $\mathbf{X} = A \mathbf{W}$ under $\mathrm{H}_0$ for some unknown orthogonal matrix $A$ and $\mathbf{W}$ having independent components with Lebesgue density $\tilde{f}(\mathbf{w})=\prod_{i=1}^p\tilde{f}_i(w_i)$, where $\mathbf{w}=(w_1,\ldots,w_p)$ and $\tilde{f}_1,\ldots,\tilde{f}_p$ are univariate densities.
Suppose the ERD is $N(\mathbf{0},I)$, with $J(\mathbf{x})=\mathbf{x}$, and $\mathbf{X}$ has finite second moments under ${\rm H}_0$. Then: $${\rm ARE}(\mathbf{W}_n,\bar{\mathbf{X}}_n;{\B \xi})\geq 1.$$
The equality is attained if and only if $\tilde{f}_i$ is the density of a normal distribution whenever $(A^\top {\B \xi})_i\neq 0$, $i=1,2,\ldots,p$.
\end{theorem}
Note that in all of the above ARE results, there is no restriction on $\G$ as long as the ERD is of the required form.~\cref{thm:blind} again illustrates the superiority of using a Gaussian ERD with the identity score function. Based on this fact, we recommend some choices of the reference measure $\nu$ below for different $\G$'s.

\begin{remark}[Examples of the reference measure $\nu$]\label{rk:ref_mea}
Here we recommend some choices of the reference measure $\nu$,
based on the canonical choice of the score function $J(\mathbf{x})=\mathbf{x}$:
\begin{enumerate}
\item Our first recommendation is a reference distribution $\nu$ with Gaussian ERD --- $N(\mathbf{0},I_p)$.
For example, for central symmetry, we can take $\nu$ to be the distribution of $ (|Z_1|,Z_2,\ldots,Z_p)$ (concentrated on $ B=(0,\infty)\times \mathbb{R}^{p-1}$) where $Z_i$'s are i.i.d.\ standard normal;
for sign symmetry, we can take $\nu$ as the distribution of $ (|Z_1|,|Z_2|,\ldots,|Z_p|)$ (here $ B=(0,\infty)^p$). For spherical symmetry, we can take
$\nu$ to be the distribution of $ \left(\sqrt{Z_1^2+\ldots+Z_p^2},{0},\ldots,{0}\right)$ and $ B=(0,\infty)\times \{0\}^{p-1}$.
Using Gaussian ERD is particularly powerful as the ARE of the GWSR test w.r.t.\ the Hotelling's $T^2$ test never falls below 1 for a large class of models (see Theorems~\ref{thm:indepcomp}-\ref{thm:blind}). 
\item Our second recommendation is a reference distribution $\nu$ with uniform ERD --- ${\rm Uniform}(-1,1)^p$.
For example, for central symmetry, we can take $\nu$ to be the distribution of $(|U_1|,U_2,\ldots,U_p)$ (concentrated on $ B=(0,1)\times (-1,1)^{p-1}$) where $U_i$'s are i.i.d.\ ${\rm Uniform}(-1,1)$;
for sign symmetry, we can take $\nu$ as the distribution of $ (|U_1|,|U_2|,\ldots,|U_p|)$ (thus $ B = (0,1)^p$).
Note that ${\rm Uniform}(-1,1)^p$ is not spherically symmetric, so this ERD is not permissible for testing spherical symmetry.
\item Our third recommendation is a reference distribution $\nu$ with spherical uniform ERD, which reduces to ${\rm Uniform}(-1,1)$ when $p=1$, and has been heavily advocated in
the literature of distribution-free inference using OT~\citep{hallin2017distribution,Shi2022indep,Hallin2022VARMA,hallin2020fully}.
Reference measures for central symmetry, sign symmetry, and spherical symmetry can also be constructed similarly as described above for the Gaussian ERD.
\end{enumerate}
\end{remark}

We have so far illustrated lower bounds on the ARE of the GWSR test against the Hotelling's $T^2$ test.
A natural question arises whether the generalized sign test would also share similar properties?
A positive result holds for sign symmetry, where the generalized signs coincide with the component-wise signs, and was studied in \citet{Bickel1965competitors}. The following theorem is a direct consequence of \citet[Equation (5.5)]{Bickel1965competitors}.
\begin{theorem}[ARE for the generalized sign test under sign symmetry \citep{Bickel1965competitors}]\label{rk:AREsign}
Suppose the reference distribution $\nu_n$ is concentrated on $(0,+\infty)^p$ for all $n$. Assume a Lebesgue density $f(\cdot)$ is sign symmetric with non-degenerate covariance matrix, and each marginal density $f_{i}$ of $f$ has $0$ as its mode, and is continuous at $0$.
Consider the model $f(\cdot -{\B \theta})$ and $H_0$: ${\B \theta}=\mathbf{0}_p$ against contiguous alternatives. 
Then the relative efficiency of the generalized sign test (using sign symmetry) with respect to Hotelling's $T^2$-test never falls below $1/3$.
\end{theorem}

\begin{table}
\caption{Relative efficiencies of the generalized sign test (for sign symmetry) versus the Hotelling's $T^2$ test for some sign symmetric distributions.
The minimally possible value of the relative efficiency is $1/3$.}
\centering
\label{table:AREsign}
\begin{tabular}{ll}
\toprule
Distribution & ARE (sign/Hotelling's $T^2$)\\
\hline
 Laplace: $\prod_{i=1}^p {\rm Laplace}(0,b_i)$& 2   \\
 Normal: $\prod_{i=1}^p N(0,\sigma_i^2)$ & $2/\pi$\\
 Logistic: $\prod_{i=1}^p {\rm Logistic}(0,s_i)$ & $\pi^2/12$\\
 Uniform: $\prod_{i=1}^p {\rm Uniform}(-c_i,c_i)$ & 1/3 (lower bound)\\
 \bottomrule
\end{tabular}
\end{table}

In Table \ref{table:AREsign},
the relative efficiencies of the generalized sign test (for sign symmetry) versus the Hotelling's $T^2$ test for some distributions are given (using the formula \citep[Equation (5.2)]{Bickel1965competitors}), which is exactly the same as in one-dimensional case \citep[Table 14.1]{vanderVaart1998}.
Besides, by designing a large density at $0$, the relative efficiency can be arbitrarily large, which would strongly favor the sign test over the Hotelling's $T^2$ test.
However, For other types of symmetry, there is in general no lower bound of ARE.
The ARE could depend on $\alpha$, $\beta$, the reference distribution, the direction of ${\B \xi}$, and the dimension $p$.
From our simulations in \cref{sec:simulations}, we also find that the generalized sign test using sign symmetry tend to have a better performance than using the central symmetry or the spherical symmetry. Hence, we would recommend using sign symmetry as the primary choice when performing the generalized sign test.

So far, we have shown that for natural (and practical) choices of sub-families of location shift alternatives and suitable score functions (mostly $J(\mathbf{x})=\mathbf{x}$),
the GWSR test performs at least as well as the Hotelling's $T^2$ test.
In the following result (proved in~\cref{pf:thm:efficiency}), we show, through two separate results, that there exists a score transformed Wilcoxon's signed-rank test (with appropriately chosen score function) that is {\it locally asymptotically optimal}, i.e., it has maximum local power (against contiguous alternatives) among all tests with fixed Type I error.
\begin{theorem}[Local asymptotic optimality]\label{thm:efficiency}
Let $\mathbf{X}_1,\ldots,\mathbf{X}_n$ be a sample from $P_{\B \theta}$ with Lebesgue density $f(\cdot-{\B \theta})$, where $f$ is $\G$-symmetric about ${\B \theta}_0 :=\mathbf{0}_p$.
Under the assumptions in \cref{thm:asymptotics_local_alternatives}, let $I_{{\B \theta}_0}$ be the nonsingular Fisher information matrix at ${\B \theta}_0 =\mathbf{0}_p$.
Suppose $\nu_S$ has a Lebesgue density. 
Let $\psi^*(\mathbf{y}):=\sup_{\mathbf{x} \in \R^p} \{\mathbf{y}^\top \mathbf{x} -\psi(\mathbf{x}) \}$ be the Legendre-Fenchel transform of $\psi$, and suppose 
\begin{equation}\label{eq:Opt-Score-Fnc}
J(\mathbf{x}):=-\frac{\nabla f(\nabla\psi^*(\mathbf{x}))}{f(\nabla\psi^*(\mathbf{x}))}, \qquad \mbox{for} \;\; \mathbf{x} \in \G B,
\end{equation}
is used in defining $\mathbf{W}_n$ (via~\eqref{eq:score_Wn}).
\begin{enumerate}
\item (Locally maximin optimal test) Let ${\B \xi} \ne \mathbf{0} \in \R^p$. For any sequence of level $\alpha \in (0,1)$ tests $\{\phi_n\}_{n \ge 1}$ for testing
\begin{equation}\label{eq:Hypo-Test-Local}
{\rm H}_0: {\B \theta}=\mathbf{0}_p\qquad {\rm versus}\qquad {\rm H}_1:{\B \theta}=\frac{{\B \xi}}{\sqrt{n}}
\end{equation}
with power function $\pi_n({\B \theta}) :=\E_{\B \theta}[\phi_n]$, such that  $\pi_n\left(\frac{{\B \xi}}{\sqrt{n}}\right)$ converges under ${\rm H}_1$ for any ${\B \xi}\in\R^p$, we have
$$\begin{aligned}
\inf_{{\B \xi}: \|I^{1/2}_{{\B \theta}_0} {\B \xi}\|=c}\lim_{n\to\infty} \pi_n \left(\frac{{\B \xi}}{\sqrt{n}} \right)\leq \inf_{{\B \xi}: \|I^{1/2}_{{\B \theta}_0} {\B \xi}\|=c} \lim_{n\to\infty} \mathbb{P}_{{\B \theta}_0 + \frac{{\B \xi}}{\sqrt{n}}}\left(\mathbf{W}_n^\top \Sigma_{\rm ERD}^{-1} \mathbf{W}_n \geq \chi_p^2(\alpha)\right),
\end{aligned}$$
for all $c>0$.
\item (Locally most powerful test) Let $\zeta:\R^p\to\R$ be differentiable at ${\B \theta}_0=\mathbf{0}_p$ and $\zeta({\B \theta}_0)=0$.
Let ${\B \xi}\in\R^p$ be such that ${\B \xi}^\top\nabla\zeta({\B \theta}_0)>0$. 
For any sequence of level $\alpha \in (0,1)$ tests $\{\phi_n\}_{n \ge 1}$ for 
\begin{equation}\label{eq:one-side-test}
    {\rm H}_0: \zeta({\B \theta})\leq 0 \qquad {\rm versus}\qquad {\rm H}_1:\zeta({\B \theta})>0
\end{equation}
with power function $\pi_n({\B \theta})=\E_{\B \theta} [\phi_n]$, we have
$$\begin{aligned}
\limsup_{n\to\infty} \pi_n \left(\frac{{\B \xi}}{\sqrt{n}} \right)\leq \lim_{n\to\infty}\mathbb{P}_{{\B \theta}_0 + \frac{{\B \xi}}{\sqrt{n}}}\left(\frac{\nabla \zeta ({\B \theta}_0)^\top I_{{\B \theta}_0}^{-1} \mathbf{W}_n}{\sqrt{\nabla \zeta ({\B \theta}_0)^\top I_{{{\B \theta}_0}}^{-1}\nabla \zeta ({\B \theta}_0)}} \geq z_{\alpha} \right),
\end{aligned}$$
where $z_\alpha$ satisfies $\Phi(z_\alpha)=1-\alpha$, and
$\frac{\nabla \zeta ({\B \theta}_0)^\top I_{{\B \theta}_0}^{-1} \mathbf{W}_n}{\sqrt{\nabla \zeta ({\B \theta}_0)^\top I_{{{\B \theta}_0}}^{-1}\nabla \zeta ({\B \theta}_0)}} \geq z_{\alpha}$ is the appropriate (asymptotic) level $\alpha$ test based on the GWSR statistic $\mathbf{W}_n$ for~\eqref{eq:one-side-test}.
\end{enumerate}
\end{theorem}

The first result above shows that the GWSR test maximizes the minimum (local) power over the contour $\{{\B \xi}\in\R^p: \|I^{1/2}_{{\B \theta}_0} {\B \xi}\|=c\}$ for any $c >0$, while the second result tells us that if the hypothesis is one-sided, then a test based on the GWSR statistic $\mathbf{W}_n$ would achieve maximum (local) power for any fixed direction ${\B \xi} \in \R^p$. 
Although tests based on $\mathbf{W}_n$ can be locally asymptotically optimal, the optimal score function $J(\cdot)$ (in~\eqref{eq:Opt-Score-Fnc}) depends on the underlying distribution $f$, and is in general unknown.


\begin{example}[Gaussian case]
Suppose it is known that $f(\cdot)$ (in~\cref{thm:efficiency}) follows a Gaussian distribution $N(\mathbf{0}_p,\Sigma)$ with unknown nonsingular covariance matrix $\Sigma$.
Then with $S\mathbf{H}\sim N(\mathbf{0},I_p)$, \cref{thm:efficiency} states that the optimal score function (in~\eqref{eq:Opt-Score-Fnc})  is $J(\mathbf{x})=\Sigma^{-1/2}\mathbf{x}$.
Although $\Sigma$ is unknown, it may be estimated by the sample covariance matrix $\hat{\Sigma}$. Define
$\hat{J}(\mathbf{x})=\hat{\Sigma}^{-1/2}\mathbf{x}$. Then, 
$$\mathbf{W}_n = \frac{1}{\sqrt{n}}\sum_{i=1}^n S_n(\mathbf{X}_i)  J(R_n(\mathbf{X}_i)) = \Sigma^{-1/2}\left( \frac{1}{\sqrt{n}}\sum_{i=1}^n S_n(\mathbf{X}_i)  R_n(\mathbf{X}_i)\right)$$
has the same asymptotic distribution as
\begin{equation}\label{eq:W_n-Est}
\hat{\Sigma}^{-1/2}\left( \frac{1}{\sqrt{n}}\sum_{i=1}^n S_n(\mathbf{X}_i)  R_n(\mathbf{X}_i)\right).
\end{equation}
This means that using the data-driven score function $\hat{J}(\cdot)$ has the same asymptotic distribution as using $J(\cdot)$, and thus the GWSR test based on~\eqref{eq:W_n-Est} also has the same local asymptotic optimal properties.
\end{example}

Generalizing the above example, if we can estimate optimal score function $J(\cdot)=-\frac{\nabla f(\nabla\psi^*(\cdot))}{f(\nabla\psi^*(\cdot))}$ in \cref{thm:efficiency} sufficiently well such that
$$\frac{1}{\sqrt{n}}\sum_{i=1}^n S_n(\mathbf{X}_i)  J(R_n(\mathbf{X}_i)) = \frac{1}{\sqrt{n}}\sum_{i=1}^n S_n(\mathbf{X}_i)  \hat{J}(R_n(\mathbf{X}_i)) + o_p(1),$$
then the GWSR test using $\hat{J}$ will be locally asymptotically optimal as described in \cref{thm:efficiency}.

However, in general, estimating the optimal $J(\cdot)=-\frac{\nabla f(\nabla\psi^*(\cdot))}{f(\nabla\psi^*(\cdot))}$ faces two main obstacles, which are both hard problems without additional information.
First, $\frac{\nabla f(\cdot)}{f(\cdot)}$ is unknown and may be estimated empirically.
Estimating the gradient of a log-density function is in principle a difficult non-parametric problem \citep{Hyv2005score}. A positive result was given in \citet{Hyv2005score}, where the author proposed the {\it score matching} method which minimizes an empirical version of
$\mathbb{E}_{\mathbf{X}\sim f}\left[\|\mathbf{s}_{\tau}(\mathbf{X})-\nabla \log f(\mathbf{X}) \|_2^2\right]$ over $\tau$,
with $\mathbf{s}_\tau:\R^p\to\R^p$ is parametrized by $\tau$. In \citet{song2019generative}, $\mathbf{s}_{\tau}$ is taken as a neural network;
$\mathbf{s}_{\tau}$ can also be taken to belong to a parametric family~\citep{Hyv2005score}.
Another obstacle is the estimation of $\nabla\psi^* (\cdot)$ (which is the transport map from $\nu_S$ to $P_{{\B \theta}_0}$). Though with enough data from $P_{{\B \theta}_0}$, we can estimate $\nabla\psi^* (\cdot)$ up to arbitrary precision ~\citep{deb2021rates,manole2021plugin}, the convergence rate may still suffer from the curse of dimensionality \citep{deb2021rates,manole2021plugin}.

\section{Finite Sample Performance}\label{sec:simulations}
In this section we examine the finite sample performance of our distribution-free procedures.
We first illustrate the ARE results in \cref{sec:ARE}, and then
compare our proposals with existing methods in the literature that test for different types of symmetries.
We will use the identity score function $J(\mathbf{x})=\mathbf{x}$ in this section.
\subsection{Empirical Validation of the ARE Results}
Consider testing $\G$-symmetry for the following distributions:
\begin{enumerate}
\item[(a)] A bivariate Epanechnikov distribution (see \cref{thm:indepcomp}) with $\sigma_i=\frac{1}{\sqrt{5}}$ and independent components, shifted by $0.05\cdot \mathbf{1}_2$. Here we test for central and sign symmetry.

\item[(b)] A bivariate standard Gaussian distribution shifted by $0.1\cdot \mathbf{1}_2$. Here we test central, sign and spherical symmetry.
\end{enumerate}
We will consider two types of effective reference distributions: (i) Gaussian ERD: $N(\mathbf{0},I_2)$, and (ii) Uniform ERD: ${\rm Uniform}(-1,1)^2$.
For different types of symmetries, we use the reference distributions $\nu$ recommended in \cref{rk:ref_mea}:
\begin{itemize}
\item Central symmetry with Uniform ERD (CU): $\nu \sim {\rm Uniform}(0,1)\times {\rm Uniform}(-1,1)$.
\item Sign symmetry with Uniform ERD (SU): $\nu \sim {\rm Uniform}(0,1)^2$.
\item Central symmetry with Gaussian ERD (CG): $\nu \sim |N(0,1)|\times N(0,1)$.
\item Sign symmetry with Gaussian ERD (SG): $\nu \sim |N(0,1)|\times |N(0,1)|$.
\item Spherical symmetry with Gaussian ERD (SpG): $\nu \sim \sqrt{\chi^2_2} \times \{0\}$.
\end{itemize}
\begin{figure}
    \centering
    \includegraphics[width = 1\textwidth]{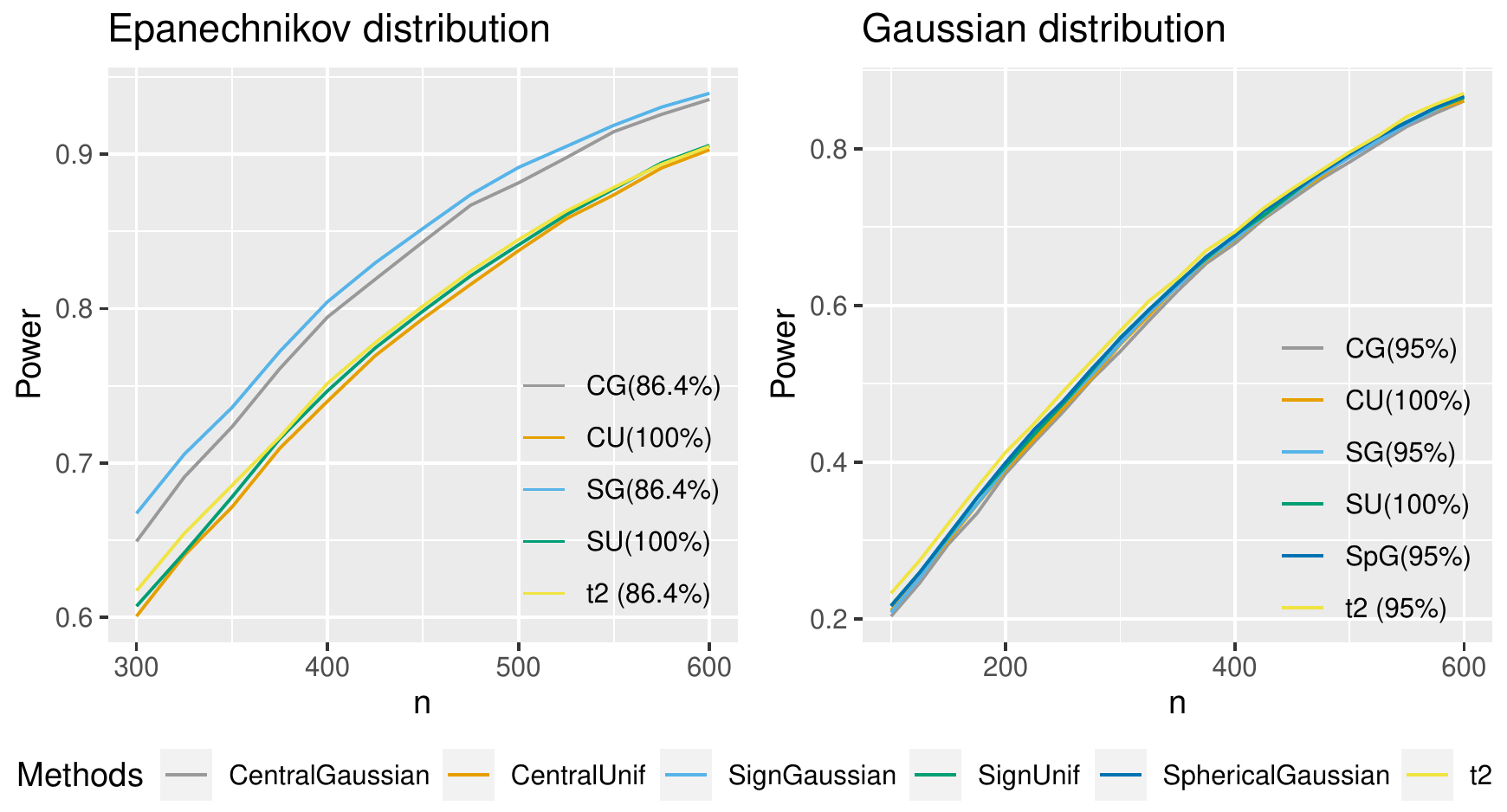}
\caption{Left: The power curves for different methods based on sample size $n$ or $0.864n$ (indicated by the percentage inside brackets) from a shifted Epanechnikov distribution. Right: The power curves of different methods based on $n$ (or $0.95n$) samples from a shifted Gaussian distribution. Here, CG, CU, SG, SU, and SpG correspond to our GWSR test for different symmetries with different ERDs; the first letter denotes the symmetry being tested (C for central, S for sign, and Sp for spherical) and the second letter denotes the ERD adopted (G for $N(\mathbf{0},I_p)$ and U for ${\rm Uniform}[-1,1]^p$).
    ``t2" denotes the Hotelling's $T^2$ test.
        }
    \label{ARE_check}
\end{figure}
Within each replication, the empirical reference distribution
$\nu_n$ is obtained by drawing $n$ i.i.d.\ observations from $\nu$ and fixing it.
The empirical power (over $10^4$ replications) of our GWSR test at level $\alpha = 0.05$ and rejection region
$$\mathbf{W}_n^\top \Sigma_{\rm ERD}^{-1} \mathbf{W}_n \geq \chi_p^2(\alpha)$$
are shown in \cref{ARE_check} for different sample sizes $n$.
The power of Hotelling's $T^2$ test is also given, with rejection region: $T^2 \geq \chi_p^2(\alpha)$; see~\eqref{eq:Hotelling-T2}.

From \cref{thm:indepcomp} we know that the Epanechnikov distribution in (a) achieves the equality in the 0.864 lower bound for Uniform ERD.
Hence with a Uniform ERD and $n$ samples, our method using central symmetry (CU), and sign symmetry (SU) achieve almost the same power as Hotelling's $T^2$ test using $0.864n$ samples; see the yellow, orange and green curves in the left panel of~\cref{ARE_check}.

If the ERD is Gaussian, then from \cref{thm:indepcomp}, the ARE of our methods with respect to Hotelling's $T^2$ is strictly larger than 1. This can also be observed in the left panel of~\cref{ARE_check} --- see the blue and grey curves (corresponding to CG and SG) which are strictly above the yellow curve (corresponding to Hotelling's $T^2$), all obtained from the same number of samples (i.e., 86.4\% $n$). Also, observe that the blue curve (SG) is slightly above the grey curve (CG); this is due to the fact that sign symmetry imposes more constraints than central symmetry (note that the Epanechnikov distribution satisfies both symmetry assumptions).

For (b), the shifted Gaussian example, we know that the
lower bound in ${\rm ARE}(\mathbf{W}_n,\bar{\mathbf{X}}_n)\geq 1$ is attained if Gaussian ERD is used (see Theorems~\ref{thm:indepcomp} and~\ref{thm:ellip}).
This can also be observed in the right panel of \cref{ARE_check} where we see that the four curves --- yellow, dark blue, grey and blue --- are on top of each other.

that our method using any of the central symmetry (CG), sign symmetry (SG), and spherical symmetry (SpG) has almost the same power as Hotelling's $T^2$ test using the same number of samples.
If using our method with a Uniform ERD, then ${\rm ARE}(\mathbf{W}_n,\bar{\mathbf{X}}_n)= 3/\pi\approx 0.95$. 
Hence, again from the right panel of \cref{ARE_check}, our method with Uniform ERD using either central symmetry (CU) or sign symmetry (SU) and $n$ samples has almost the same power as Hotelling's $T^2$ test using $0.95n$ samples.

\subsection{Power Comparison}\label{subsec:power_compare}
In this subsection, we compare our methods with existing procedures in the literature that test for different types of symmetries.
We will use the recommended Gaussian ERD with identity score (see~\cref{rk:ref_mea}).
Apart from $\nu_n$ obtained via random sampling of the $\mathbf{h}_i$'s from $\nu$, as adopted earlier, we will also use a quasi-Monte Carlo sequence (the Halton transformed sequence) to construct $\nu_n$ as explained below. The Halton sequence \citep{halton1964algorithm} $\{\mathbf{U}_i^p\}_{i\ge 1}$ (resp.\ $\{{U}_i^1\}_{i\ge 1}$) provides a discretization of ${\rm Uniform}(0,1)^p$ (resp.\  ${\rm Uniform}(0,1)$).
Let $F_{Z}(\cdot)$ be the distribution function of $Z\sim N(0,1)$, $F_{|Z|}(\cdot)$
be the c.d.f.~of $|Z|\sim |N(0,1)|$, and $F_{\|\mathbf{Z}\|}$ be the c.d.f.~of $\|\mathbf{Z}\|\sim \|N(\mathbf{0},I_p)\|$. 
The Halton transformed $\nu_n$'s for central, sign and spherical symmetries are given by the discrete uniform distribution over $\big(F_{|Z|}^{-1}(\mathbf{U}_{i,1}^p),F_{Z}^{-1}(\mathbf{U}_{i,2}^p),\ldots,F_{Z}^{-1}(\mathbf{U}_{i,d}^p)\big)$, for $i=1,\ldots,n$, the discrete uniform distribution over
$\big(F_{|Z|}^{-1}(\mathbf{U}_{i,1}^p),F_{|Z|}^{-1}(\mathbf{U}_{i,2}^p),\ldots,F_{|Z|}^{-1}(\mathbf{U}_{i,p}^p)\big)$, for $i=1,\ldots,n$,
and the discrete uniform distribution over
$\big(F_{\|\mathbf{Z}\|}^{-1}({U}_{i}^1),0,\ldots,0\big)$, for $i=1,\ldots,n$, respectively.

For low-dimensional scenarios (i.e., $p=2$), we will use the asymptotic $\chi^2$ critical values for our methods and also for Hotelling's $T^2$. In the high-dimensional settings (i.e., $p=50$), we will use the exact distribution of $T^2$ (assuming Gaussianity) and our statistics (which are distribution-free). In particular, for the GWSR test, given the reference points $\mathbf{h}_1,\ldots,\mathbf{h}_n$, we sample 1000 observations from the null distribution of $\mathbf{W}_n$, i.e., $\frac{1}{\sqrt{n}}\sum_{i=1}^n S_i\mathbf{h}_i$ where $S_i$ are i.i.d.\ from ${\rm Uniform}(\G)$, and reject the null hypothesis if the observed value of $\|\mathbf{W}_n\|^2$ is greater than 95\% of the observations from the null distribution.

We consider the following data distributions to test central symmetry (C1-C10), sign symmetry (S1-S10), and spherical symmetry (Sp1-Sp10):

\begin{enumerate}
\item[C1.] Multivariate Gaussian: $\mathbf{X}\sim N\left(\mathbf{0},
\begin{pmatrix}
2&1\\
1&3\\
\end{pmatrix}\right) + \lambda\cdot \mathbf{1}_2
$; $n=200$.
\item[C2.] Multivariate $t$-distribution with 1 degree of freedom, location parameter $\lambda\cdot \mathbf{1}_2$ and scale parameter $\begin{pmatrix}
2&1\\
1&3\\
\end{pmatrix}$; $n=200$.
\item[C3.] Uniform distribution over $[-1,0]^2 \cup [0,1]^2$, plus $\lambda\cdot \mathbf{1}_2$; $n=200$.
\item[C4.] Exponential distribution: $n = 100$, $p=2$, ${X}_1,{X}_2$ i.i.d.\ follow ${\rm Exp}(1) - 1$. $\mathbf{X}$ has mean $\mathbf{0}_p$.
\item[C5.] Exponential distribution (larger sample): Same as above except that $n=200$.
\item[C6.] Chi-squared distribution: $n = 100$, $p=2$, ${X}_1,{X}_2$ i.i.d.\ follow $\chi^2(1) - 1$. $\mathbf{X}$ has mean $\mathbf{0}_p$.
\item[C7.] Pareto distribution: $n = 200$, $p=2$, ${X}_1,{X}_2$ i.i.d.\ follow ${\rm Pareto}(1) - 2$. $\mathbf{X}$ has median $\mathbf{0}_p$.
\item[C8.] ${\rm AR}(1)$ under ${\rm H}_0$: $n = 200$, $p=50$, ${X}_1= Z_1$, $X_i = 0.5 X_{i-1} + Z_i$, $i=2,\ldots,p$, where $Z_i$ are i.i.d.\ $N(0,1)$.
\item[C9.] ${\rm AR}(1)$ under ${\rm H}_1$: Same as above except that the distribution of $\mathbf{X}$ is shifted by $0.15\cdot \mathbf{1}_d$.
\item[C10.] AR$(1)$ (${\rm H}_1$, heavy-tailed): $n = 200$, $p=50$, $Z_i$ are i.i.d.\ from the standard $t$-distribution with 1 degree of freedom. ${Y}_1= Z_1$, $Y_i = 0.5 Y_{i-1} + Z_i$, $i=2,\ldots,p$, $X_i = Y_i + 0.9$, $i=1,\ldots,p$.
\item[S1.] Multivariate Gaussian: $\mathbf{X}\sim N\left(\mathbf{0},
\begin{pmatrix}
2&0\\
0&3\\
\end{pmatrix}\right) + \lambda\cdot \mathbf{1}_2
$; $n=200$.
\item[S2.] Multivariate $t$-distribution with 1 degree of freedom, location parameter $\lambda\cdot \mathbf{1}_2$ and scale parameter $\begin{pmatrix}
2&0\\
0&3\\
\end{pmatrix}$; $n=200$.
\item[S3.] Uniform distribution over $[-1,1]^2$, plus $\lambda\cdot \mathbf{1}_2$; $n=200$.
\item[S4.] Correlation: $n=200$, $p=2$, $\mathbf{X}\sim N\left(\mathbf{0},
\begin{pmatrix}
2&1\\
1&3\\
\end{pmatrix}\right)
$.
\item[S5.] Exponential distribution: $n = 100$, $p=2$, ${X}_1,{X}_2$ i.i.d.\ follow ${\rm Exp}(1) - 1$.
\item[S6.] Chi-squared distribution: $n=100$, $p=2$, $X_1$, $X_2$ i.i.d.\ follow $\chi^2(1)-1$.
\item[S7.] Laplace: $n=200$, $p=2$, $X_1$, $X_2$ i.i.d.\ follow a Laplace distribution with location parameter 0.2 and scale parameter 1.
\item[S8.] High-dimensional (${\rm H}_0$): $n=200$, $p=50$, $Z_1,\ldots,Z_p \overset{i.i.d.}{\sim} N(0,1)$, $X_i = \sin(i)Z_i$.
\item[S9.] High-dimensional (${\rm H}_1$): $n=200$, $p=50$, $Z_1,\ldots,Z_p \overset{i.i.d.}{\sim} N(0,1)$, $X_i = \sin(i)Z_i + 0.003$.
\item[S10.] High-dimensional (${\rm H}_1$, heavy-tailed): $n=200$, $p=50$, $X_i = \sin(i)Z_i + 0.003$; here $\mathbf{Z}$ follows a multivariate $t$-distribution with 1 degree of freedom, location parameter $\mathbf{0}_p$ and scale parameter $I_p$.
\item[Sp1.] Multivariate Gaussian: $\mathbf{X}\sim N\left(\mathbf{0}, I_2\right) + \lambda\cdot \mathbf{1}_2$; $n=200$.
\item[Sp2.] Multivariate $t$-distribution with 1 degree of freedom, location parameter $\lambda\cdot \mathbf{1}_2$ and scale parameter $I_2$; $n=200$.
\item[Sp3.] Uniform distribution over the unit disk $\{\mathbf{x}\in\R^2 : \|\mathbf{x}\|\leq 1\}$, plus $\lambda\cdot \mathbf{1}_2$; $n=200$.
\item[Sp4.] Elliptical: $n=200$, $p=2$, $\mathbf{X}=(2Z_1,Z_2)^\top$, where $Z_1$, $Z_2$ i.i.d.\ $N(0,1)$.
\item[Sp5.] Correlation: $n=200$, $p=2$, $\mathbf{X}\sim N\left(\mathbf{0},
\begin{pmatrix}
1&\rho\\
\rho&1\\
\end{pmatrix}\right)
$, where $\rho=0.6$.
\item[Sp6.] Non-spherical uniform distribution: $n=1000$, $p=2$, $\mathbf{X}$ follows the uniform distribution over $[-1,1]^2$.
\item[Sp7.] Chi-squared distribution: $n=100$, $p=2$, $X_1$, $X_2$ i.i.d.\ follow $\chi^2(1)-1$.
\item[Sp8.] High-dimensional (${\rm H}_0$): $n=200$, $p=50$, $\mathbf{X}\sim N(\mathbf{0},I_p)$.
\item[Sp9.] High-dimensional (${\rm H}_1$): $n=200$, $p=50$, $\mathbf{X}\sim N(\mathbf{0},I_p) + 0.05\cdot \mathbf{1}_p$.
\item[Sp10.] High-dim (${\rm H}_1$, heavy-tailed): $n=200$, $p=50$, $\mathbf{X}$ follows the multivariate $t$-distribution with 1 degree of freedom, location parameter $0.05\cdot \mathbf{1}_p$ and scale parameter $I_p$.
\end{enumerate}

For central symmetry (C1-C10), we compare our GWSR test (denoted by ``OT-Wilcox"), our generalized sign test (denoted by ``OT-sign") and Hotelling's $T^2$ test with the depth-based test proposed in \citet{Dyckerhoff2015depth} (denoted by ``DLP"), and the test proposed in \citet{EG2016} (denoted by ``EG") which compares empirical measures of opposite regions. For EG, we use the same parameter settings as in~\citet[Section 3]{EG2016}. 

For sign symmetry (S1-S10), we compare our proposals --- OT-Wilcox and OT-sign --- with Hotelling's $T^2$ test, the sign-change version of {\it spatial sign} test (denoted by ``SS") \citep[Section 6.1.2]{Oja2010nonp} and the sign-change version of the {\it spatial signed-rank} test (denoted by ``SSR") \citep[Section 7.1]{Oja2010nonp}. The last two tests are implemented in the \texttt{R} package ``MNM" \citep{MNM2011}.

For spherical symmetry (Sp1-Sp10), we compare our methods with \citet{Baringhaus1991} (denoted by ``LB"), a test for multivariate spherical symmetry that is consistent against any fixed alternative.
For LB \citep{Baringhaus1991}, a function $h$ needs to be specified. When $p=2$, we use the function $h(t)=\frac{t-1/4}{17/8-t}$ as in \citet[Section 4]{Baringhaus1991}, which yields a tractable  asymptotic null distribution.
When $p=50$, we use the function $h(t)=(1-2tw+w^2)^{-\lambda}-1$ \citep[Equation 3.10]{Baringhaus1991} with $\lambda=24$ and $w=1/4$ as in \citet[Section 5]{Baringhaus1991}, and
the same first-order approximation to the upper tail probability for computing the critical value, as advocated in \citet[Equation 4.4]{Baringhaus1991}.
We also compare our methods with \citet{Henze2014sph} (denoted by ``HHM"), a test for spherical symmetry based on the empirical characteristic function. We use the suggested settings in \citet[Section 5]{Henze2014sph}.\footnote{More specifically, we use the Kolmogorov-Smirnov type statistic with the hyper-parameter $\mathcal{R}=2$. For $p=2$, we use 8 rings and 9 grid points as in \citet[Section 5]{Henze2014sph}. For $p=50$, we use 8 rings and 200 grid points).}

The nominal size of all the testing procedures is set at  $\alpha =0.05$. Tables~\ref{tab:cen}, \ref{tab:sign} and~\ref{tab:sph} show the empirical power of the different methods for testing central, sign and spherical symmetries, respectively. For our methods --- OT-Wilcox and OT-sign --- and Hotelling's $T^2$, the empirical power is averaged over $10^4$ replications. For other competing methods, the empirical power is based on 1000 replications, as some of them are computationally intensive.

\begin{table}
\centering
\caption{Power of the competing methods for testing central symmetry, with nominal level 0.05. For OT-based methods, two choices of $\nu_n$ are used (random sampling/[Halton transformed]).}
\label{tab:cen}
\begin{tabular}{c||ccccc}
 \toprule
 \multicolumn{6}{c}{Multivariate Gaussian Distribution} \\
 \hline
C1& $T^2$ & OT-Wilcox & OT-sign &DLP \citep{Dyckerhoff2015depth}&EG \citep{EG2016}\\
 \hline
$\lambda = 0.0$   & 0.06    & 0.05/[0.04]&  0.06/[0.06]  & 0.05 &0.04\\
$\lambda = 0.1$   & 0.15    & 0.15/[0.13]& 0.11/[0.11]  & 0.06  & 0.14 \\
$\lambda = 0.2$   & 0.49    & 0.46/[0.42]&  0.29/[0.29] & 0.08 &  0.46\\
$\lambda = 0.3$   & 0.84    & 0.82/[0.80]&  0.55/[0.54]  &  0.12& 0.85 \\
$\lambda = 0.4$   & 0.98    & 0.98/[0.97]& 0.78/[0.77]  & 0.20 & 0.98 \\
\midrule
 \multicolumn{6}{c}{Multivariate $t$-Distribution} \\
 \hline
C2& $T^2$ & OT-Wilcox & OT-sign &DLP \citep{Dyckerhoff2015depth}&EG \citep{EG2016}\\
 \hline
$\lambda = 0.0$ & 0.03 & 0.06/[0.05] & 0.05/[0.05] & 0.05  & 0.04  \\
$\lambda = 0.2$ & 0.03 & 0.15/[0.13] & 0.19/[0.18] & 0.07  &  0.21 \\
$\lambda = 0.4$ & 0.04 & 0.48/[0.48] & 0.54/[0.54] & 0.12  & 0.62  \\
$\lambda = 0.6$ & 0.08 & 0.81/[0.79] & 0.86/[0.87] & 0.26  &0.94   \\
$\lambda = 0.8$ & 0.12 & 0.96/[0.96] & 0.98/[0.97] & 0.49  & 1.00  \\
\midrule
\multicolumn{6}{c}{Multivariate Uniform Distribution} \\
\hline
C3& $T^2$ & OT-Wilcox & OT-sign &DLP \citep{Dyckerhoff2015depth}&EG \citep{EG2016}\\
\hline
$\lambda = 0.00$ & 0.05 & 0.05/[0.04] & 0.05/[0.05] & 0.05  & 0.05  \\
$\lambda = 0.03$ & 0.10 & 0.22/[0.18] & 0.09/[0.09] & 0.08  & 0.13  \\
$\lambda = 0.06$ & 0.27 & 0.64/[0.57] & 0.19/[0.18] & 0.15  & 0.45  \\
$\lambda = 0.09$ & 0.55 & 0.91/[0.88] & 0.35/[0.34] & 0.24  & 0.80  \\
$\lambda = 0.12$ & 0.81 & 0.99/[0.98] & 0.54/[0.52] & 0.37  & 0.97 \\
\midrule
\multicolumn{6}{c}{Other  Alternatives} \\
\hline
C4-C7& $T^2$ & OT-Wilcox & OT-sign &DLP \citep{Dyckerhoff2015depth}&EG \citep{EG2016}\\
\hline
Exponential& 0.07 & 0.10/[0.09] & 0.63/[0.66] & 0.84 &0.88\\
Exponential (L) & 0.06 & 0.15/[0.15] & 0.91/[0.93] & 1.00 & 1.00\\
Chi-squared & 0.09 & 0.21/[0.19] & 0.83/[0.88] & 0.99 &1.00\\
Pareto & 0.84 & 0.04/[0.02] & 0.90/[0.96]& 1.00  &1.00 \\
\midrule
\multicolumn{6}{c}{High-Dimensional Setting} \\
\hline
C8-C10& $T^2$ & OT-Wilcox & OT-sign &DLP \citep{Dyckerhoff2015depth}&EG \citep{EG2016}\\
\hline
AR$(1)$ (${\rm H}_0$) &0.05 & 0.05/[0.05] & 0.04/[0.04] & --- & ---\\
AR$(1)$ (${\rm H}_1$) & 0.97 & 0.41/[0.99] & 0.07/[0.16] & --- & ---\\
AR$(1)$ (${\rm H}_1$, heavy-tailed) & 0.35 & 0.24/[0.89] & 0.06/[0.23] & --- & ---\\
\bottomrule
\end{tabular}
\end{table}

\begin{table}
\centering
\caption{Power of the competing methods for testing sign symmetry, with nominal level 0.05. For OT-based methods, two choices of $\nu_n$ are used (random sampling/[Halton transformed]).}
\label{tab:sign}
\begin{tabular}{c||ccccc}
 \toprule
 \multicolumn{6}{c}{Multivariate Gaussian Distribution} \\
 \hline
S1& $T^2$ & OT-Wilcox & OT-sign & SSR \citep{Oja2010nonp}& SS \citep{Oja2010nonp}\\
 \hline
$\lambda = 0.0$ & 0.06  & 0.05/[0.04]   & 0.05/[0.05]  & 0.05 & 0.05  \\
$\lambda = 0.1$ & 0.21  & 0.20/[0.18]   & 0.14/[0.14]  & 0.19 & 0.16  \\
$\lambda = 0.2$ & 0.64  & 0.62/[0.60]   & 0.43/[0.43]  & 0.61 & 0.52  \\
$\lambda = 0.3$ & 0.94  & 0.93/[0.93]   & 0.79/[0.79]  & 0.94 & 0.87  \\
$\lambda = 0.4$ & 1.00  & 1.00/[1.00]   & 0.97/[0.97]  & 1.00 & 0.99  \\
\midrule
\multicolumn{6}{c}{Multivariate $t$ Distribution} \\
 \hline
S2& $T^2$ & OT-Wilcox & OT-sign & SSR \citep{Oja2010nonp}& SS \citep{Oja2010nonp}\\
 \hline
$\lambda = 0.0$   & 0.03   & 0.06/[0.05]    & 0.04/[0.04] & 0.05 & 0.05\\
$\lambda = 0.1$   & 0.04   & 0.10/[0.09]    & 0.12/[0.12] & 0.09 & 0.12\\
$\lambda = 0.2$   & 0.05   & 0.22/[0.20]    & 0.29/[0.29] & 0.24 & 0.34\\
$\lambda = 0.3$   & 0.06   & 0.44/[0.43]    & 0.58/[0.58] & 0.48 & 0.66\\
$\lambda = 0.4$   & 0.06   & 0.66/[0.65]    & 0.83/[0.83] & 0.72 & 0.89\\
\midrule
\multicolumn{6}{c}{Multivariate Uniform Distribution} \\
 \hline
S3& $T^2$ & OT-Wilcox & OT-sign & SSR \citep{Oja2010nonp}& SS \citep{Oja2010nonp}\\
 \hline
$\lambda = .000$   & 0.05   & 0.05/[0.04]  & 0.05/[0.05] & 0.06 & 0.05\\
$\lambda = .025$   & 0.12   & 0.16/[0.14]  & 0.07/[0.07] & 0.10 & 0.08\\
$\lambda = .050$   & 0.33   & 0.51/[0.46]  & 0.13/[0.13] & 0.29 & 0.18\\
$\lambda = .075$   & 0.65   & 0.83/[0.80]  & 0.25/[0.25] & 0.58 & 0.37\\
$\lambda = .100$   & 0.89   & 0.97/[0.96]  & 0.42/[0.42] & 0.83 & 0.60\\
\midrule
\multicolumn{6}{c}{Other Alternatives} \\
\hline
S4-S7 & $T^2$ & OT-Wilcox & OT-sign & SSR \citep{Oja2010nonp}& SS \citep{Oja2010nonp}\\
\hline
Correlation &0.06 &0.06/[0.05] &0.05/[0.05] & 0.05 & 0.05\\
Exponential & 0.07 & 0.14/[0.12] & 0.93/[0.93] & 0.34 & 0.79\\
Chi-squared & 0.09 & 0.30/[0.28] & 1.00/[1.00] & 0.59 & 0.98\\
Laplace & 0.73&0.81/[0.79] &0.92/[0.92] & 0.83 & 0.88\\
\midrule
\multicolumn{6}{c}{High-Dimensional Setting} \\
\hline
S8-S10 & $T^2$ & OT-Wilcox & OT-sign & SSR \citep{Oja2010nonp}& SS \citep{Oja2010nonp}\\
\hline
High-dim (${\rm H}_0$) & 0.04 & 0.05/[0.05] & 0.05/[0.05] & 0.05 & 0.05\\
High-dim (${\rm H}_1$) & 0.65 & 0.30/[0.28] & 0.49/[0.49] & 0.63 & 0.62\\
High-dim (${\rm H}_1$, heavy-tailed) & 0.05 & 0.17/[0.13] & 0.28/[0.28] &0.27 &0.35\\
\bottomrule
\end{tabular}
\end{table}

\begin{table}
\centering
\caption{Power of the competing methods for testing spherical symmetry, with nominal level 0.05. For OT-based methods, two choices of $\nu_n$ are used (random sampling/[Halton transformed]).}
\label{tab:sph}
\begin{tabular}{c||ccccc}
 \toprule
 \multicolumn{6}{c}{Multivariate Gaussian Distribution} \\
 \hline
Sp1& $T^2$ & OT-Wilcox & OT-sign &LB \citep{Baringhaus1991} & HHM \citep{Henze2014sph}\\
 \hline
$\lambda = 0.00$ & 0.06  & 0.05/[0.05] & 0.05/[0.06] & 0.05&0.05 \\
$\lambda = 0.05$ & 0.15  & 0.14/[0.14] & 0.07/[0.06] & 0.08& 0.10 \\
$\lambda = 0.10$ & 0.42  & 0.42/[0.42] & 0.12/[0.12] & 0.18&0.26 \\
$\lambda = 0.15$ & 0.77  & 0.76/[0.76] & 0.21/[0.22] & 0.38&0.54 \\
$\lambda = 0.20$ & 0.95  & 0.95/[0.95] & 0.36/[0.37] & 0.62& 0.84\\
\midrule
 \multicolumn{6}{c}{Multivariate $t$ Distribution} \\
 \hline
Sp2& $T^2$ & OT-Wilcox & OT-sign &LB \citep{Baringhaus1991}& HHM \citep{Henze2014sph}\\
 \hline
$\lambda = 0.0$ & 0.03  & 0.07/[0.06] & 0.05/[0.06] & 0.04& 0.05\\
$\lambda = 0.1$ & 0.05  & 0.16/[0.15] & 0.09/[0.09] & 0.21& 0.19\\
$\lambda = 0.2$ & 0.06  & 0.46/[0.45] & 0.20/[0.20] & 0.75& 0.60\\
$\lambda = 0.3$ & 0.07  & 0.79/[0.78] & 0.39/[0.39] & 0.97& 0.94\\
$\lambda = 0.4$ & 0.10  & 0.96/[0.96] & 0.62/[0.62] & 1.00& 0.99\\
\midrule
\multicolumn{6}{c}{Multivariate Uniform Distribution} \\
 \hline
Sp3& $T^2$ & OT-Wilcox & OT-sign &LB \citep{Baringhaus1991}&HHM \citep{Henze2014sph}\\
 \hline
$\lambda = .000$ & 0.06  & 0.05/[0.04] & 0.05/[0.05] & 0.04& 0.05\\
$\lambda = .025$ & 0.14  & 0.21/[0.19] & 0.06/[0.06] & 0.04& 0.09\\
$\lambda = .050$ & 0.40  & 0.61/[0.59] & 0.09/[0.09] & 0.05& 0.25\\
$\lambda = .075$ & 0.75  & 0.90/[0.90] & 0.14/[0.15] & 0.07& 0.55\\
$\lambda = .100$ & 0.96  & 0.99/[0.99] & 0.23/[0.23] & 0.11& 0.83\\
\midrule
\multicolumn{6}{c}{Other Alternatives} \\
 \hline 
Sp4-Sp7 & $T^2$ & OT-Wilcox & OT-sign &LB \citep{Baringhaus1991}&HHM \citep{Henze2014sph}\\
\hline 
Elliptical & 0.06 & 0.06/[0.06]&0.05/[0.05]&0.50 & 1.00\\
Correlation &0.06 &0.06/[0.06] &0.05/[0.05] &0.50 & 1.00\\
Non-spherical uniform &0.05 & 0.05/[0.05] &0.05/[0.05]  &0.07 & 0.05\\
Chi-squared & 0.09 & 0.27/[0.25] &0.45/[0.45] &1.00 & 1.00\\
\midrule
\multicolumn{6}{c}{High-Dimensional Setting} \\
\hline
Sp8-Sp10 & $T^2$ & OT-Wilcox & OT-sign &LB \citep{Baringhaus1991}&HHM \citep{Henze2014sph}\\
\hline
High-dim (${\rm H}_0$) &0.04 &0.05/[0.05] &0.05/[0.05] &0.00 &0.04\\
High-dim (${\rm H}_1$) &0.54 &0.68/[0.68] &0.05/[0.05] &0.00 &0.15\\
High-dim (${\rm H}_1$, heavy-tailed) & 0.05 & 0.35/[0.35] & 0.05/[0.05]& 0.00 & 0.11\\
\bottomrule
\end{tabular}
\end{table}

The first three examples (C1-C3, S1-S3, Sp1-Sp3) are location shift models, a class of alternatives where OT-Wilcox is powerful. When $\lambda=0$, these distributions are (centrally/sign/spherically) symmetric, and
it can be seen from Tables~\ref{tab:cen}-\ref{tab:sph} that all methods have the correct size around 0.05.
When $\lambda >0$,
in the Gaussian cases (C1, S1, Sph1),  Hotelling's $T^2$ test works the best,
and OT-Wilcox is only slightly inferior to it, which is a consequence of \cref{prop:Gaussian} (or Theorems~\ref{thm:indepcomp},~\ref{thm:ellip}). In the $t$-distribution example (C2, S2, Sp2), Hotelling's $T^2$ test has almost no power, while OT-Wilcox still maintains high power. The power of the OT-sign test is even higher than that of OT-Wilcox in settings C2 and S2.
In the light-tailed uniform distribution case (settings C3, S3, Sp3), Hotelling's $T^2$ regains some power, but it is still inferior to OT-Wilcox, which achieves the highest power among all comparing methods.

The above scenarios C1-C3, S1-S3 and Sp1-Sp3 show that our OT-Wilcox is a nonparametric competitor of Hotelling's $T^2$ test, and is powerful in detecting location shifts of a symmetric distribution. However, our methods may be less powerful against non-location shift alternatives, such as a skewed alternative with mean or median 0.
We illustrate this using examples C4-C7, S4-S6, and Sph4-Sph7, where
the power of OT-Wilcox is lower compared to OT-sign and most of the other competing methods, though it is still higher than Hotelling's $T^2$ test in general.
Note that OT-Wilcox being powerless for scenarios S4 and Sp4-Sp6 may be explained by~\cref{prop:consistency_signed_rank}:  when testing sign and spherical symmetry, OT-Wilcox will not be consistent against centrally symmetric alternatives because $\mathbb{E}[\nabla\psi(\mathbf{X})] = \mathbf{0}$ (with $J(\mathbf{x})=\mathbf{x}$), a consequence of $\nabla\psi(-\mathbf{x}) = -\nabla\psi(\mathbf{x})$.

C8, S8 and Sp8 correspond to high-dimensional settings under ${\rm H}_0$.
It can be seen that both our methods and Hotelling's $T^2$ have size close to 0.05, as expected. Under ${\rm H}_1$ (i.e., settings C9, S9, Sp9), OT-Wilcox has nontrivial power, but is not competitive compared to Hotelling's $T^2$ (see settings C9 and S9). One issue here is that when the dimension $p$ is high, it can be difficult for $\nu_n$ to approximate $\nu$ well. For example, for central symmetry, $\nu$ is supported on $2^{p-1}$ orthants, but at most $n$ of the orthants have point(s) from $\nu_n$. 
Thus, from setting C8 in \cref{tab:cen}, OT-Wilcox with $\nu_n$ obtained from random sampling is seen to be less powerful than $T^2$. This may be improved by using $\nu_n$ obtained from the quasi-Monte Carlo (Halton transformed) points, which greatly boosts the power of OT-Wilcox (see C8 in \cref{tab:cen}). However, Halton transformed $\nu_n$ does not lead to significant power improvements for sign symmetry (S8), spherical symmetry (Sp8), or in any of the low-dimensional examples, where $\nu$ is easier to approximate.
In the heavy-tailed high-dimensional settings (i.e., C10, S10, Sp10), OT-Wilcox is clearly more powerful (and robust) than Hotelling's $T^2$.

Let us now compare the performance of the other methods mentioned above. Note that DLP \citep{Dyckerhoff2015depth} and EG \citep{EG2016} are designed for testing bivariate central symmetry, so they cannot be used in high-dimensional settings. LB \citep{Baringhaus1991} is powerless in high-dimensional scenarios.\footnote{Here the first-order approximation to the upper tail probability is no longer accurate --- the infinite product in \citep[Equation 4.4]{Baringhaus1991}, which is only practically computable for small $p$, becomes astronomically large when $p=50$.}
HHM \citep{Henze2014sph} achieves relatively good performance for settings Sp4, Sp5, Sp7, but is less satisfactory for the high-dimensional cases Sp8-Sp10.
SSR \citep{Oja2010nonp} and SS \citep{Oja2010nonp} perform quite well in testing sign symmetry, except for location shifts of the uniform distribution (i.e., setting S3), where they have even lower power than Hotelling's $T^2$.

In example S7 the ARE of our generalized sign test w.r.t.\ Hotelling's $T^2$ is 2 (see~\cref{table:AREsign});~\cref{tab:sign} (setting S7) shows that OT-sign does have higher power than Hotelling's $T^2$.

\section{Distribution-Free Confidence Sets}\label{subsec:conf_set}
We say that $\mathbf{X} \sim \mu \in \mathcal{P}_{\rm ac}(\R^p)$ is $\G$-symmetric around $\B \theta^* \in \R^p$ if $\mathbf{X}-\B \theta^* \overset{d}{=} Q(\mathbf{X}-\B \theta^*)$ for all $Q\in \G$. Here $\B \theta^*$ is called the center of symmetry for the distribution $\mu$. As a direct consequence of our results, the GWSR test (or the generalized sign test) can be used to obtain a distribution-free confidence set for $\B \theta^*$. The idea is just to invert the collection of hypotheses $\mathrm{H}_{0,\B \theta}$ defined as: $$\mathrm{H}_{0,\B \theta}: (\mathbf{X} - \B \theta) \stackrel{d}{=} Q(\mathbf{X} - \B \theta), \quad \mbox{ for all }\;Q \in \G,$$ for $\B \theta \in \R^p$, using the distribution-free (see Proposition~\ref{prop:properties_of_generalized_sign_rank}) GWSR test with data $\{\mathbf{X}_i - \B \theta\}_{i=1}^n$ at level $\alpha \in (0,1)$. This yields an exact $(1-\alpha)$ distribution-free confidence set for $\B \theta^*$:
$$\mathcal{C}:= \left\{\B \theta \in \R^p: \mathrm{H}_{0,\B \theta} \mbox{ is  accepted using data } \{\mathbf{X}_i - \B \theta\}_{i=1}^n \right\},$$ 
without making any distributional assumptions on $\mathbf{X} \sim \mu \in \mathcal{P}_{\rm ac}(\mathbb{R}^p)$ except the assumption of $\G$-symmetry. For implementational purposes, we recommend using a fine grid of $\B \theta$'s to test $\mathrm{H}_{0,\B \theta}$ and approximate $\mathcal{C}$. For $p \ge 2$, an alternative to `gridding' (which may be computationally infeasible) is to test the hypothesis of $\G$-symmetry $\mathrm{H}_{0,\B \theta}$ at $\B \theta = \mathbf{X}_i$, for each $i=1,\ldots, n$, and then to consider the (possibly larger) convex hull of the accepted $\B \theta$'s.  

\cref{fig:Conf_set} shows the distribution-free 95\% confidence sets (colored in yellow; obtained via a fine gridding of $\B \theta \in \R^p$) using the GWSR test when $p=2$ with the Gaussian ERD and the identity score function, when: (i) $\G = \{-I,I\}$ (corresponding to central symmetry), and (ii) $\G = O(p)$ (corresponding to spherical symmetry);  see~\cref{rk:ref_mea} for our choice of $\nu$ (here we used a random sample from $\nu$ to obtain $\nu_n$ as in~\cref{subsec:power_compare}). 
The first two rows show data points drawn from the Gaussian distribution $N(0.5\cdot \mathbf{1}_2, I_2)$; the third row depicts data points obtained from the multivariate $t$-distribution with 1 degree of freedom with location parameter $0.5\cdot \mathbf{1}_2$ and scale parameter $I_2$.
In each example, the true center of symmetry $\B \theta^*$ is marked by ``{\color{red}$+$}".

It can be seen that our GWSR based confidence sets successfully cover the true unknown center  $\B \theta^*$ in the depicted cases. Moreover, for testing spherical symmetry with Gaussian data (see the first row of~\cref{fig:Conf_set}), the two confidence sets --- our proposal based on the GWSR test and the one obtained from Hotelling's $T^2$ test --- are almost indistinguishable even for sample size $n=20$. For testing central symmetry (see the second row of~\cref{fig:Conf_set}) using Gaussian data we observe a similar phenomenon, although for $n=20$ the confidence set  based on the Wilcoxon signed-rank test has a more jagged boundary.

When the Gaussian assumption is violated and we have heavy-tailed observations (see the third row of~\cref{fig:Conf_set}) the confidence sets obtained from the Hotelling's $T^2$ test get severely inflated. However, our distribution-free confidence sets are still quite small and reliable.

\section*{Acknowledgements}
The authors would like to thank Nabarun Deb for many helpful discussions.

\begin{figure}
    \centering
    \includegraphics[width = 1\textwidth]{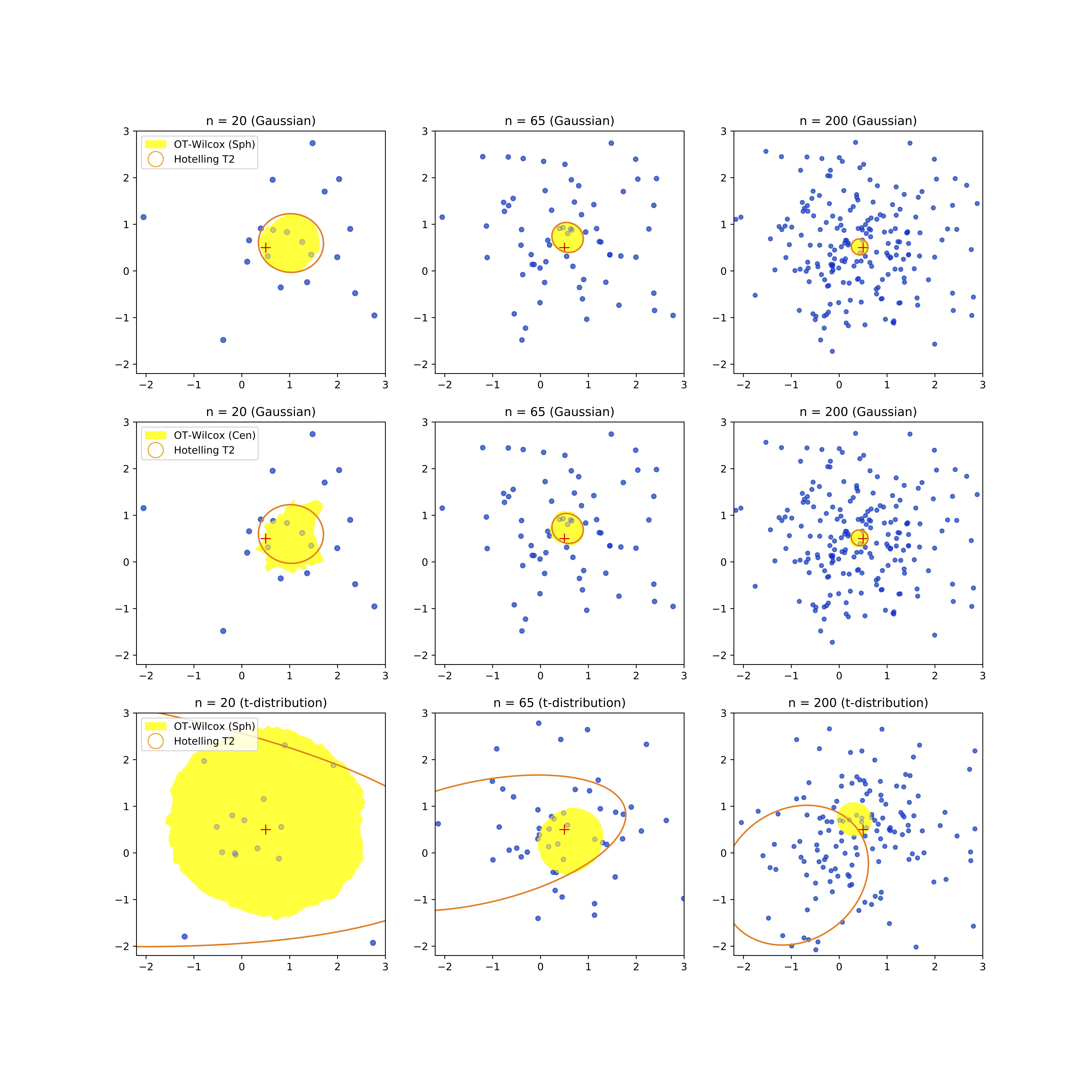}
    \caption{The confidence sets obtained by inverting OT-Wilcox (for central symmetry and spherical symmetry) and the Hotelling's $T^2$ test. The true center of the underlying Gaussian distribution is marked by ``{\color{red}$+$}". The data distribution is either Gaussian (the first two rows) or the $t$-distribution with 1 degree of freedom.}
    \label{fig:Conf_set}
\end{figure}

\appendix
\section{Some Discussions}\label{sec:discussions}
\subsection{Regularity Conditions}\label{sec:regu}
In this section, we describe the regularity conditions that are assumed in \cref{thm:asymptotics_local_alternatives}, in order to study the power performance against local contiguous alternatives.
We adopt the standard smooth parametric model assumptions that guarantee local asymptotic normality (LAN) (see e.g., \citet[Chapter 7]{vanderVaart1998}).
Consider the parametric family consisting of all locational shifts of a symmetric distribution $f$, i.e.,
$\{\mathcal{P}_{\B \theta} \}_{{\B \theta}\in\R^p}= \{f(\cdot -{\B \theta}):{\B \theta} \in \R^p\}$. Assume:
\begin{enumerate}
\item The family $\{\mathcal{P}_{\B \theta} \}_{{\B \theta}\in\R^p}$ is {\it quadratic mean differentiable} (QMD) at ${\B \theta}_0 = \mathbf{0}$ (see \cite[Definition 12.2.1]{Lehmann2005hypo} for relevant definitions).
\item $\mathbb{E}_{\mathbf{X}\sim \mathcal{P}_{{\B \theta}_0}} \left( \|\eta (\mathbf{X},{\B \theta}_0)\|^2\right)<\infty$, where $\eta(\cdot,{\B \theta}):=\frac{\nabla f(\cdot - {\B \theta})}{f(\cdot - {\B \theta})}$ is the {\it score function}. Also assume that the Fisher information at ${\B \theta}_0$, that is, $I_{{\B \theta}_0}:= \mathbb{E}_{\mathbf{X}\sim \mathcal{P}_{{\B \theta}_0}} \left[ \eta(\mathbf{X},{\B \theta})\eta(\mathbf{X},{\B \theta})^\top\right]$ exists and is invertible.
\end{enumerate}

\subsection{Generating a Random Sign for Spherical Symmetry}\label{sec:gen_sign_sph}
To show that the procedure described in~\cref{rem:Sp-Sym} in \cref{subsec:consist_efficiency} generates a sample from the conditional distribution of ${\rm Uniform(O}(p))$ given that it maps $\mathbf{v}$ to $\mathbf{w}$,
we note that a uniform distribution over ${\rm O}(p)$ can be generated in the following way \citep[Example 4.1]{Eaton1989group}:
\begin{enumerate}
\item First draw $\mathbf{\omega}$ from the uniform distribution over the sphere $\{\mathbf{x}\in\R^p:\|\mathbf{x}\|=1\}$.
\item Then draw $p(p-1)$ i.i.d.\ standard normal random variables $\varepsilon\in\R^{p(p-1)}$ independent of $\mathbf{\omega}$.
\item Apply Gram-Schmidt orthogonalization to the columns of $[\mathbf{\omega}\ \varepsilon] $ in the order of the columns. This yields an orthogonal matrix ${\rm GS}([\mathbf{\omega}\ \varepsilon])$ following the uniform distribution over ${\rm O}(p)$ \citep[Example 4.1]{Eaton1989group}.
\item Define $f(\mathbf{\omega},\varepsilon):= {\rm GS}([\mathbf{\omega}\ \varepsilon])[\mathbf{v}\ V]^\top$ which also follows the uniform distribution over ${\rm O}(p)$ (as $[\mathbf{v}\ V] \in {\rm O}(p)$), by the definition of the Haar measure (invariant under group action).
\end{enumerate}
Note that the set $\argmin_{Q\in \G} \|Q^\top\mathbf{x} - \mathbf{h}\|^2$ is described by $\{Q\in {\rm O}(p):\mathbf{w}=Q\mathbf{v}\}$.
If $f(\omega,\varepsilon)\mathbf{v}=\mathbf{w}$, we know that this condition is equivalent to ${\rm GS}([\mathbf{\omega}\ \varepsilon]) \mathbf{e}_1 = \mathbf{w}$, where $\mathbf{e}_1 = (1,0,\ldots,0)^\top$, and it is further equivalent to $\omega = \mathbf{w}$. Therefore, by the independence of $\omega$ and $\varepsilon$, given $\omega = \mathbf{w}$, the conditional distribution of $f(\omega,\varepsilon)$ is given by the unconditional distribution of $f(\mathbf{w},\varepsilon)$. Thus the validity of the procedure is shown.
\qed \\

\noindent\textbf{A second way to generate the conditional distribution.} From
$$\{Q\in \mathrm{O}(p):\;\mathbf{w}=Q\mathbf{v}\}=\{\mathbf{w}\mathbf{v}^\top+WUV^\top:\;U\in \mathrm{O}(p-1)\},$$
where $V$ and $W$ are $p\times (p-1)$ matrices such that $V^\top V=W^\top W=I_{p-1}$, $V^\top \mathbf{v}=W^\top \mathbf{w} = \mathbf{0}$,
we can first generate $U$ from $\mathrm{O}(p-1)$ and then $\mathbf{w}\mathbf{v}^\top+WUV^\top$ has the desired conditional distribution. The validity of this procedure can be shown as follows:
\begin{enumerate}
    \item Draw $p^2$ i.i.d.\ standard normal variables $\varepsilon \in \R^{p\times p}$ to form the $p\times p$ matrix $[\varepsilon_{ij}]$.
    \item Apply Gram-Schmidt orthogonalization to the columns of $[\varepsilon_{ij}] $ in the order of the columns. This yields an orthogonal matrix ${\rm GS}([ \varepsilon_{ij}])$ that follows the uniform distribution over ${\rm O}(p)$ \citep[Example 4.1]{Eaton1989group}.
    \item Define $f(\varepsilon) := [\mathbf{w}\ W]{\rm GS}([\varepsilon_{ij}])[\mathbf{v}\ V]^\top$ which also follows the uniform distribution over ${\rm O}(p)$ by the invariance property of the Haar measure.
\end{enumerate}
If $f(\varepsilon)\mathbf{v}=\mathbf{w}$, we know that this condition is equivalent to saying that the first column of ${\rm GS}([\varepsilon_{ij}])$ is $\mathbf{e}_1$. From the Gram-Schmidt procedure we know that this is further equivalent to the event $\{\varepsilon_{j1} = 0,{\ \rm for\ } j=2,3,\ldots,p,\ {\rm and}\ \varepsilon_{11} > 0\}$ which only depends on the first column of $[\varepsilon_{ij}] $. By independence, given this condition, other $\varepsilon_{ij}$ from the second column to the last column are still i.i.d.\ standard normal. Hence this implies the bottom right $(p-1)\times(p-1)$ sub-matrix of ${\rm GS}([\varepsilon_{ij}])$ (denoted by $U$) conditionally follows the uniform distribution over ${\rm O}(p-1)$, and ${\rm GS}([\varepsilon_{ij}])=
\begin{pmatrix}
1&\mathbf{0}^\top\\
\mathbf{0} & U\\
\end{pmatrix}.$
Hence the conditional distribution of the uniform distribution
$f(\varepsilon)$ given $\{f(\varepsilon)\mathbf{v}=\mathbf{w}\}$ has the distribution of $[\mathbf{w}\ W]\begin{pmatrix}
1&\mathbf{0}^\top\\
\mathbf{0} & U\\
\end{pmatrix}[\mathbf{v}\ V]^\top=\mathbf{w}\mathbf{v}^\top+WUV^\top$
with $U\sim {\rm Uniform}({\rm O}(p-1))$.
Therefore, the validity of the second generating procedure is shown.\qed

\section{Proofs}\label{sec:proofs}
\subsection{Proof of \cref{prop:sample_uniq}}\label{pf:uniq_sample}
We first show the uniqueness of $R_n(\cdot)$. If the uniqueness fails, then necessarily, there exist two different (non-random) permutations $\sigma_1$ and $\sigma_2$ such that
$\sum_{i=1}^n c(\mathbf{X}_i,\mathbf{h}_{\sigma_1(i)}) = \sum_{i=1}^n c(\mathbf{X}_i,\mathbf{h}_{\sigma_2(i)})$.
Without loss of generality, suppose $\sigma_1(1)\neq \sigma_2(1)$.
Then
\begin{equation}\label{eq:reorganize_c}
    c(\mathbf{X}_1,\mathbf{h}_{\sigma_1(1)}) - c(\mathbf{X}_1,\mathbf{h}_{\sigma_2(1)}) = \sum_{i=2}^n \left(c(\mathbf{X}_i,\mathbf{h}_{\sigma_2(i)}) - c(\mathbf{X}_i,\mathbf{h}_{\sigma_1(i)})\right).
\end{equation}
Define the distance from a point $\mathbf{x}\in\R^p$ to a set $A\subset\R^p$ as $d(\mathbf{x},A) :=\inf_{\mathbf{y}\in A}\|\mathbf{x}-\mathbf{y}\|$.
Note that $c(\mathbf{X}_1,\mathbf{h}_{\sigma_1(1)}) - c(\mathbf{X}_1,\mathbf{h}_{\sigma_2(1)}) = d(\mathbf{X}_1,\G \mathbf{h}_{\sigma_1(1)})^2 - d(\mathbf{X}_1,\G \mathbf{h}_{\sigma_2(1)})^2$.
The following lemma shows that as long as the orbits $\G \mathbf{h}_{\sigma_1(1)}\neq \G \mathbf{h}_{\sigma_2(1)}$ are different,
$$\{\mathbf{x}\in\R^p: d(\mathbf{x},\G \mathbf{h}_{\sigma_1(1)})^2 - d(\mathbf{x},\G \mathbf{h}_{\sigma_2(1)})^2 =C\}$$
has Lebesgue measure 0 for any constant $C$. This implies
\eqref{eq:reorganize_c} holds with probability 0 by considering the conditional distribution of $\mathbf{X}_1$ given $\mathbf{X}_2,\ldots,\mathbf{X}_n$. Hence, $\hat{\sigma}$ and thus $R_n(\cdot)$ are unique.
\begin{lemma}\label{lem:D_mea_zero}
Let $A,B\subset\R^p$ be two disjoint compact sets with Lebesgue measure 0, and $C$ be a constant. Then
$$D:=\{\mathbf{x}\in\R^p: d(\mathbf{x},A)^2 - d(\mathbf{x},B)^2 =C\}$$
has Lebesgue measure 0.
\end{lemma}
\begin{proof}
It is known that the distance function $f(\mathbf{x}):= d(\mathbf{x},A)$ is 1-Lipschitz and is thus Lebesgue a.e.\ differentiable (by Rademacher's theorem).
If $\mathbf{x}_0 \notin A$ is a differentiable point of $f$ then its projection onto $A$ must be unique\footnote{To make this uniqueness argument precise, suppose that $\mathbf{y}_0$ (a priori not unique) attains the minimal distance from $\mathbf{x}_0$ to $A$, i.e., $f(\mathbf{x}_0)=\|\mathbf{x}_0-\mathbf{y}_0\|$. Then $\mathbf{y}_0$ also attains the minimal distance from $A$ to any point in the line segment $[\mathbf{x}_0,\mathbf{y}_0]$. This is because, for any $\mathbf{y}' \in A$, $t\in[0,1]$, and $\mathbf{z}:=\mathbf{x}_0 + t(\mathbf{y}_0 - \mathbf{x}_0)$, we have by triangle inequality,
$$\| \mathbf{z} - \mathbf{y}'\|\geq \| \mathbf{x}_0 - \mathbf{y}'\| - \| \mathbf{x}_0 - \mathbf{z}\|\geq f(\mathbf{x}_0) - \|t(\mathbf{y}_0 - \mathbf{x}_0)\| = (1-t) \|\mathbf{y}_0 - \mathbf{x}_0\| = \|\mathbf{z} -\mathbf{y}_0\|. $$
Therefore, $f(\mathbf{z}) = \|\mathbf{z}-\mathbf{y}_0\|$; further $f(\mathbf{z}) = f(\mathbf{x}_0) - (\mathbf{z}-\mathbf{x}_0)^\top \frac{\mathbf{y}_0 - \mathbf{x}_0}{\|\mathbf{y}_0 - \mathbf{x}_0\| }$.
As $f$ is differentiable at $\mathbf{x}_0$, for $\mathbf{z}$ close to $\mathbf{x}_0$, we have $f(\mathbf{z}) = f(\mathbf{x}_0) +(\mathbf{z}-\mathbf{x}_0)^\top \nabla f(\mathbf{x}_0) + o(\|\mathbf{z}-\mathbf{x}_0\|)$, which further yields $\nabla f(\mathbf{x}_0) = -\frac{\mathbf{y}_0 - \mathbf{x}_0}{\|\mathbf{y}_0 - \mathbf{x}_0\| }$. Hence, $\mathbf{y}_0$ is determined from $\mathbf{x}_0$ and the gradient $\nabla f(\mathbf{x}_0)$ as $\mathbf{y}_0=\mathbf{x}_0 - f(\mathbf{x}_0)\nabla f(\mathbf{x}_0)$, and hence $\mathbf{y}_0$ unique.}, otherwise there will be two directions where $f$ descends the steepest (with rate 1), which is a contradiction to the differentiability of $f$.
Moreover, $\nabla f(\mathbf{x}_0)$ is a unit vector on the line that connects $\mathbf{x}_0$ and its projection on $A$.

Now, take $\mathbf{x}_0 \in D$ to be a point that is not in $A$ or $B$, and is a differentiable point of both $d(\mathbf{x},A)$ and $d(\mathbf{x},B)$.
Then, for $\mathbf{x}$ in a neighborhood of $\mathbf{x}_0$, using the definition of differentiability, $d(\mathbf{x},A)^2 - d(\mathbf{x},B)^2 =C$ can be written as:
\begin{equation}\label{eq:first_order_app}
2\left[d(\mathbf{x}_0,A)\nabla d(\mathbf{x}_0,A) - d(\mathbf{x}_0,B)\nabla d(\mathbf{x}_0,B) \right]^\top (\mathbf{x}-\mathbf{x}_0) + o(\|\mathbf{x}-\mathbf{x}_0\|) = 0.
\end{equation}
Let $\mathbf{f} := 2\left[d(\mathbf{x}_0,A)\nabla d(\mathbf{x}_0,A) - d(\mathbf{x}_0,B)\nabla d(\mathbf{x}_0,B) \right] \in \R^p$. Then $\mathbf{f}$ is nonzero, otherwise the projections of $\mathbf{x}_0$ onto $A$ and $B$ are the same, which is a contradiction to $A$, $B$ being disjoint. Hence, \eqref{eq:first_order_app} implies the set $D$ is locally contained in a hyperplane with arbitrarily small thickness, i.e., for any $\delta > 0$, there exists $\varepsilon_0 > 0 $ such that for all  $\varepsilon<\varepsilon_0$,
$$D\cap B(\mathbf{x}_0,\varepsilon)\subset \{\mathbf{x}\in B(\mathbf{x}_0,\varepsilon): | \mathbf{f}^\top (\mathbf{x}-\mathbf{x}_0) |  \leq \delta \|\mathbf{x}-\mathbf{x}_0\|\} .$$
Here $B(\mathbf{x}_0,\varepsilon) := \{\mathbf{x}\in\R^p:\|\mathbf{x}-\mathbf{x}_0\|<\varepsilon\}$ is the open ball centered at $\mathbf{x}_0$ with radius $\varepsilon >0$.
Hence, with $m$ denoting the Lebesgue measure,
$$\begin{aligned}
\limsup_{\varepsilon \to 0} \frac{m(D\cap B(\mathbf{x}_0,\varepsilon))}{m( B(\mathbf{x}_0,\varepsilon))}\leq
\limsup_{\varepsilon \to 0} \frac{m(\{\mathbf{x}\in B(\mathbf{x}_0,\varepsilon): | \mathbf{f}^\top (\mathbf{x}-\mathbf{x}_0) |  \leq \delta \|\mathbf{x}-\mathbf{x}_0\|\})}{m( B(\mathbf{x}_0,\varepsilon))}.
\end{aligned}$$
By applying a translation and a rescaling, we observe that 
$$\frac{m(\{\mathbf{x}\in B(\mathbf{x}_0,\varepsilon): | \mathbf{f}^\top (\mathbf{x}-\mathbf{x}_0) |  \leq \delta \|\mathbf{x}-\mathbf{x}_0\|\})}{m( B(\mathbf{x}_0,\varepsilon))}=\frac{m(\{\mathbf{x}\in B(\mathbf{0},1): | \mathbf{f}^\top \mathbf{x} |  \leq \delta \|\mathbf{x}\|\})}{m( B(\mathbf{0},1))}$$
for all $\mathbf{x}_0\in\R^p$ and $\varepsilon > 0$.
Therefore,
$$\limsup_{\varepsilon \to 0} \frac{m(D\cap B(\mathbf{x}_0,\varepsilon))}{m( B(\mathbf{x}_0,\varepsilon))}\leq \frac{m(\{\mathbf{x}\in B(\mathbf{0},1): | \mathbf{f}^\top \mathbf{x} |  \leq \delta \|\mathbf{x}\|\})}{m( B(\mathbf{0},1))}.$$
Since $\delta$ is arbitrary (and the set on the right side of the above display converges to a subset of a hyperplane, as $\delta$ approaches 0), we further have that
$$\lim_{\varepsilon \to 0} \frac{m(D\cap B(\mathbf{x}_0,\varepsilon))}{m( B(\mathbf{x}_0,\varepsilon))} = 0.$$
But from the Lebesgue differentiation theorem, for Lebesgue a.e.-$\mathbf{x}_0$ in $D$, 
$$\lim_{\varepsilon \to 0} \frac{m(D\cap B(\mathbf{x}_0,\varepsilon))}{m( B(\mathbf{x}_0,\varepsilon))} = 1_D(\mathbf{x}_0) = 1.$$
Comparing the last two displays we see that $D$ must have Lebesgue measure 0.
\end{proof}

Next, we show the uniqueness of $S_n(\cdot)$. Since $R_n(\mathbf{X}_1)$ must be one of $\mathbf{h}_1,\ldots,\mathbf{h}_n$,
if there are two elements $Q_1\neq Q_2\in\G$ that belong to
$\argmin_{Q\in\G} \|Q^\top \mathbf{X}_1 - R_n(\mathbf{X}_1)\|^2$,
then necessarily, there exists $\mathbf{h}\in \{\mathbf{h}_1,\ldots,\mathbf{h}_n\}$ such that
$$\|Q_1^\top \mathbf{X}_1 - \mathbf{h}\|^2 = \|Q_2^\top \mathbf{X}_1 - \mathbf{h}\|^2,\ {\rm or\ equivalently}\ \langle \mathbf{X}_1,Q_1\mathbf{h} -Q_2\mathbf{h}\rangle =0. $$
The free group action implies $Q_1\mathbf{h} -Q_2\mathbf{h}\neq \mathbf{0}$. Hence the above event happens with probability 0 (as $\mathbf{X}_1$ is absolutely continuous).

Finally, we show below the uniqueness of the signed-rank. We will use the following lemma.
\begin{lemma}\label{lem:uni_proj}
Let $C\subset\R^p$ be a compact set.
Then for Lebesgue almost every point, it has a unique nearest point in $C$.
\end{lemma}
\begin{proof}
It is known that the distance to $C$, i.e., $f(\mathbf{x}):= d(\mathbf{x},C)$ is 1-Lipschitz and is thus Lebesgue a.e.\ differentiable (by Rademacher's theorem). For any differentiable point $\mathbf{x}_0$ of $f$ that is not in $C$, its projection onto $C$ must be unique (see the proof of \cref{lem:D_mea_zero}).
When $\mathbf{x}_0\in C$, the projection is $\mathbf{x}_0$ itself, and the uniqueness is trivial.
\end{proof}

Since $S_n(\mathbf{X}_i)=\argmin_{Q\in\G} \| \mathbf{X}_i - QR_n(\mathbf{X}_i)\|^2$, $S_n(\mathbf{X}_i)R_n(\mathbf{X}_i)$ is the point in the orbit of $R_n(\mathbf{X}_i)$ that is closest to $\mathbf{X}_i$. But $R_n(\mathbf{X}_i)$ must be one of $\mathbf{h}_1,\ldots,\mathbf{h}_n$.
If $S_n(\mathbf{X}_i)R_n(\mathbf{X}_i)$ is not unique, then there exists $\mathbf{h}\in \{\mathbf{h}_1,\ldots,\mathbf{h}_n\}$ such that the projection of $\mathbf{X}_i$ to $\G\mathbf{h}$ is not unique.
However, by \cref{lem:uni_proj},
for Lebesgue a.e.\ $\mathbf{x}\in\R^p$, $\mathbf{x}$ has a unique closest point in $\G \mathbf{h}$.
Since $\mathbf{X}_i$ is absolutely continuous, $S_n(\mathbf{X}_i)R_n(\mathbf{X}_i)$ is almost surely unique.
\qed

\subsection{Proof of \cref{prop:properties_of_generalized_sign_rank}}\label{sec:pf_properties_of_generalized_sign_rank}
Since $\mathbf{X}_1,\ldots,\mathbf{X}_n$ are i.i.d., the $\hat{\sigma}$ that minimizes the objective function
$$\sum_{i=1}^n c( \mathbf{X}_{\sigma(i)},\mathbf{h}_i),\quad {\rm where}\ c(\mathbf{x},\mathbf{h})=\min_{Q\in\G}\|Q^\top \mathbf{x}-\mathbf{h}\|^2,$$
is equally likely to take any of the $n!$ permutations \citep{Nabarun2021rank}.
Hence, the first conclusion follows.

To show the other two properties, suppose ${\rm H}_0$ holds.
We consider the $\sigma$-field $\mathcal{F}$ generated by the $n$ orbits $\G \mathbf{X}_1,\ldots,\G \mathbf{X}_n$, and the possible randomness introduced to select $\hat{\sigma}$ (if $\hat{\sigma}$ is not unique). Then given $\mathcal{F}$, $\hat{\sigma}$ is determined, i.e., $\hat{\sigma}$ is measurable w.r.t.\ $\mathcal{F}$.

Given $\mathcal{F}$, $\hat{Q}_i$ (see~\eqref{eq:OTproblem}) is defined as \begin{equation}\label{eq:defSi}
    \hat{Q}_i = \argmin_{S_i\in\G}\ \lVert S_i^\top \mathbf{X}_{\hat{\sigma}(i)}-\mathbf{h}_i\rVert^2.
\end{equation}
From \cref{lem:abs_cont} (see below), when $\mathbf{X}$ has a $\G$-symmetric distribution, the density $f$ of $\mathbf{X}$ satisfies $f(\mathbf{x})=f(\mathbf{y})$ whenever $\mathbf{x}$ and $\mathbf{y}$ are on the orbit of $\G$.
Note that given $\mathcal{F}$, only the orbit of $\mathbf{X}_{\hat{\sigma}(i)}$ is known, and $\mathbf{X}_{\hat{\sigma}(i)}$ is equally likely to be any point in this orbit due to $\G$-symmetry, i.e.,
$\mathbf{X}_{\hat{\sigma}(i)}$ conditionally (on $\mathcal{F}$) follows the uniform distribution over the $\G$-orbit $\G \mathbf{X}_{\hat{\sigma}(i)}$.
For any $Q\in \G$, define $\tilde{Q}_i$ as
\begin{equation}\label{eq:tildeSi}
    \tilde{Q}_i:= \argmin_{S_i\in\G}\ \lVert S_i^\top Q^{\top} \mathbf{X}_{\hat{\sigma}(i)}-\mathbf{h}_i\rVert^2.
\end{equation}
If the above does not have a unique minimizer, choose $\tilde{Q}_i$ from the distribution ${\rm Uniform}(\G)$ restricted to the set of minimizers of~\eqref{eq:tildeSi}.
Since $ \mathbf{X}_{\hat{\sigma}(i)}$ and $Q^{\top} \mathbf{X}_{\hat{\sigma}(i)}$ have the same conditional distribution given $\mathcal{F}$,
we know that $\tilde{Q}_i$ has the same conditional distribution as $\hat{Q}_i$.
On the other hand, if we view $S_i^\top Q^\top$ in~\eqref{eq:tildeSi} as a whole, then \eqref{eq:tildeSi}
reduces to the minimization problem in \eqref{eq:defSi}, and thus
$Q\tilde{Q}_i$ has the same conditional distribution as $\hat{Q}_i$.
Hence, $Q\hat{Q}_i\overset{d}{=}\hat{Q}_i$ given $\mathcal{F}$.
This implies given $\mathcal{F}$, $\hat{Q}_i$ are i.i.d.\ following the uniform distribution over $\G$. Consequently, $\{\hat{Q}_i\}_{i=1}^n$ is independent of $\mathcal{F}$ and thus independent of $\hat{\sigma}$.
Hence, the second and the third conclusions hold. \qed

\subsection{Proof of \cref{thm:CLT_wilcoxon}}\label{sec:pf_Wil_CLT}
By the independence of $(S_n(\mathbf{X}_1),\ldots, S_n(\mathbf{X}_n))$ and $(R_n(\mathbf{X}_1),\ldots,R_n(\mathbf{X}_n))$, and
the fact that $(R_n(\mathbf{X}_1),\ldots,R_n(\mathbf{X}_n))$ is a permuted version of $(\mathbf{h}_1,\ldots,\mathbf{h}_n)$,
$\mathbf{W}_n$ has the same distribution as
$$\frac{1}{\sqrt{n}} \sum_{i=1}^n S_n(\mathbf{X}_i) J(\mathbf{h}_i),$$ which has mean $\mathbf{0}$ by~\cref{assump:G}-$(i)$. We first show the convergence of the covariance matrix $\frac{1}{n}\sum_{i=1}^n \mathbb{E}\left[S_n(\mathbf{X}_i) J(\mathbf{h}_i)  J(\mathbf{h}_i)^\top S_n(\mathbf{X}_i)^\top\right]$.
Note that $f(\mathbf{x}):=\mathbb{E}_{S}\left[S \mathbf{x} \mathbf{x}^\top S^\top\right]$, where $S \sim {\rm Uniform}(\G)$, is continuous in $\mathbf{x}$ (by an application of the dominated convergence theorem).
Then the covariance matrix can be written as $ \mathbb{E} f(J(\mathbf{H}_n))$,
where $\mathbf{H}_n \sim \frac{1}{n}\sum_{j=1}^n \delta_{\mathbf{h}_j}$converges weakly to $\mathbf{H}\sim \nu$.

To show $\mathbb{E} f(J(\mathbf{H}_n))\to \mathbb{E} f(J(\mathbf{H}))$ as $n \to \infty$,
by the representation theorem, we may assume that $\mathbf{H}_n\overset{a.s.}{\longrightarrow}\mathbf{H}$ (by~\cref{assump:Weak-nu_n}). Then, as $J(\cdot)$ is $\nu$-a.e.~continuous, by the continuous mapping theorem, $J(\mathbf{H}_n)\overset{a.s.}{\longrightarrow}J(\mathbf{H})$.
Together with $\mathbb{E}|J(\mathbf{H}_{n})_i J(\mathbf{H}_{n})_j|\to \mathbb{E}|J(\mathbf{H})_i J(\mathbf{H})_j|$ (by~\cref{assump:nu_n-nu}-$(iii)$),
we have $J(\mathbf{H}_n) J(\mathbf{H}_n)^\top\overset{L^1}{\longrightarrow} J(\mathbf{H}) J(\mathbf{H})^\top$ by Scheff{\' e}'s lemma. 
Note that $\|f(\mathbf{x})\|_F = \|\mathbb{E}_{S}\left[S \mathbf{x} \mathbf{x}^\top S^\top\right]\|_F \leq \E_S [\|S\mathbf{x}\mathbf{x}^\top S^\top\|_F] = \E_S [\|\mathbf{x}\mathbf{x}^\top \|_F]=\|\mathbf{x}\mathbf{x}^\top \|_F$, where we have used the convexity of the Frobenius norm and $\|AS\|_F= \|SA\|_F = \|A\|_F$ for orthogonal $S$ and any matrix $A$.
Hence, $\|f(J(\mathbf{H}_n))\|_F\leq  \|J(\mathbf{H}_n) J(\mathbf{H}_n)^\top\|_F$.
Together with the fact that  $J(\mathbf{H}_n)J(\mathbf{H}_n)^\top$ converges in $L^1$, we have that  $\|f(J(\mathbf{H}_n))\|_F$ is uniformly integrable.
Since $f(J(\mathbf{H}_n))\overset{a.s.}{\longrightarrow} f(J(\mathbf{H}))$, we further have
$\mathbb{E} f(J(\mathbf{H}_n))\to \mathbb{E} f(J(\mathbf{H}))$.

Since the covariance converges, we only need to check the Lindeberg's condition:
$$\begin{aligned}
&\quad \frac{1}{n}\sum_{i=1}^n \mathbb{E}\left[ \|S_n(\mathbf{X}_i) J(\mathbf{h}_i)\|^2 1_{\|S_n(\mathbf{X}_i) J(\mathbf{h}_i)\| \geq \varepsilon\sqrt{n}}\right] \\
&= \frac{1}{n}\sum_{i=1}^n  \| J(\mathbf{h}_i)\|^2 1_{\| J(\mathbf{h}_i)\| \geq \varepsilon\sqrt{n}}=\mathbb{E} \left[ \|J(\mathbf{H}_n)\|^2 1_{\|J(\mathbf{H}_n)\|\geq \varepsilon \sqrt{n}} \right].
\end{aligned}$$
Note that since $J(\mathbf{H}_n)$ is bounded in probability, $\mathbb{P}(\|J(\mathbf{H}_n)\|\geq \varepsilon \sqrt{n})\to 0$ as $n\to\infty$, so $ \|J(\mathbf{H}_n)\|^2 1_{\|J(\mathbf{H}_n)\|\geq \varepsilon \sqrt{n}}$ converges in probability to 0.
Note that $ \|J(\mathbf{H}_n)\|^2$ is uniformly integrable because $J(\mathbf{H}_n)J(\mathbf{H}_n)^\top$ converges in $L^1$, and hence $ \|J(\mathbf{H}_n)\|^2 1_{\|J(\mathbf{H}_n)\|\geq \varepsilon \sqrt{n}}$ is uniformly integrable. Thus,
$\mathbb{E} \left[\|J(\mathbf{H}_n)\|^2 1_{\|J(\mathbf{H}_n)\|\geq \varepsilon \sqrt{n}}\right] \to 0$. Hence by Lindeberg-Feller's CLT, as $n \to \infty$,
$$\frac{1}{\sqrt{n}} \sum_{i=1}^n S_n(\mathbf{X}_i) J(\mathbf{h}_i)\overset{d}{\to}N\left(\mathbf{0}, \mathbb{E} f(J(\mathbf{H}))\right) \equiv  N\left(\mathbf{0},{\rm Var}(SJ(\mathbf{H}))\right).$$
\qed

\subsection{Proof of Theorems \ref{thm:population_rank}, \ref{cor:conv_popu}, \ref{prop:consistency_signed_rank}, and \ref{prop:consistency_sign}}\label{sec:pf_consistency}
\begin{definition}[$c$-convexity \citep{villani2009OT}]\label{def_cconvex}
A function $\psi:\R^p \to \R\cup\{ + \infty\}$ is said to be {\it $c$-convex} if it is not identically $+\infty$, and there exists $\zeta:\mathbb{R}^p\to\R\cup\{\pm\infty\}$ such that
$$\psi(\mathbf{x}) = \sup_{\mathbf{y}\in\R^p}\left( \zeta(\mathbf{y})-c(\mathbf{x},\mathbf{y})\right),\quad {\rm for\ all}\ \mathbf{x}\in\R^p.$$
Its $c$-transform is the function $\psi^c$ defined by
\begin{equation}\label{eq:c-transform}
\psi^c(\mathbf{y}):=\inf_{\mathbf{x}\in\R^p}\left(\psi(\mathbf{x})+c(\mathbf{x},\mathbf{y})\right),\quad {\rm for\ all}\ \mathbf{y}\in\R^p.
\end{equation}
Its $c$-subdifferential is defined as
$$\partial_c\psi:=\left\{ (\mathbf{x},\mathbf{y})\in\R^p\times \R^p:\psi^c (\mathbf{y})-\psi(\mathbf{x})=c(\mathbf{x},\mathbf{y})\right\}.$$
\end{definition}
If $c(\cdot,\cdot)$ is continuous, then a $c$-convex function is automatically lower semicontinuous (l.s.c.) and Borel measurable \citep[Remark 5.5]{villani2009OT}.

We consider the equivalent cost function $c(\textbf{x},\textbf{y}):=-\max_{Q\in \G} \mathbf{x}^\top Q \mathbf{y}$.
Note that with our choice of $c$, $\psi(\mathbf{x}) = \sup_{\mathbf{y},Q}\left(\zeta(\mathbf{y})+\mathbf{x}^\top Q\mathbf{y} \right)$
is the supremum of linear functions. This implies that a $c$-convex function is convex in the usual sense,
and is thus Lebesgue a.e.\ differentiable in its domain \citep{anderson1952convex}.
It is also easy to see that any $c$-convex function $\psi$ is $\G$-invariant, i.e., $\psi(Q\mathbf{x}) = \psi(\mathbf{x})$ for any $Q\in \G$, and thus $\nabla\psi(\cdot)$ is $\G$-equivariant, i.e., $\nabla\psi(Q\mathbf{x})=Q\nabla\psi(\mathbf{x})$ whenever $\mathbf{x}$ and $Q\mathbf{x}$ are differentiable points of $\psi$.

\begin{lemma}\label{lem:abs_cont}
Let $\mathbf{X}\sim\mu\in\mathcal{P}_{\rm ac}(\mathbb{R}^p)$, and $\G$ be a compact subgroup of the orthogonal group ${\rm O}(p)$. Let $S$ follow the uniform distribution over $\G$, independent of $\mathbf{X}$. Then $S\mathbf{X}\sim\mu_S$ is also Lebesgue absolutely continuous, and $\mu$ is absolutely continuous w.r.t.\ $\mu_S$.
\end{lemma}

\begin{proof}
Suppose $\mathbf{X}$ has a Lebesgue density $f(\cdot)$. We show in the following that $S\mathbf{X}$ has a Lebesgue density given by, for $\mathbf{x} \in \R^p$, $g(\mathbf{x})=\mathbb{E}_S [f(S\mathbf{x})]$, where $S\sim {\rm Uniform}(\G)$. Let $m$ denote the uniform distribution over $\G$ (i.e., the Haar measure on $\G$). For any Borel set $B\subset \R^p$,
$$\begin{aligned}
\mathbb{P}(S\mathbf{X}\in B) &= \int 1_{Q\mathbf{x}\in B}f(\mathbf{x}){\rm d}\mathbf{x}\,{\rm d}m(Q)\\
&=\int 1_{\mathbf{y}\in B} f(Q^{-1} \mathbf{y}){\rm d}\mathbf{y}\,{\rm d}m(Q)\qquad &\left({\rm by\ a\ change\ of \ variable\ \mathbf{y}=Q\mathbf{x}}\right)\\
&=\int 1_{\mathbf{y}\in B} f(Q \mathbf{y}){\rm d}\mathbf{y}\,{\rm d}m(Q)\qquad &\left({\rm as\ }S^{-1}\overset{d}{=} S\sim {\rm Uniform}(\G)\right)\\
&=\int 1_{\mathbf{y}\in B} \mathbb{E}_S f(S\mathbf{y}) {\rm d}\mathbf{y}. \qquad & \left({\rm by\ Fubini\text{'}s\ theorem}\right)
\end{aligned}$$
Hence, $S\mathbf{X}$ has a Lebesgue density given by $g(\mathbf{x})=\mathbb{E}_S f(S\mathbf{x})$, for $\mathbf{x}\in\R^p$.

Recall that the support of a probability measure $\mu$ on $\R^p$ is the closed set defined as:
$${\rm supp}(\mu):=\{ \mathbf{x} \in \R^p : \forall {\rm\ open\ neighborhood\ }N_\mathbf{x} {\rm\ of\ } \mathbf{x}, \,\mu(N_\mathbf{x}) > 0\}.$$
If $\mu$ has a Lebesgue density $f(\cdot)$, then $f(\mathbf{x})>0$ (Lebesgue) a.e.-$\mathbf{x}$\ on ${\rm supp}(\mu)$, and $f(\mathbf{x})=0$ (Lebesgue) a.e.-$\mathbf{x}$\ on $\R^p\backslash {\rm supp}(\mu)$.

We will show that every point in the support of $\mathbf{X}$ is in the support of $S\mathbf{X}$ and thus $\frac{f(\cdot)}{g(\cdot)}$ is well-defined (which is the Radon-Nikodym derivative of $\mathbf{X}$ w.r.t.\ $S\mathbf{X}$). This will establish the absolute continuity of $\mathbf{X}$ w.r.t.\ $S\mathbf{X}$.

Note that every element $Q\in\G$ is in the support of ${\rm Uniform}(\G)$.\footnote{Otherwise, there exists $Q_0 \in \G$ and a neighborhood $U$ of $Q_0$ with probability 0. Consider the following open cover of $\G$: $\{QU: Q \in \G\}$. By the compactness of $\G$, we can find a finite sub-cover $\cup_{i=1}^k Q_iU$ that covers $\G$ and has zero probability, yielding a contradiction.} In particular, the identity matrix $I \equiv I_p$ is in the support of ${\rm Uniform}(\G)$.
Let $\mathbf{x} \in {\rm supp}(\mu)$.
For any neighborhood $N_\mathbf{x}$ of $\mathbf{x}$, there exists an open ball $B(\mathbf{x},\varepsilon)$ contained in $N_\mathbf{x}$.
By the continuity of the group action
$\G\times\R^p\to \R^p:(Q,\mathbf{x})\mapsto Q\mathbf{x}$,
there exists a neighborhood $V$ of $I$ such that for all $Q\in V$ and $\mathbf{y}\in B(\mathbf{x},\varepsilon/2)$,
we have $Q \mathbf{y} \in N_\mathbf{x}$.
Therefore, by independence,
$\mathbb{P}(S\mathbf{X}\in N_\mathbf{x})\geq \mathbb{P}(S\in V,\mathbf{X}\in B(\mathbf{x},\varepsilon/2))
=\mathbb{P}(S\in V)\mathbb{P}(\mathbf{X}\in B(\mathbf{x},\varepsilon/2)) > 0 $, where we have used that $I$, $\mathbf{x}$ are in the supports of ${\rm Uniform}(\G)$ and $\mu$, respectively. Hence, every point in the support of $\mathbf{X}$ is in the support of $S\mathbf{X}\sim\mu_S$.

Note that $g(\cdot)$ is Lebesgue a.e.~positive on ${\rm supp}(\mu_S)$. Thus for any Borel set $B \subset \R^p$,
$$\mathbb{P}(\mathbf{X}\in B) = \int_{{\rm supp}(\mu_S)\cap B} f(\mathbf{x})\,{\rm d}\mathbf{x} = \int_{{\rm supp}(\mu_S)\cap B} \frac{f(\mathbf{x})}{g(\mathbf{x})}g(\mathbf{x})\,{\rm d}\mathbf{x} = \int_{{\rm supp}(\mu_S)\cap B} \frac{f(\mathbf{x})}{g(\mathbf{x})}\,{\rm d}\mu_S(\mathbf{x}).$$
Thus, $\mathbf{X}$ is absolutely continuous w.r.t.\ $S\mathbf{X}$ with Radon-Nikodym derivative $\frac{f(\cdot)}{g(\cdot)}$. \end{proof}

\begin{proof}[Proof of \cref{thm:population_rank}]
By 
\cref{prop:existence_c_cyc},
the continuity of $c(\cdot,\cdot)$ implies that there exists $\pi \in \Pi(\mu,\nu)$ with a $c$-cyclically monotone support. By 
\cref{prop:existence_c_cvx}, for any $\pi \in \Pi(\mu,\nu)$ with a $c$-cyclically monotone support, there exists a $c$-convex function $\psi:\R^p \to \R\cup\{+\infty\}$ such that the support of $\pi$ is contained in the $c$-subdifferential set
$$\partial_c\psi :=\left\{ (\mathbf{x},\mathbf{y})\in\R^p\times \R^p:\psi^c (\mathbf{y})-\psi(\mathbf{x})=c(\mathbf{x},\mathbf{y})\right\}.$$
Since $\psi$ is $c$-convex, and $c(\cdot,\cdot)$ is continuous, $\psi$ is a l.s.c.~convex function in the usual sense. Thus, $\psi$ is Lebesgue a.e.~differentiable in its domain (which contains the support of $\mu$).
Fix any pair of $(\mathbf{x},\mathbf{y})$ in the support of $\pi$ and suppose $\psi$ is differentiable at $\mathbf{x}$.
Let $ \nabla_\mathbf{x} \left(-c(\textbf{x},\textbf{y})\right)$ denote a sub-gradient of $g(\textbf{x}):=\max_{Q\in \G} \mathbf{x}^\top Q \mathbf{y}$,
which clearly exists (as $g(\cdot)$ is a convex function) and is an element in the convex hull of $\{Q \mathbf{y}: Q {\rm\ is\ a\ maximizer\ of\ }\max_{Q\in \G} \mathbf{x}^\top Q \mathbf{y}\}$ (i.e., the gradients of the active linear functions at $\mathbf{x}$; see e.g.,~\citep[Theorem 4.4.2]{Hiriart1993Cvx}).
As $\psi(\cdot)$ is differentiable at $\mathbf{x}$, for $\mathbf{h} \in \R^p$ in a neighborhood of $\mathbf{0} \in \R^p$, we have
\begin{equation}\label{eq:Taylor}
\psi(\textbf{x}+\textbf{h}) = \psi(\textbf{x}) + \nabla\psi(\textbf{x}) \cdot \textbf{h} + o(\textbf{h}).
\end{equation}
Now, using the definition of $c$-transform (see~\eqref{eq:c-transform}): 
$$\begin{aligned}
\psi(\textbf{x}+\textbf{h}) & \geq \psi^c(\textbf{y}) - c(\textbf{x}+\textbf{h},\textbf{y})\\
& \geq \; \psi^c(\textbf{y}) - c(\textbf{x},\textbf{y})+\nabla_\mathbf{x}\left(- c(\textbf{x},\textbf{y})\right)\cdot \textbf{h} \quad \text{ (by the definition\ of\ sub-gradient of } g(\cdot))\\
& =\;\psi(\textbf{x})+\nabla_\mathbf{x}\left(- c(\textbf{x},\textbf{y})\right)\cdot \textbf{h}.
\end{aligned}$$
Since $\mathbf{h}$ can take arbitrary directions, comparing the above display with~\eqref{eq:Taylor}, we get
\begin{equation}\label{eq:Grad-psi}
\nabla \psi(\textbf{x}) =  \nabla_\mathbf{x} \left(-c(\textbf{x},\textbf{y})\right).
\end{equation}
This also implies that the sub-gradient $\nabla_\mathbf{x}\left(- c (\mathbf{x},\mathbf{y})\right)$ is unique, i.e.,
$\{\tilde Q \mathbf{y}: \tilde Q  \in \argmax_{Q\in \G} \mathbf{x}^\top Q \mathbf{y}\}$ consists of only one element.
Let $Q_0:= \argmax_{Q\in \G} \mathbf{x}^\top Q \mathbf{y}$
be the unique maximizer of $\mathbf{x}^\top Q \mathbf{y}$. Then
$\nabla \psi(\textbf{x}) = Q_0 \mathbf{y}$, so that $\textbf{y}= Q_0^{-1}\nabla \psi(\textbf{x})$.
Since the orbit $\{Q \nabla\psi(\textbf{x}): Q\in \G\}$ intersects $ B$ at one point at most, we have $\mathbf{y}= q\left(\nabla \psi(\textbf{x})\right)$  (as $\mathbf{y} \in B$ and  $\nu$ is concentrated on $B$).

Thus, for $\mu$-a.e.~$\textbf{x}$, there is only one $\textbf{y}$ such that $(\textbf{x}, \textbf{y}) \in \partial_c \psi$. So, $\pi$ is given by a deterministic coupling, i.e., $({\rm identity},q\circ \nabla\psi)\#\mu$ and the existence of $R$ is shown.
Although $q\circ \nabla\psi$ is only defined Lebesgue a.e.~on the domain of $\psi$, it has a $\mu$-a.e.\ unique Borel measurable version (see the proof of \citet[Theorem 5.30]{villani2009OT}) which we define to be $R$ (note that $q\circ \nabla\psi$ coincides with $R$ on a Borel set with $\mu$-measure 1).

Let $S \sim {\rm Uniform}(\G)$ independent of everything else. Thus, $S  \stackrel{d}{=} S Q$ for any $Q \in \G$. Therefore, for $\mu$-a.e.~$\mathbf{x}$, $S\nabla\psi(\mathbf{x})\overset{d}{=} S R(\mathbf{x})$ as $\nabla\psi(\mathbf{x})$ and $R(\mathbf{x})=q(\nabla\psi(\mathbf{x}))$ are on the same orbit (i.e., $R(\mathbf{x}) = Q \nabla\psi(\mathbf{x})$ for some $Q \in \G$). Therefore, $S\nabla\psi(\mathbf{X})\overset{d}{=} S R(\mathbf{X}) \overset{d}{=} S\mathbf{H}$ (recall $ R(\mathbf{X}) \sim \mathbf{H}$).
Since $\psi$ is a $c$-convex function,
the $\G$-equivariance of $\nabla\psi(\cdot)$ (see the discussion following Definition~\ref{def_cconvex}) then implies
$\nabla\psi(S\mathbf{X})\overset{d}{=}S\mathbf{H}$.
As $\nabla\psi(\cdot)$ --- gradient of a convex function --- pushes the symmetrized version $S\mathbf{X}$ to $S\mathbf{H}$, by McCann's theorem~\citep{McCann1995} (which does not need any moment assumptions) $\nabla\psi(\cdot)$ is  $\mu_S$-a.e.\ uniquely determined (by the absolute continuity of $\mu_S$).
As $\mu$ is absolutely continuous w.r.t.\ $\mu_S$, $\nabla\psi(\cdot)$ is also $\mu$-a.e.\ uniquely determined, and the uniqueness of $R$ is shown.

To show $(iv)$, note that by definition $$S(\mathbf{x},\mathbf{h}) = \argmin_{Q\in \G} \|Q^\top \mathbf{x} - \mathbf{h}\|^2 = \argmin_{Q\in \G} \| \mathbf{x} - Q\mathbf{h}\|^2 = \argmax_{Q\in \G} \mathbf{x}^\top Q \mathbf{h},$$ and thus,
$S(\mathbf{x},R( \mathbf{x})) R( \mathbf{x})$ is a point in the orbit
$\{QR( \mathbf{x}): Q\in \G\}$ that is closest to $\mathbf{x}$.
However, as shown in~\eqref{eq:Grad-psi}, for $\mu$-a.e.\ $\mathbf{x}$, $\nabla\psi(\mathbf{x})$ is the unique element in $\{\tilde Q R(\mathbf{x}): \tilde Q\in \argmax_{Q\in \G} \mathbf{x}^\top Q R(\mathbf{x})\}$. This implies that $S(\mathbf{X},R( \mathbf{X})) R( \mathbf{X})\overset{a.s.}{=}\nabla\psi(\mathbf{X})$.

Finally, $(v)$ is a direct consequence of \citet[Theorem 5.10]{villani2009OT}.
\end{proof}

At this point, \cref{thm:population_rank} has been proved. In the following, we will develop some further useful results for Theorems \ref{cor:conv_popu}, \ref{prop:consistency_signed_rank} and \ref{prop:consistency_sign}.

\begin{lemma}[Continuity and uniqueness of $S(\cdot,\cdot)$ when $\G$ acts freely on $\G B$]\label{lem:free_G_cont}
Under the assumptions of \cref{thm:population_rank}, if $\G$ acts freely on $\G B$, then $S(\mathbf{x},\mathbf{y}) :=\argmin_{Q\in \G} \|Q^\top \mathbf{x} - \mathbf{y}\|^2$ is unique and continuous for $({\rm identity}, {R})\#\mu$-a.e.~$(\mathbf{x}, \mathbf{y})$.
\end{lemma}
\begin{proof}
We first show the uniqueness. Suppose that  there are two minimizers $Q_1,Q_2$ of $\|Q^\top \mathbf{x} - \mathbf{y}\|^2$. Then $Q_1,Q_2$ are also maximizers of $\mathbf{x}^\top Q \mathbf{y}$. Since for $\mu$-a.e.~$\mathbf{x}$ and $\mathbf{y} = R(\mathbf{x})$, from the proof of \cref{thm:population_rank}, $\{Q \mathbf{y}: Q {\rm\ is\  a\ maximizer\ of\ }\max_{S\in \G} \mathbf{x}^\top Q \mathbf{y}\}$ has only one element $\nabla \psi(\mathbf{x})$, we must have $Q_1 \mathbf{y} = Q_2 \mathbf{y}=\nabla \psi(\mathbf{x})$.
This implies $Q_1^{-1}Q_2 \mathbf{y} =  \mathbf{y}$ and thus $Q_1=Q_2$,  as $\G$ acts freely on $\G B$.
Hence $\argmin_{Q\in \G} \|Q^\top \mathbf{x} - \mathbf{y}\|^2$ is unique $({\rm identity}, {R})\#\mu$-a.s.

Next, we show the continuity.
Suppose $(\mathbf{x}_n,\mathbf{y}_n)\to (\mathbf{x}_0,\mathbf{y}_0)$ as $n\to\infty$ and $S(\mathbf{x},\mathbf{y})$ is uniquely defined at $(\mathbf{x}_0,\mathbf{y}_0)$.
Consider the function $f_n:\R^{p\times p}\to\R$ defined by $f_n(Q)=\|Q^\top \mathbf{x}_n-\mathbf{y}_n\|^2$ if $Q\in \G$, and $f_n(Q)= +\infty$ otherwise.
Let $f:\R^{p\times p}\to \R$ be the function defined by $f(Q)=\|Q^\top \mathbf{x}_0-\mathbf{y}_0\|^2$ if $Q\in \G$ and $f(Q)= +\infty$ otherwise. For any $Q\in\R^{p\times p}$, it is easy to see that
for every sequence $\{Q_n\}_{n \ge 1} \subset \R^{p \times p}$ such that $\|Q_n - Q\|_F \to 0$,
$\liminf_n f_n(Q_n)\geq f(Q)$ (by using the form of $f_n$, which is continuous on $\G$, and the compactness of $\G$). Moreover, for some sequence $\{Q_n\}_{n \ge 1} \subset \R^{p \times p}$ such that $\|Q_n - Q\|_F \to 0$, we also have $\limsup_n f_n(Q_n)\leq f(Q)$.
This implies $f_n$ {\it epi-converges} to $f$ \citep[Proposition 7.2]{rockafellar1998va}.
Note that $f_n$ and $f$ are l.s.c.~with bounded level sets $\{Q\in\R^{p\times p}\,:\, f_n(Q)\leq \alpha \}$ for all $\alpha\in\R$.
By \citet[Theorem 7.33]{rockafellar1998va}, the minimizer(s) of $f_n$ must converge to the minimizer of $f$.
Hence, the continuity of $S(\mathbf{x},\mathbf{y})$ is shown.
\end{proof}

The following weak convergence results will be used in the proofs of Theorems \ref{cor:conv_popu}, \ref{prop:consistency_signed_rank} and \ref{prop:consistency_sign}.
\begin{lemma}[Joint weak convergence of empirical ranks and signed-ranks]\label{lem:joint_weak}
    Suppose Assumptions~\ref{assump:nu-G} and \ref{assump:Weak-nu_n} hold. Let $\mu_n$ denote the empirical distribution on $\{\mathbf{X}_1,\ldots,\mathbf{X}_n\}$ and $\mathcal{X}_n\sim \mu_n$. 
    With probability 1, we have 
$(\mathcal{X}_n,R_n(\mathcal{X}_n))\overset{w}{\longrightarrow}(\mathbf{X},R(\mathbf{X}))$ and $(\mathcal{X}_n,S_n(\mathcal{X}_n)R_n(\mathcal{X}_n))\overset{w}{\longrightarrow}(\mathbf{X},\nabla\psi(\mathbf{X}))$.
\end{lemma}

\begin{proof}
It is known that $\mu_n$ converges weakly to $\mu$ with probability 1.
Fix $\omega \in \Omega$ (the sample space) such that $\mu_n(\omega)\overset{w}{\longrightarrow}\mu$. 
Consider the sequence of random measures $({\rm identity}, {R}_n)\#\mu_n(\omega)$, where ${R}_n$ is the empirical rank map and the randomness comes from the random draw according to $\mu_n(\omega)$.
As ${R}_n\# \mu_n = \nu_n$ converges weakly (to $\nu$, by \cref{assump:Weak-nu_n}) and $\mu_n(\omega)$ converges weakly (to $\mu$), we have that $({\rm identity}, {R}_n)\#\mu_n(\omega)$ is tight. Hence by Prokhorov's theorem, every subsequence of $({\rm identity}, {R}_n)\#\mu_n(\omega)$ has
a further subsequence that converges weakly.
We will show that this limit is $({\rm identity}, {R})\#\mu$ regardless of what sub-sequences are chosen. Consequently, since weak convergence is metrizable, $({\rm identity},{R}_n)\#\mu_n(\omega)$ converges weakly to the same limit.

Now suppose the convergent sub-sequence is $({\rm identity}, {R}_{n_k})\#\mu_{n_k}(\omega)$. We know that $({\rm identity}, {R}_{n_k})\#\mu_{n_k}(\omega)$ is supported on a $c$-cyclically monotone set, following the optimality of ${R}_{n_k}$ (see e.g., the proof of \cref{prop:existence_c_cyc}). This, together with the continuity of $c(\cdot,\cdot)$, implies that its weak limit must also be supported on a $c$-cyclically monotone set (this is also shown in the proof of \cref{prop:existence_c_cyc}).
From \cref{thm:population_rank}, the unique $c$-cyclically monotone measure in $\Pi(\mu,\nu)$ is $({\rm identity}, {R})\#\mu$;
this establishes $({\rm identity}, {R}_n)\#\mu_n(\omega)\overset{w}{\longrightarrow}({\rm identity}, {R})\#\mu$.
Or equivalently, with $\mathcal{X}_n\sim \mu_n$, we have $(\mathcal{X}_n,R_n(\mathcal{X}_n))\overset{w}{\longrightarrow}(\mathbf{X},R(\mathbf{X}))$.

Next, we use the above weak convergence to show that
$(\mathcal{X}_n,S_n(\mathcal{X}_n)R_n(\mathcal{X}_n))\overset{w}{\longrightarrow}(\mathbf{X},\nabla\psi(\mathbf{X}))$.
Recall that the signed-rank is the point in the orbit of the rank vector that is closest to the data point. So by defining $f(\mathbf{x},\mathbf{y}):={\rm Proj}_{\G \mathbf{y}}(\mathbf{x})$ as the projection of $\mathbf{x}$ onto the orbit of $\mathbf{y}$, we can write $S_n(\mathcal{X}_n)R_n(\mathcal{X}_n)=f(\mathcal{X}_n,R_n(\mathcal{X}_n))$,
and $\nabla\psi(\mathbf{X})=f(\mathbf{X},R(\mathbf{X}))$.
We will show that $f(\mathbf{x},\mathbf{y})$ is continuous at $({\rm identity}, {R})\#\mu$ a.e.-$(\mathbf{x},\mathbf{y})$.
We have shown in the proof of \cref{thm:population_rank} that the signed-rank $f(\mathbf{x},\mathbf{y})={\rm Proj}_{\G \mathbf{y}}(\mathbf{x})$ is 
$({\rm identity}, {R})\#\mu$ a.e.-$(\mathbf{x},\mathbf{y})$ uniquely defined. In case where the nearest point in $\G\mathbf{y}$ to $\mathbf{x}$ is not unique, $f(\mathbf{x},\mathbf{y})$ can be chosen as any projection, and we can show that the continuity still holds.

Let $(\mathbf{x}_0,\mathbf{y}_0)$ be a point where the projection ${\rm Proj}_{\G \mathbf{y}_0}(\mathbf{x}_0)$ is unique. We argue by contradiction. If $f(\mathbf{x},\mathbf{y})$ is not continuous at $(\mathbf{x}_0,\mathbf{y}_0)$, then there exist $\delta > 0 $ and a sequence $(\mathbf{x}_n,\mathbf{y}_n)\to (\mathbf{x}_0,\mathbf{y}_0)$ such that $\|{\rm Proj}_{\G \mathbf{y}_n}(\mathbf{x}_n) - {\rm Proj}_{\G \mathbf{y}_0}(\mathbf{x}_0)\|\geq \delta$ for all $n$.
Define the distance from a point $\mathbf{x}\in\R^p$ to a set $A\subset\R^p$ as $d(\mathbf{x},A):=\inf_{\mathbf{y}\in A}(\mathbf{x},\mathbf{y})$. Let $d_0 := d(\mathbf{x}_0,\G\mathbf{y}_0)$, and suppose ${\rm Proj}_{\G \mathbf{y}_0}(\mathbf{x}_0) = g_0 \mathbf{y}_0$ is the point in $\G\mathbf{y}_0$ that achieves this minimal distance, i.e., $\|\mathbf{x}_0 - g_0 \mathbf{y}_0\| = d_0 $.
Observe that $d(\mathbf{x}_n,\G\mathbf{y}_n)\leq d(\mathbf{x}_n,\mathbf{x}_0) + d(\mathbf{x}_0,\G\mathbf{y}_n)\leq d(\mathbf{x}_n,\mathbf{x}_0) + d(\mathbf{x}_0,g_0\mathbf{y}_n)\leq d(\mathbf{x}_n,\mathbf{x}_0) + d(\mathbf{x}_0,g_0\mathbf{y}_0) + d(g_0\mathbf{y}_0,g_0\mathbf{y}_n)=\|\mathbf{x}_n-\mathbf{x}_0\| + d_0 + \|\mathbf{y}_n-\mathbf{y}_0\|$. Therefore, $\|{\rm Proj}_{\G \mathbf{y}_n}(\mathbf{x}_n) - \mathbf{x}_0\|\leq \|{\rm Proj}_{\G \mathbf{y}_n}(\mathbf{x}_n) - \mathbf{x}_n\| + \|\mathbf{x}_n - \mathbf{x}_0\| = d(\mathbf{x}_n,\G\mathbf{y}_n)+\|\mathbf{x}_n - \mathbf{x}_0\|\leq 2\|\mathbf{x}_n-\mathbf{x}_0\| + d_0 + \|\mathbf{y}_n-\mathbf{y}_0\| \to d_0$, as $n\to\infty$.
Hence, ${\rm Proj}_{\G \mathbf{y}_n}(\mathbf{x}_n)$ must have a sub-sequence that converges to a point in $\{\mathbf{x}\in\R^p: \|\mathbf{x} - \mathbf{x}_0\|\leq d_0\}\cap \G\mathbf{y}_0$, which consists of one single point ${\rm Proj}_{\G \mathbf{y}_0}(\mathbf{x}_0)$, by the uniqueness of the projection ${\rm Proj}_{\G \mathbf{y}_0}(\mathbf{x}_0)$. This is a contradiction to $\|{\rm Proj}_{\G \mathbf{y}_n}(\mathbf{x}_n) - {\rm Proj}_{\G \mathbf{y}_0}(\mathbf{x}_0)\|\geq \delta>0$ for all $n$.

Now, by continuous mapping applied to $(\mathbf{x},\mathbf{y})\mapsto (\mathbf{x},f(\mathbf{x},\mathbf{y}))$ and $(\mathcal{X}_n,R_n(\mathcal{X}_n))\overset{w}{\longrightarrow}(\mathbf{X},R(\mathbf{X}))$, we have $(\mathcal{X}_n,S_n(\mathcal{X}_n)R_n(\mathcal{X}_n))\overset{w}{\longrightarrow}(\mathbf{X},\nabla\psi(\mathbf{X}))$.
\end{proof}

Finally, we present a simple lemma which serves as a convenient tool in proving \cref{cor:conv_popu}.
\begin{lemma}\label{lem:UI}
    Suppose Assumption \ref{assump:Weak-nu_n} holds. For a continuous function $f:B\to \R$, if we have $\frac{1}{n}\sum_{i=1}^n |f(\mathbf{h}_{n,i})| \to \int |f(\mathbf{x})|{\rm d}\nu(\mathbf{x})$, then $f\# \nu_n$ is uniformly integrable.
\end{lemma}
\begin{proof}
    Since $\mathbf{H}_n \sim \nu_n$ converges weakly to $\nu$, by choosing an almost sure representative, we may assume
    $\mathbf{H}_n\overset{a.s.}{\longrightarrow}\mathbf{H}$, and by continuous mapping $f(\mathbf{H}_n)\overset{a.s.}{\longrightarrow}f(\mathbf{H})$. Since the moments also converge, i.e., $\E |f(\mathbf{H}_n)|\to \E |f(\mathbf{H})|$, Scheff{\' e}'s lemma implies
$f(\mathbf{H}_n)$ converges to $f(\mathbf{H})$ in $L^1$. This implies $\{f(\mathbf{H}_n)\}_{n\geq 1}$ is uniformly integrable.
\end{proof}


\begin{proof}[Proof of \cref{cor:conv_popu}]

We first state a continuity property of gradients of convex functions. Note that the gradient of convex function $\nabla\psi(\cdot)$ is uniquely defined Lebesgue a.e.
In the negligible case where the gradient does not exist at some point $\mathbf{x}$, we may define $\nabla\psi(\mathbf{x})$ as any subgradient at $\mathbf{x}$, i.e., an element of the subdifferential set $\partial \psi(\mathbf{x})$ (see e.g.~\citep{McCann1995}).
A subgradient always exists in the interior of the domain of a convex function. This notion of $\nabla\psi(\cdot)$ is {\it continuous} in the following sense:  if $\mathbf{x} \in \R^p$ is a differentiable point of $\psi(\cdot)$, and $\mathbf{x}_n\to\mathbf{x}$ and $\mathbf{y}_n\in\partial \psi(\mathbf{x}_n)$, then $\mathbf{y}_n\to\nabla\psi(\mathbf{x})$.

Let us now prove part-1.~of~\cref{cor:conv_popu}. From \cref{lem:joint_weak}, we have with probability 1, $(\mathcal{X}_n,R_n(\mathcal{X}_n))\overset{w}{\longrightarrow}(\mathbf{X},R(\mathbf{X}))$ and $(\mathcal{X}_n,S_n(\mathcal{X}_n)R_n(\mathcal{X}_n))\overset{w}{\longrightarrow}(\mathbf{X},\nabla\psi(\mathbf{X}))$, where
$\mathcal{X}_n\sim \mu_n$. As $q(\cdot)$ is continuous, $R(\cdot) = q(\nabla\psi(\cdot))$ is $\mu$-a.e.\ continuous. Hence, by the continuous mapping theorem applied to $f(\mathbf{x},\mathbf{y})=\rho(R(\mathbf{x}),\mathbf{y})$,
$(\mathcal{X}_n,R_n(\mathcal{X}_n))\overset{w}{\longrightarrow}(\mathbf{X},R(\mathbf{X}))$ implies
$\rho(R(\mathcal{X}_n),R_n(\mathcal{X}_n)) \overset{w}{\longrightarrow} \rho(R(\mathbf{X}),R(\mathbf{X})) =0$.
Note that $R_n(\mathcal{X}_n) \sim \nu_n$.
$\E \rho(\mathbf{H}_{n},\mathbf{0}) \to \E \rho(\mathbf{H},\mathbf{0})$ together with \cref{lem:UI} implies
$\rho(R_n(\mathcal{X}_n),\mathbf{0})$ is uniformly integrable.
Hence, $\rho(R(\mathbf{X}),\mathbf{0})$, as the weak limit of $\rho(R_n(\mathcal{X}_n),\mathbf{0})$, is integrable by Fatou's lemma.
The strong law of large numbers then implies
$\E\rho(R(\mathcal{X}_n),\mathbf{0})$ converges to $\E\rho(R(\mathbf{X}),\mathbf{0})$. Since $\rho(R(\mathcal{X}_n),\mathbf{0})$ converges in distribution to $\rho(R(\mathbf{X}),\mathbf{0})$,
Scheff{\' e}'s lemma again implies $\rho(R(\mathcal{X}_n),\mathbf{0})$ is uniformly integrable.
Hence, $\rho(R(\mathcal{X}_n),R_n(\mathcal{X}_n))\lesssim \rho(R(\mathcal{X}_n),\mathbf{0}) + \rho(R_n(\mathcal{X}_n),\mathbf{0})$ is uniformly integrable and its mean converges to the mean of 0, i.e.,
$$\frac{1}{n}\sum_{i=1}^n \rho(R(\mathbf{X}_i), R_n(\mathbf{X}_i))\to 0.$$

Next, we show the convergence of signs.
By~\cref{lem:free_G_cont} and the continuous mapping theorem applied to $f(\mathbf{x},\mathbf{y})=\tilde{\rho}(S(\mathbf{x},\mathbf{y}),S(\mathbf{x},R(\mathbf{x})))$,
$(\mathcal{X}_n,R_n(\mathcal{X}_n))\overset{w}{\longrightarrow}(\mathbf{X},R(\mathbf{X}))$ implies
$\tilde{\rho}(S(\mathcal{X}_n,R_n(\mathcal{X}_n)),S(\mathcal{X}_n,R(\mathcal{X}_n)))\overset{w}{\longrightarrow}0$.
As $\tilde{\rho}$ is continuous and the signs are bounded, we can apply the bounded convergence theorem which implies that the expectation also converges, i.e.,
$$\frac{1}{n}\sum_{i=1}^n \tilde{\rho}(S(\mathbf{X}_i,R_n(\mathbf{X}_i)) , S(\mathbf{X}_i,R(\mathbf{X}_i)) ) \to 0 \qquad \mbox{as } n \to \infty.$$
Since $S(\mathbf{X}_i,R_n(\mathbf{X}_i))= S_n(\mathbf{X}_i)$, the convergence result follows.

Finally, we show the convergence of the empirical signed-ranks.
Since, with probability 1, $(\mathcal{X}_n,S_n(\mathcal{X}_n)R_n(\mathcal{X}_n))\overset{w}{\longrightarrow}(\mathbf{X},\nabla\psi(\mathbf{X}))$ (by~\cref{lem:joint_weak}), by continuous mapping theorem applied to
$f(\mathbf{x},\mathbf{y})=\rho(\mathbf{y},\nabla\psi(\mathbf{x}))$, we have $\rho(S_n(\mathcal{X}_n)R_n(\mathcal{X}_n),\nabla\psi(\mathcal{X}_n))\overset{w}{\longrightarrow} 0$.
We have shown that
$\rho(R_n(\mathcal{X}_n),\mathbf{0})$ is uniformly integrable.
It follows from $\rho(Q\mathbf{x},\mathbf{0}) \lesssim \rho(\mathbf{x},\mathbf{0}) + 1$ that
$\rho(S_n(\mathcal{X}_n)R_n(\mathcal{X}_n),\mathbf{0})$ is also uniformly integrable.
Similarly, since $\nabla\psi(\mathcal{X}_n)$ and $R(\mathcal{X}_n)$ are on a same orbit, the uniform integrability of $\rho(R(\mathcal{X}_n),\mathbf{0})$ implies that $\rho(\nabla\psi(\mathcal{X}_n),\mathbf{0})$ is uniformly integrable.
Hence, $\rho(S_n(\mathcal{X}_n)R_n(\mathcal{X}_n),\nabla\psi(\mathcal{X}_n))$ is uniformly integrable and its mean converges to the mean of 0, i.e.,
$$\frac{1}{n}\sum_{i=1}^n \rho(S_n(\mathbf{X}_i)R_n(\mathbf{X}_i) , \nabla\psi(\mathbf{X}_i))\to 0\qquad \mbox{as } n \to \infty.$$
\end{proof}

\begin{proof}[Proof of the Theorems~\ref{prop:consistency_signed_rank} and \ref{prop:consistency_sign}]
It suffices to show the following claims:
\begin{enumerate}
\item $\frac{1}{n}\sum_{i=1}^n S_n(\mathbf{X}_i) \overset{a.s.}{\longrightarrow}\mathbb{E} [S(\mathbf{X},R( \mathbf{X}))]$;
\item $\frac{1}{n}\sum_{i=1}^n S_n(\mathbf{X}_i) J(R_n(\mathbf{X}_i)) \overset{a.s.}{\longrightarrow}\mathbb{E} [J(\nabla\psi(\mathbf{X}))]$.
\end{enumerate}
From \cref{lem:joint_weak} we have shown that for a.e.~$\omega$, 
$({\rm identity}, {R}_n)\#\mu_n(\omega)\overset{w}{\longrightarrow}({\rm identity}, {R})\#\mu$.
From \cref{lem:free_G_cont}, $S(\cdot,\cdot)$ is continuous $({\rm identity}, {R})\#\mu$ a.s.~(and for spherical symmetry, $S(\cdot,\cdot,\cdot)$ is continuous $({\rm identity}, {R})\#\mu\otimes P_\varepsilon$-a.s.~as in \cref{rem:Sp-Sym}), by finding an almost sure representative and dominated convergence, the integration converges:
$$\frac{1}{n}\sum_{i=1}^n S_n(\mathbf{X}_i)  = \frac{1}{n}\sum_{i=1}^n S(\mathbf{X}_i,R_n( \mathbf{X}_i), \varepsilon_i)
 \to \mathbb{E} [S(\mathbf{X},R( \mathbf{X}),\varepsilon)].$$

To prove part-2.~above, let $C\subset \R^p$ be the set of points where $J(\cdot)$ is continuous. Then, $\nu_S(C)=1$ by \cref{assump:equivarJ}.
From \cref{thm:population_rank}, as $\nabla\psi(\cdot)$ pushes $\mu_S$ to $\nu_S$, we have $\mathbb{P}(\nabla\psi(S\mathbf{X})\in C)=1$.
Since $\mu$ is absolutely continuous w.r.t.\ $\mu_S$ (by~\cref{lem:abs_cont}), $\mathbb{P}(\nabla\psi(\mathbf{X})\in C)=1$. This allows us to apply the continuous mapping theorem in the following.

From \cref{lem:joint_weak}, for $\mathcal{X}_n\sim \mu_n(\omega)$, we have $(\mathcal{X}_n,S_n(\mathcal{X}_n)R_n(\mathcal{X}_n))\overset{w}{\longrightarrow}(\mathbf{X},\nabla\psi(\mathbf{X}))$. In particular, by the equivariance and continuity of $J(\cdot)$,
$S_n(\mathcal{X}_n)J(R_n(\mathcal{X}_n)) = J(S_n(\mathcal{X}_n)R_n(\mathcal{X}_n))\overset{w}{\longrightarrow}J(\nabla\psi(\mathbf{X}))$.
Note that $\|S_n(\mathcal{X}_n)J(R_n(\mathcal{X}_n))\|=\|J(R_n(\mathcal{X}_n))\|$ is uniformly integrable as: (i) $\E \|J(\mathbf{H}_n)\| \to \E \|J(\mathbf{H})\| $, (ii) $J(\mathbf{H}_n) \stackrel{w}{\longrightarrow} J(\mathbf{H})$, and (iii) an application of Scheff{\' e}'s lemma. Hence, we have the convergence of the first moment:
$$\begin{aligned}
\frac{1}{n}\sum_{i=1}^n S_n(\mathbf{X}_i) J({R}_n(\mathbf{X}_i) ) \to \mathbb{E} [J(\nabla\psi(\mathbf{X}))].
\end{aligned}$$ 
\end{proof}

\subsection{Proof of \cref{rk:classical_Wilcoxon}}\label{sec:pf_classical_Wilcoxon}
Under the assumptions of \cref{rk:classical_Wilcoxon}, $S(X,R(X))={\rm sign}(X)$ and
$R(X) = F_{|X|}(|X|)$, since ${\rm sign}(\cdot)F_{|X|}(|\cdot|)$  transports $\mu_S$ to ${\rm Uniform}(-1,1)$
and appears as the gradient of a convex function.
Hence from \cref{prop:consistency_signed_rank}, the GWSR test is consistent against $\mathbb{E}\left[ {\rm sign}(X)F_{|X|}(|X|)\right] \neq  0$. Suppose $f$ is the density function of $X$. Then:
$$\begin{aligned}
&\quad\mathbb{E}\left[ {\rm sign}(X)F_{|X|}(|X|)\right] >  0\\
& \Leftrightarrow \int_0^\infty \left[ \frac{f(x)}{f(x)+f(-x)} - \frac{f(-x)}{f(x)+f(-x)}\right] F_{|X|}(x)\left[f(x)+f(-x)\right]{\rm d}x > 0\\
&\Leftrightarrow\int_0^\infty \left[ f(x)-f(-x)\right] F_{|X|}(x){\rm d}x > 0\\
&\Leftrightarrow\int_0^\infty \left[ f(x)-f(-x)\right] \int_{0}^x \left[f(y)+f(-y)\right]{\rm d}y {\rm d}x > 0\\
&\Leftrightarrow\int_{0\leq y\leq x} \left[ f(x)-f(-x)\right]\left[f(y)+f(-y)\right] {\rm d}x{\rm d}y >0\\
&\Leftrightarrow\int_{0\leq y\leq x}  f(x)\left[f(y)+f(-y)\right] {\rm d}x{\rm d}y > \int_{0\leq y\leq x}f(-x)\left[f(y)+f(-y)\right] {\rm d}x{\rm d}y\\
&\Leftrightarrow \mathbb{P}(0\leq X_2\leq X_1) + \mathbb{P}(0\leq -X_2\leq X_1) >
\mathbb{P}(0\leq X_2\leq -X_1) + \mathbb{P}(0\leq -X_2\leq -X_1) \\
&\Leftrightarrow \mathbb{P}(-X_1\leq X_2\leq X_1) > \mathbb{P}(X_1\leq X_2\leq -X_1) \\
&\Leftrightarrow \mathbb{P}(X_1+X_2\geq 0, X_2\leq X_1) > \mathbb{P}(X_1+ X_2\leq 0, X_1\leq X_2) \\
&\Leftrightarrow \frac{1}{2}\mathbb{P}(X_1+X_2\geq 0) > \frac{1}{2} \mathbb{P}(X_1+ X_2\leq 0)\\
&\Leftrightarrow \mathbb{P}(X_1+X_2 > 0) > \frac{1}{2}.
\end{aligned}$$
By replacing $X$ with $-X$, we have:
$$\mathbb{E}\left[ {\rm sign}(X)F_{|X|}(|X|)\right] <  0\quad \Leftrightarrow\quad  \mathbb{P}(X_1+X_2 > 0) < \frac{1}{2},$$
which concludes the proof.
\qed

\subsection{Proof of \cref{cor:location_shift}}\label{sec:pf_location_shift}
By \cref{prop:consistency_signed_rank}, we only need to show that $ \mathbb{E}[\nabla\psi(\mathbf{X})]\neq \mathbf{0}$.

We begin by showing $ \nabla\psi(\mathbf{X})$ and $ \nabla\psi(\mathbf{X}-2\mathbf{\Delta}) $ are integrable.
Note that $\E[\|\nabla\psi(\mathbf{X})\|] = \E[\|\mathbf{H}\|]$ where $\mathbf{H} \sim \nu$; hence $\nabla\psi(\mathbf{X})\in L^1$.
Since $-I\in\G$ (by Assumption \ref{assump:G}-$(ii)$),
\begin{equation}\label{eq:eqdist}
2\mathbf{\Delta}- \mathbf{X} = 
\mathbf{\Delta} + (\mathbf{\Delta}-\mathbf{X})\overset{d}{=}
\mathbf{\Delta} + (\mathbf{X}-\mathbf{\Delta})=\mathbf{X}.\end{equation}
Since $\nabla \psi$ is equivariant under the group action of $\G$,
$$\nabla\psi(\mathbf{X}-2\mathbf{\Delta}) =- \nabla\psi(2\mathbf{\Delta}-\mathbf{X}) \overset{d}{=} -\nabla\psi( \mathbf{X}) \in L^1.$$
To show $\mathbb{E}\left[\nabla\psi(\mathbf{X})\right] \ne \mathbf{0}$ we argue by contradiction. If $\mathbb{E}\left[\nabla\psi(\mathbf{X})\right]=\mathbf{0}$, then \eqref{eq:eqdist} would imply
$\mathbb{E}\left[\nabla\psi( \mathbf{X}-2\mathbf{\Delta})\right]=\mathbf{0}$
and thus:
$$\mathbb{E}\left[  \mathbf{\Delta}^\top \left\{ \nabla\psi(\mathbf{X}) - \nabla\psi(\mathbf{X}-2\mathbf{\Delta})  \right\}\right]= 0.$$
By the monotonicity of gradients of convex functions:
$2\mathbf{\Delta}^\top \left\{ \nabla\psi(\mathbf{X}) - \nabla\psi(\mathbf{X}-2\mathbf{\Delta})  \right\}\geq 0$.
This implies $\mathbf{\Delta}^\top \left\{ \nabla\psi(\mathbf{X}) - \nabla\psi(\mathbf{X}-2\mathbf{\Delta})  \right\}=0$ a.s.

Let $A$ be an open set of positive measure w.r.t.~$\mu$ where $\psi$ is strictly convex. Let $\mathbf{x}_0 \in A$.
Consider $g(t) := \psi(\mathbf{x}_0 - t \mathbf{\Delta})$, $t\in \R$, which is a one-dimensional convex function, and is strictly convex in a neighborhood of 0.
This implies $g'(2) > g'(0)$, i.e., $\mathbf{\Delta}^\top \left\{ \nabla\psi(\mathbf{x}_0) - \nabla\psi(\mathbf{x}_0-2\mathbf{\Delta})  \right\}>0$. As $\mathbf{\Delta}^\top \left\{ \nabla\psi(\mathbf{x}_0) - \nabla\psi(\mathbf{x}_0-2\mathbf{\Delta})  \right\}\geq 0$ for $\mathbf{x}_0\notin A$, we obtain 
$$\begin{aligned}
& \quad \mathbb{E}\left[  \mathbf{\Delta}^\top \left\{ \nabla\psi(\mathbf{X}) - \nabla\psi(\mathbf{X}-2\mathbf{\Delta})  \right\}\right] \\
& = \mathbb{E}\left[  \mathbf{\Delta}^\top \left\{ \nabla\psi(\mathbf{X}) - \nabla\psi(\mathbf{X}-2\mathbf{\Delta})  \right\} \mathbf{1}_A(\mathbf{X}) \right] + \mathbb{E}\left[  \mathbf{\Delta}^\top \left\{ \nabla\psi(\mathbf{X}) - \nabla\psi(\mathbf{X}-2\mathbf{\Delta})  \right\} \mathbf{1}_{A^c}(\mathbf{X}) \right]>0,
\end{aligned}$$
which yields a contradiction.


\qed

\subsection{Proof of \cref{lem:population_signed_rank}}\label{sec:pf_lem:population_signed_rank}
Observe that
\begin{equation}\label{eq:L2diff}\begin{aligned}
&\quad \E \left\|\frac{1}{\sqrt{n}}\sum_{i=1}^n S_n(\mathbf{X}_i)  J(R_n(\mathbf{X}_i))
- \frac{1}{\sqrt{n}}\sum_{i=1}^n S(\mathbf{X}_i,R(\mathbf{X}_i))  J(R(\mathbf{X}_i))\right\|^2\\
&=\frac{1}{n} \E \left\|\sum_{i=1}^n S_n(\mathbf{X}_i)  J(R_n(\mathbf{X}_i)) \right\|^2
+ \frac{1}{n} \E \left\|\sum_{i=1}^n S(\mathbf{X}_i,R(\mathbf{X}_i))  J(R(\mathbf{X}_i)) \right\|^2\\
&\quad -\frac{2}{n}\E \left( \sum_{i=1}^n S_n(\mathbf{X}_i)  J(R_n(\mathbf{X}_i)) \right)^\top \left(\sum_{i=1}^n S(\mathbf{X}_i,R(\mathbf{X}_i))  J(R(\mathbf{X}_i))\right).
\end{aligned}\end{equation}
Since $(S_n(\mathbf{X}_1),\ldots,S_n(\mathbf{X}_n))\indep (R_n(\mathbf{X}_1),\ldots,R_n(\mathbf{X}_n))$ and follows i.i.d.\ ${\rm Uniform}(\G)$,
$$\frac{1}{n} \E \left\|\sum_{i=1}^n S_n(\mathbf{X}_i)  J(R_n(\mathbf{X}_i)) \right\|^2 =\frac{1}{n} \sum_{i=1}^n \E\left\| S_n(\mathbf{X}_i)  J(R_n(\mathbf{X}_i)) \right\|^2 = 
\E\left\| J(R_n(\mathbf{X}_1)) \right\|^2,$$
as the cross terms all have expectation 0 (by \cref{assump:G}-$(i)$).

By the $\G$-symmetry of $\mathbf{X}_i$ and the equivariance of $\nabla\psi$ and $J$ (here $S \sim {\rm Uniform}(\G)$),
$$\E[S(\mathbf{X}_i,R(\mathbf{X}_i))  J(R(\mathbf{X}_i))] = \E[J(\nabla\psi(\mathbf{X}_i))] =\E[J(\nabla\psi(S\mathbf{X}_i))] =  \E[S]\E[J(\nabla\psi(\mathbf{X}_i))] = \mathbf{0}.$$
This implies
$$\frac{1}{n} \E \left\|\sum_{i=1}^n S(\mathbf{X}_i,R(\mathbf{X}_i))  J(R(\mathbf{X}_i)) \right\|^2 = \frac{1}{n}
\sum_{i=1}^n \E \left\| S(\mathbf{X}_i,R(\mathbf{X}_i))  J(R(\mathbf{X}_i)) \right\|^2
=\E \left\|J(R(\mathbf{X}_1)) \right\|^2.$$
To simplify the third term in \eqref{eq:L2diff}, we first show that for $i\neq j$,
$$\E \left[\Big(  S_n(\mathbf{X}_i)  J(R_n(\mathbf{X}_i)) \Big)^\top \Big( S(\mathbf{X}_j,R(\mathbf{X}_j))  J(R(\mathbf{X}_j))\Big) \right]= 0.$$
We first condition on the $\mathbf{X}_1,\ldots,\mathbf{X}_{j-1},\mathbf{X}_{j+1},\ldots,\mathbf{X}_n$ and the orbit of $\mathbf{X}_j$, i.e., $\{Q\mathbf{X}_j:Q\in \G\}$.
Then $S_n(\mathbf{X}_i)J(R_n(\mathbf{X}_i))$ is known, and $\mathbf{X}_j$ is still equally likely to take any element in the orbit, due to the $\G$ symmetry of $\mathbf{X}_j$.
Hence, by the equivariance of $\nabla\psi$ and $J(\cdot)$,
$S(\mathbf{X}_j,R(\mathbf{X}_j))J(R(\mathbf{X}_j))=J(\nabla\psi(\mathbf{X}_j))$ has conditional expectation $\mathbf{0}$ and consequently the unconditional expectation is $\mathbf{0}$ as well.
Hence,
$$\begin{aligned}
&\quad \frac{1}{n}\E \left( \sum_{i=1}^n S_n(\mathbf{X}_i)  J(R_n(\mathbf{X}_i)) \right)^\top \left(\sum_{i=1}^n S(\mathbf{X}_i,R(\mathbf{X}_i))  J(R(\mathbf{X}_i))\right)\\
&=\frac{1}{n} \sum_{i=1}^n \E \Big(  S_n(\mathbf{X}_i)  J(R_n(\mathbf{X}_i)) \Big)^\top \Big( S(\mathbf{X}_i,R(\mathbf{X}_i))  J(R(\mathbf{X}_i))\Big)\\
&=\frac{1}{n} \sum_{i=1}^n \E \left[ J(S_n(\mathbf{X}_i)R_n (\mathbf{X}_i))^\top  J(\nabla\psi(\mathbf{X}_i))\right].
\end{aligned}$$
Recall that from~\cref{lem:joint_weak} if $\mathcal{X}_n\sim \mu_n(\omega)$, we have $(\mathcal{X}_n,S_n(\mathcal{X}_n)R_n(\mathcal{X}_n))\overset{w}{\longrightarrow}(\mathbf{X},\nabla\psi(\mathbf{X}))$. Therefore, by continuous mapping theorem, by~\cref{assump:equivarJ},
$$\begin{aligned}
&\quad J(S_n(\mathcal{X}_n)R_n (\mathcal{X}_n))^\top  J(\nabla\psi(\mathcal{X}_n))
+ \|J(S_n(\mathcal{X}_n)R_n (\mathcal{X}_n))\|^2 + \|J(\nabla\psi(\mathcal{X}_n))\|^2\\
&\overset{w}{\longrightarrow}J(\nabla\psi(\mathbf{X}))^\top  J(\nabla\psi(\mathbf{X}))
+ \|J(\nabla\psi(\mathbf{X}))\|^2 + \|J(\nabla\psi(\mathbf{X}))\|^2.
\end{aligned}$$
Since the above random variables are nonnegative, by Fatou's lemma, the expectation w.r.t.\ $\mu_n(\omega)$ satisfies:
$$
\begin{aligned}
&\liminf_{n\to\infty}\frac{1}{n}\sum_{i=1}^n\left[J(S_n(\mathbf{X}_i)R_n (\mathbf{X}_i))^\top  J(\nabla\psi(\mathbf{X}_i))
+ \|J(S_n(\mathbf{X}_i)R_n (\mathbf{X}_i))\|^2 + \|J(\nabla\psi(\mathbf{X}_i))\|^2 \right]\\
&\geq \E \left[ 3\left\|J(\nabla\psi(\mathbf{X})) \right\|^2\right]= 3\E \left\|J(R(\mathbf{X})) \right\|^2.
\end{aligned}$$
Note that the above holds for a.e.~$\omega$. Hence we can take the usual expectation and again by Fatou's lemma:
$$\begin{aligned}
&\quad \liminf_{n\to\infty}\frac{1}{n} \sum_{i=1}^n \E \left[ J(S_n(\mathbf{X}_i)R_n (\mathbf{X}_i))^\top  J(\nabla\psi(\mathbf{X}_i))+ \|J(S_n(\mathbf{X}_i)R_n (\mathbf{X}_i))\|^2 + \|J(\nabla\psi(\mathbf{X}_i))\|^2\right]\\
&\geq \mathbb{E}\left[ \liminf_{n\to\infty}\frac{1}{n} \sum_{i=1}^n \left[J(S_n(\mathbf{X}_i)R_n (\mathbf{X}_i))^\top  J(\nabla\psi(\mathbf{X}_i))
+ \|J(S_n(\mathbf{X}_i)R_n (\mathbf{X}_i))\|^2 + \|J(\nabla\psi(\mathbf{X}_i))\|^2 \right] \right]\\
&\geq \E \left[3\E \left\|J(R(\mathbf{X})) \right\|^2\right] =3\E \left\|J(R(\mathbf{X})) \right\|^2.
\end{aligned}$$

Note that
$\E\left\| J(S_n(\mathbf{X}_i)R_n(\mathbf{X}_i)) \right\|^2 = \E\left\| J(R_n(\mathbf{X}_1)) \right\|^2$ for all $i$, and they converge to $\E \left\|J(R(\mathbf{X})) \right\|^2$ following from the convergence of the second moments of $J\# \nu_n$ (\cref{assump:nu_n-nu}-$(ii)$). Therefore,

$$\liminf_{n\to\infty}\frac{1}{n} \sum_{i=1}^n \E \left[ J(S_n(\mathbf{X}_i)R_n (\mathbf{X}_i))^\top  J(\nabla\psi(\mathbf{X}_i))\right]
\geq \E \left\|J(R(\mathbf{X})) \right\|^2.$$
Aggregating the previous discussions, we have
$$\begin{aligned}
&\quad \limsup_{n\to\infty} \E \left\|\frac{1}{\sqrt{n}}\sum_{i=1}^n S_n(\mathbf{X}_i)  J(R_n(\mathbf{X}_i))
- \frac{1}{\sqrt{n}}\sum_{i=1}^n S(\mathbf{X}_i,R(\mathbf{X}_i))  J(R(\mathbf{X}_i))\right\|^2\\
&=\limsup_{n\to\infty}  \Bigg[ \E\left\| J(R_n(\mathbf{X}_1)) \right\|^2 + \E \left\|J(R(\mathbf{X}_1)) \right\|^2 - \frac{2}{n} \sum_{i=1}^n \E \left[ J(S_n(\mathbf{X}_i)R_n (\mathbf{X}_i))^\top  J(\nabla\psi(\mathbf{X}_i))\right]\Bigg]\\
&\leq  \lim_{n\to\infty}\E\left\| J(R_n(\mathbf{X}_1)) \right\|^2 + \E \left\|J(R(\mathbf{X}_1)) \right\|^2 - 2\liminf_{n\to\infty}\frac{1}{n} \sum_{i=1}^n \E \left[ J(S_n(\mathbf{X}_i)R_n (\mathbf{X}_i))^\top  J(\nabla\psi(\mathbf{X}_i))\right]\\
&\leq \E \left\|J(R(\mathbf{X})) \right\|^2 + \E \left\|J(R(\mathbf{X})) \right\|^2 - 2\E \left\|J(R(\mathbf{X})) \right\|^2=0,
\end{aligned}$$ which completes the proof.
\qed

\subsection{Proof of \cref{thm:asymptotics_local_alternatives}}\label{sec:pf_thm:asymptotics_local_alternatives}
Let $Q_n$ be the law of  $(\mathbf{X}_1,\ldots,\mathbf{X}_n)$ under $\textrm{H}_1$,
and $P_n$ be the law of  $(\mathbf{X}_1,\ldots,\mathbf{X}_n)$ under $\textrm{H}_0$.
Under the regularity conditions in \cref{sec:regu}, if ${\B \xi}_n \to{\B \xi}\in\R^p$, then by \citet[Theorem 7.2]{{vanderVaart1998}}, under $\textrm{H}_0$:
$$\log \frac{{\rm d}Q_n}{{\rm d}P_n} = \sum_{i=1}^n\log \frac{f(\mathbf{X}_i- {\B \xi}_n/\sqrt{n})}{f(\mathbf{X}_i)} = - \frac{1}{\sqrt{n}}\sum_{i=1}^n \frac{{\B \xi}^\top \nabla f (\mathbf{X}_i)}{f(\mathbf{X}_i)} - \frac{1}{2}{\B \xi}^\top I_{{\B \theta}_0} {\B \xi} + o_p(1),$$
where $ I_{{\B \theta}_0}$ is the Fisher information at ${\B \theta}_0=\mathbf{0}$.

By \cref{lem:population_signed_rank} and the multivariate central limit theorem:
$$\begin{aligned}
&\quad \begin{pmatrix}
\mathbf{W}_n\\
\log \frac{{\rm d}Q_n}{{\rm d}P_n}\\
\end{pmatrix}
=
\begin{pmatrix}
\frac{1}{\sqrt{n}}\sum_{i=1}^n J(\nabla\psi(\mathbf{X}_i))\\
-\frac{1}{\sqrt{n}}\sum_{i=1}^n \frac{{\B \xi}^\top \nabla f (\mathbf{X}_i)}{f(\mathbf{X}_i)} - \frac{1}{2}{\B \xi}^\top I_{{\B \theta}_0} {\B \xi}\\
\end{pmatrix} + o_p(1)\\
&\overset{d}{\longrightarrow}
N\left(
\begin{pmatrix}
\mathbf{0}\\
-\frac{{\B \xi}^\top I_{{\B \theta}_0} {\B \xi} }{2}\\
\end{pmatrix},
\begin{pmatrix}
{\rm Var}(J(\nabla\psi(\mathbf{X}))) & {\B \gamma}\\
{\B \gamma}^\top &{\B \xi}^\top I_{{\B \theta}_0} {\B \xi}\\
\end{pmatrix}
\right),
\end{aligned}$$
where ${\B \gamma} = -\mathbb{E}_{\textrm{H}_0} \left[ J(\nabla\psi(\mathbf{X}))\frac{{\B \xi}^\top \nabla f (\mathbf{X})}{f(\mathbf{X})}\right]$.
Now, by Le Cam's third lemma \citep[Example 6.7]{vanderVaart1998}, under $\textrm{H}_1$,
$$\mathbf{W}_n\overset{d}{\longrightarrow} N\left({\B \gamma},{\rm Var}\big(J(\nabla\psi(\mathbf{X}))\big)\right)= N\left({\B \gamma},\Sigma_{\rm ERD}\right).$$
\qed

\subsection{Proof for \cref{defn:ARE}}\label{sec:pf:defnARE}
We first consider a sequence ${\B \theta}_n =\frac{{\B \xi}_n}{\sqrt{n}} $, where ${\B \xi}_n\to{\B \xi}\neq \mathbf{0}$.
From \cref{sec:pf_thm:asymptotics_local_alternatives} we know that under ${\B \theta}_n =\frac{{\B \xi}_n}{\sqrt{n}} $,
$$\Sigma_{\rm ERD}^{-1/2} \mathbf{W}_n\overset{d}{\longrightarrow}
\Sigma_{\rm ERD}^{-1/2}{\B \gamma} + N(\mathbf{0},I_p),  \quad {\rm and}\quad \sqrt{n} S_n^{-1/2}\bar{\mathbf{X}}_n \overset{d}{\longrightarrow} \Sigma_\mathbf{X}^{-1/2}{\B \xi} + N(\mathbf{0}, I_p),$$
with ${\B \gamma}\equiv {\B \gamma}({\B \xi})=-\mathbb{E}_{\textrm{H}_0} \left[J(\nabla\psi(\mathbf{X}))\frac{{\B \xi}^\top \nabla f (\mathbf{X})}{f(\mathbf{X})}\right]$.

Hence, the level $\alpha$ GWSR test and Hotelling's $T^2$ test based on $n$ observations have powers converging to:
$$\mathbb{P}\left(\left\| \Sigma_{\rm ERD}^{-1/2}{\B \gamma}+ N(\mathbf{0}, I_p)   \right\|^2 \geq \chi_p^2(\alpha) \right),\quad {\rm and} \quad 
\mathbb{P}\left(\left\| \Sigma_\mathbf{X}^{-1/2}{\B \xi} + N(\mathbf{0}, I_p)   \right\|^2 \geq \chi_p^2(\alpha) \right),$$
respectively, where $\chi_p^2(\alpha)\in\R$ satisfies $\mathbb{P}\left(\chi_p^2 \geq \chi_p^2(\alpha) \right) = \alpha$.

If we use $C_1 n$ observations for the GWSR test with some $C_1 > 0$, then
under ${\B \theta}_n =\frac{{\B \xi}_n}{\sqrt{n}}=\frac{\sqrt{C_1}{\B \xi}_n}{\sqrt{C_1n}}$, where $\sqrt{C_1}{\B \xi}_n \to \sqrt{C_1} {\B \xi}$, the power of the level $\alpha$ GWSR test converges to
$$\mathbb{P}\left(\left\| \Sigma_{\rm ERD}^{-1/2}\sqrt{C_1}{\B \gamma}+ N(\mathbf{0}, I_p)   \right\|^2 \geq \chi_p^2(\alpha) \right).$$
Note that for a given $\beta \in (\alpha,1)$, there exists a unique $C_1>0$ such that
$$\mathbb{P}\left(\left\| \Sigma_{\rm ERD}^{-1/2}\sqrt{C_1}{\B \gamma}+ N(\mathbf{0}, I_p)   \right\|^2 \geq \chi_p^2(\alpha) \right)=\beta.$$
Similarly, there exists a unique $C_2>0$ such that under ${\B \theta}_n =\frac{{\B \xi}_n}{\sqrt{n}}$, the level $\alpha$ Hotelling's $T^2$ test based on $C_2 n$ observations has power converging to
$$\mathbb{P}\left(\left\| \Sigma_\mathbf{X}^{-1/2}\sqrt{C_2}{\B \xi} + N(\mathbf{0}, I_p)   \right\|^2 \geq \chi_p^2(\alpha) \right)=\beta.$$
This implies $\|\Sigma_{\rm ERD}^{-1/2}\sqrt{C_1}{\B \gamma}\| = \|\Sigma_\mathbf{X}^{-1/2}\sqrt{C_2}{\B \xi}\|$ in order to achieve the same power $\beta$, and
$$\lim_{n\to\infty}\frac{N_2\left(\alpha,\beta,\frac{{\B \xi}_n}{\sqrt{n}}\right)}{N_1\left(\alpha,\beta,\frac{{\B \xi}_n}{\sqrt{n}}\right)}=\lim_{n\to\infty}\frac{C_2 n + o(n)}{C_1 n + o(n)}= \frac{\|\Sigma_{\rm ERD}^{-1/2}{\B \gamma}\|^2}{\| \Sigma_\mathbf{X}^{-1/2}{\B \xi} \|^2}.$$
Note that for any $\varepsilon_n {\B \xi}_n \to 0$ such that $\varepsilon_n\to 0$ and ${\B \xi}_n \to {\B \xi}\neq 0$, any subsequence of $\{\varepsilon_n {\B \xi}_n \}_{n\geq 1}$ has a further subsequence that is contained in a sequence of the form $\{\frac{\tilde{{\B \xi}}_n}{\sqrt{n}}\}_{n\geq 1}$ with $\tilde{{\B \xi}}_n\to {\B \xi}$.
This implies the $\frac{1}{\sqrt{n}}$ rate is not important, and
we also have $\lim_{n\to\infty}\frac{N_2\left(\alpha,\beta,\varepsilon_n {\B \xi}_n\right)}{N_1\left(\alpha,\beta,\varepsilon_n {\B \xi}_n\right)} = \frac{\|\Sigma_{\rm ERD}^{-1/2}{\B \gamma}\|^2}{\| \Sigma_\mathbf{X}^{-1/2}{\B \xi} \|^2}$.

\subsection{Proof of \cref{prop:Gaussian}}\label{sec:pf_prop:Gaussian}
Note that all the score functions $J(\cdot)$ used in this proof satisfy \cref{assump:equivarJ}, and hence
\cref{thm:asymptotics_local_alternatives} applies.\\

\noindent Part 1.~(Gaussian ERD with identity score function):
From \cref{thm:population_rank} we know that
$\nabla\psi$ is the optimal transport map under the square loss that transports $\mu=\mu_S$ to $\nu_S=N(\mathbf{0},I)$ under ${\rm H}_0$,  and it is also the unique map pushing $\mu$ to $N(\mathbf{0},I)$ that appears as the gradient of a convex function \citep{McCann1995}.

In this Gaussian case, $\nabla\psi(\mathbf{x}) = \Sigma^{-1/2}\mathbf{x}$, which maps $N(\mathbf{0},\Sigma)$ to $N(\mathbf{0},I)$, and is the gradient of a convex function. Let $\mathbf{b}:=\Sigma^{-1/2}{\B \xi} $ and $\mathbf{G}\sim N(0,I)$. Then:
$$\begin{aligned}\|\Sigma_{\rm ERD}^{-1/2}{\B \gamma}\|^2 = \|{\B \gamma}\|^2 = \left\|\mathbb{E}_{\textrm{H}_0} \left[ \nabla\psi(\mathbf{X})\frac{{\B \xi}^\top \nabla f (\mathbf{X})}{f(\mathbf{X})}\right]\right\|^2
&=\left\| \mathbb{E} \left[ \Sigma^{-1/2}\mathbf{X} \cdot {\B \xi}^\top \Sigma^{-1}\mathbf{X}\right] \right\|^2\\
&=\left\| \mathbb{E} \left[ \mathbf{G} \cdot \mathbf{b}^\top \mathbf{G}\right] \right\|^2\\
&=\|\mathbf{b}\|^2=\|\Sigma^{-1/2}{\B \xi} \|^2.
\end{aligned}$$
This implies ${\rm ARE}(\mathbf{W}_n,\bar{\mathbf{X}}_n;{\B \xi})=1$.\\

\noindent Part 2.~(Uniform ERD): 
In this case, we still have $\nu_S= N(\mathbf{0},I)$ 
so $\nabla\psi(\mathbf{x}) = \Sigma^{-1/2}\mathbf{x}$.
For ${\rm ERD}={\rm Uniform}(-1,1)^p$, $\Sigma_{\rm ERD}=\frac{1}{3} I$. 
Using the integration by parts formula,
$$\begin{aligned}\|\Sigma_{\rm ERD}^{-1/2}{\B \gamma}\|^2 &= 3\|{\B \gamma}\|^2 =3 \left\|\mathbb{E}_{\textrm{H}_0} \left[ J(\nabla\psi(\mathbf{X}))\frac{{\B \xi}^\top \nabla f (\mathbf{X})}{f(\mathbf{X})}\right]\right\|^2\\
&=3\sum_{j=1}^p\left(\int J(\nabla\psi(\mathbf{x}))_j \sum_{i=1}^p\xi_i \frac{\partial}{\partial x_i}f(\mathbf{x}) {\rm d}\mathbf{x} \right)^2\\
&= 3\sum_{j=1}^p \left(\sum_{i=1}^p  \E \left[\xi_i\frac{\partial}{\partial X_i} J(\nabla \psi(\mathbf{X}))_j \right]\right)^2\\
&=3\sum_{j=1}^p \left(\sum_{i=1}^p  \E \left[\xi_i\frac{\partial}{\partial X_i} \left(2\Phi(e_j^\top \Sigma^{-1/2}\mathbf{X})-1\right)\right]\right)^2\\
&=3\sum_{j=1}^p \left(\sum_{i=1}^p  \E \left[2\xi_i \Sigma^{-1/2}_{ji} \phi(e_j^\top \Sigma^{-1/2}\mathbf{X})\right]\right)^2,
\end{aligned}$$
where in the third line, we have used the integration by parts. Here
$\phi(\cdot)$ denotes the density function of the standard normal distribution.
From~\citet[Eq.~(C.32) in Appendix]{deb2021efficiency}, $\E \left[\phi(e_j^\top \Sigma^{-1/2}\mathbf{X})\right] = \frac{1}{2\sqrt{\pi}}$.
Hence, the above equation can be further simplified to
$$\begin{aligned}
3\sum_{j=1}^p \left(\sum_{i=1}^p  \xi_i \Sigma^{-1/2}_{ji} \frac{1}{\sqrt{\pi}}\right)^2=\frac{3}{\pi}\|\Sigma^{-1/2}{\B \xi} \|^2.
\end{aligned}$$ which yields the desired result. \newline

\noindent Part 3.~(Spherical uniform ERD with identity score function):
Write $\nabla\psi(\mathbf{x})=(r_1(\mathbf{x}),\ldots,r_p(\mathbf{x}))$.
For the spherical uniform distribution $\Sigma_{\rm ERD} = \frac{1}{3p} I$.
Using the integration by parts formula,
$$\begin{aligned}\|\Sigma_{\rm ERD}^{-1/2}{\B \gamma}\|^2 &= 3p\|{\B \gamma}\|^2 =3p \left\|\mathbb{E}_{\textrm{H}_0} \left[ \nabla\psi(\mathbf{X})\frac{{\B \xi}^\top \nabla f (\mathbf{X})}{f(\mathbf{X})}\right]\right\|^2\\
&=3p\sum_{j=1}^p\left(\sum_{i=1}^p  \E \left[\xi_i\frac{\partial}{\partial X_i} r_j(\mathbf{X}) \right]\right)^2\\
&=3\|\Sigma^{-1/2}{\B \xi}\|^2\cdot \frac{1}{p}\left[\frac{1}{2^{p-1}}\cdot\frac{\Gamma(p-0.5)}{(\Gamma(p/2))^2} +\frac{\sqrt{2\pi}(p-1)}{2^{p/2}\Gamma(p/2)}\E[|G_1|^{p-2}]-\mathcal{C}_p  \right]^2,
\end{aligned}$$
where the last equality follows from~\citet[Eq.~(C.44) in Appendix]{deb2021efficiency} with the notations therein. 
Hence,
${\rm ARE}(\mathbf{W}_n,\bar{\mathbf{X}}_n;{\B \xi})=\kappa_p$.
\qed

\subsection{Proof of \cref{thm:indepcomp}}\label{sec:pf_thm:indepcomp}
\noindent Part 1. (${\rm ERD} = {\rm Uniform}(-1,1)^p$ with identity score function):
In this case $\Sigma_{\rm ERD} = \frac{1}{3}I$, and it can be shown that
$$\nabla\psi(\mathbf{x}) = \left(2F_1(x_1)-1,\ldots,2F_p(x_p)-1 \right)^\top,$$
where $F_i$ is the cumulative distribution function (cdf) for $\tilde{f}_i$. To see this, note that the above map pushes $\mu$ to ${\rm Uniform}(-1,1)^p$, and is the gradient of the convex function: $2\sum_{i=1}^p \int_{-\infty}^{x_i}F_i(t_i){\rm d}{t_i} - \mathbf{1}^\top \mathbf{x}$. Hence,
$$\begin{aligned}
\left\|\Sigma_{\rm ERD}^{-1/2}\mathbb{E}_{\textrm{H}_0} \left[ \nabla\psi(\mathbf{X})\frac{{\B \xi}^\top \nabla f (\mathbf{X})}{f(\mathbf{X})}\right]\right\|^2 &= 3\left\|\mathbb{E}_{\textrm{H}_0} \left[ \nabla\psi(\mathbf{X})\frac{{\B \xi}^\top \nabla f (\mathbf{X})}{f(\mathbf{X})}\right]\right\|^2\\
&=3\sum_{i=1}^p \left( \E_{\textrm{H}_0}\left[(2F_i(X_i)-1)\frac{{\B \xi}^\top \nabla f (\mathbf{X})}{f(\mathbf{X})} \right] \right)^2\\
&=12\sum_{i=1}^p \left( \E_{\textrm{H}_0}\left[F_i(X_i)\frac{{\B \xi}^\top \nabla f (\mathbf{X})}{f(\mathbf{X})} \right] \right)^2\\
&=12\sum_{i=1}^p \left( \int F_i(x_i)\xi_i  \tilde{f}_i' (x_i){\rm d}x_i \right)^2\\
&=12\sum_{i=1}^p \left(\xi_i \int  \tilde{f}_i^2 (x_i){\rm d}x_i \right)^2.
\end{aligned}$$
Let $\sigma_i^2$ be the variance of $F_i$. Then $\|\Sigma_\mathbf{X}^{-1/2}{\B \xi}\|^2 = {\B \xi}^\top \Sigma_\mathbf{X}^{-1}{\B \xi} = \sum_{i=1}^p \frac{\xi_i^2}{\sigma^2_i}$.

From the derivation of ARE of Wilcoxon's test w.r.t.\ the $t$-test in one-dimensional case \citep{Hodges1956efficiency}, we know:
$$\sigma_i^2 \left(\int \tilde{f}_i^2(x_i){\rm d}x\right)^2 \geq \frac{9}{125}.$$
This implies
$${\rm ARE}(\mathbf{W}_n,\bar{\mathbf{X}}_n;{\B \xi})=\frac{\left\|\Sigma_{\rm ERD}^{-1/2}\mathbb{E}_{\textrm{H}_0} \left[ \nabla\psi(\mathbf{X})\frac{{\B \xi}^\top \nabla f (\mathbf{X})}{f(\mathbf{X})}\right]\right\|^2}{\| \Sigma_\mathbf{X}^{-1/2}{\B \xi} \|^2}\geq \frac{108}{125}.$$ \newline

\noindent Part 2.~(Gaussian ERD with identity score function): In this case $\Sigma_{\rm ERD} = I$, and
$$\nabla\psi(\mathbf{x}) = \left(\Phi^{-1} \circ F_1(x_1),\ldots,\Phi^{-1} \circ F_p(x_p) \right).$$
To see this, note that the above map pushes $\mu$ to $N(\mathbf{0},I)$, and is the gradient of the convex function  $\psi(\mathbf{x}) := \sum_{i=1}^p \int_{-\infty}^{x_i}\Phi^{-1}\circ F_i(t_i){\rm d}{t_i}$.
Hence,
\begin{equation}\label{eq:normal_bound}\begin{aligned}
\left\|\Sigma_{\rm ERD}^{-1/2}\mathbb{E}_{\textrm{H}_0} \left[ \nabla\psi(\mathbf{X})\frac{{\B \xi}^\top \nabla f (\mathbf{X})}{f(\mathbf{X})}\right]\right\|^2 &= \left\|\mathbb{E}_{\textrm{H}_0} \left[ \nabla\psi(\mathbf{X})\frac{{\B \xi}^\top \nabla f (\mathbf{X})}{f(\mathbf{X})}\right]\right\|^2\\
&=\sum_{i=1}^p \left( \E_{\textrm{H}_0}\left[\Phi^{-1} \big(F_i(X_i)\big)\frac{{\B \xi}^\top \nabla f (\mathbf{X})}{f(\mathbf{X})} \right] \right)^2\\
&=\sum_{i=1}^p \left( \int \Phi^{-1} \big(F_i(x_i)\big)\xi_i  \tilde{f}_i' (x_i)\,{\rm d}x_i \right)^2\\
&=\sum_{i=1}^p \left(\xi_i \int \frac{1}{\phi\left(\Phi^{-1} \big(F_i(x_i)\big)\right)} \tilde{f}_i^2 (x_i)\,{\rm d}x_i \right)^2,
\end{aligned}\end{equation}
where the last equality follows from integration by parts, with $\phi$ being the density of $N(0,1)$.
It follows from the optimization problem in one-dimension \citep[Theorem 2.1]{Gastwirth1968elementary} that
\begin{equation}\label{eq:normal_bound_inf}
\sigma_i^2  \left( \int \frac{1}{\phi\left(\Phi^{-1} \big(F_i(x_i)\big)\right)} \tilde{f}_i^2 (x_i)\,{\rm d}x_i \right)^2 \geq 1,
\end{equation}
and the minimizing density is that of $N(0,\sigma^2)$ for any $\sigma>0$. This implies:
\begin{equation}\label{eq:indepARE}
{\rm ARE}(\mathbf{W}_n,\bar{\mathbf{X}}_n;{\B \xi})= \frac{\sum_{i=1}^p \left(\xi_i \int \frac{1}{\phi\left(\Phi^{-1} \big(F_i(x_i)\big)\right)} \tilde{f}_i^2 (x_i)\,{\rm d}x_i \right)^2}{ \sum_{i=1}^p \frac{\xi_i^2}{\sigma^2_i}}\geq 1.
\end{equation}
\qed

\subsection{Proof of \cref{thm:ellip}}\label{sec:pf_ellip}
We assume the same regularity conditions on $\underline{f}$ as in
\citet[Appendix C.5.3]{deb2021efficiency}:
\begin{itemize}
    \item $\int_{\R^+}r^{d+1}\underline{f}(r){\rm d}r<\infty$.
    \item $\sqrt{\underline{f}}(\cdot)$ admits a weak derivative, which is denoted by $\left(\sqrt{\underline{f}}\right)'(\cdot)$. This means that for all $\psi:\R^+\to\R^+ $ which is compactly supported and infinitely differentiable, we have
    $$\int_{\R^+} \sqrt{\underline{f}(r)}\psi'(r)\,{\rm d}r = -\int_{\R^+}\left(\sqrt{\underline{f}}\right)'(r)\psi(r){\rm d}r. $$
    \item $\int_{\R^+} r^{d-1}\left[\left(\sqrt{\underline{f}}\right)'(r)\right]^2\, {\rm d}r <\infty .$
\end{itemize}

\noindent Part 1. (Spherical uniform ERD with identity score):
Write $\bar{\mathbf{X}}:=\Sigma^{-1/2}(\mathbf{X}-{\B \theta})$, which has a spherically symmetric distribution under the null, and its covariance matrix is given by
$\frac{1}{p} \E\|\bar{\mathbf{X}}\|^2 I_p$. Hence, $\Sigma_\mathbf{X} = \frac{1}{p} \E\|\bar{\mathbf{X}}\|^2\cdot  \Sigma$ and 
$$\| \Sigma_\mathbf{X}^{-1/2}{\B \xi} \|^2 = {\B \xi}^\top \Sigma_\mathbf{X}^{-1}{\B \xi}=\frac{p}{ \E\|\bar{\mathbf{X}}\|^2}\cdot {\B \xi}^\top \Sigma^{-1}{\B \xi}.$$
From \citet[Eq.~(C.55) in Appendix]{deb2021efficiency},
$$\left\|\Sigma_{\rm ERD}^{-1/2}\mathbb{E}_{\textrm{H}_0} \left[ \nabla\psi(\mathbf{X})\frac{{\B \xi}^\top \nabla f (\mathbf{X})}{f(\mathbf{X})}\right]\right\|^2=
\frac{3}{p}\cdot {\B \xi}^\top \Sigma^{-1}{\B \xi} \left( \E\bar{h}\left( \|\bar{\mathbf{X}}\|\right)+(p-1)\E\left[ \frac{\bar{H}\left(\|\bar{\mathbf{X}}\|\right)}{\|\bar{\mathbf{X}}\|}\right]\right)^2,$$
where $\bar{H}$ is the distribution function of $\|\bar{\mathbf{X}}\|$, and $\bar{h}$ is the corresponding density function.
Hence,

$$\begin{aligned}
{\rm ARE}(\mathbf{W}_n,\bar{\mathbf{X}}_n;{\B \xi})&=\frac{\left\|\Sigma_{\rm ERD}^{-1/2}\mathbb{E}_{\textrm{H}_0} \left[ \nabla\psi(\mathbf{X})\frac{{\B \xi}^\top \nabla f (\mathbf{X})}{f(\mathbf{X})}\right]\right\|^2}{\| \Sigma_\mathbf{X}^{-1/2}{\B \xi} \|^2}\\
&=\frac{3}{p^2}\cdot \E \|\bar{\mathbf{X}}\|^2 \cdot \left(\E\bar{h}\left(\|\bar{\mathbf{X}}\|\right) + (p-1) \E\left[ \frac{\bar{H}\left(\|\bar{\mathbf{X}}\|\right)}{\|\bar{\mathbf{X}}\|}\right]\right)^2\\
&\geq \frac{81}{500}\cdot \frac{(\sqrt{2p-1}+1)^5}{p^2(\sqrt{2p-1}+5)},
\end{aligned}$$
where the inequality follows from the optimization problem solved in \citet[Eq.~(C.56) in Appendix]{deb2021efficiency}. The density that attains the minimum is also given in \citet[Appendix C.5.3]{deb2021efficiency}. \newline

\noindent Part 2. (Gaussian ERD with identity score function):
From \citet[Lemma C.2]{deb2021efficiency},
$$\nabla\psi(\mathbf{X})=\frac{\Sigma^{-1/2}(\mathbf{X}-{\B \theta})}{\|\Sigma^{-1/2}(\mathbf{X}-{\B \theta}) \|}\cdot H_p^{-1}\left(\bar{H}\left(\|\Sigma^{-1/2}(\mathbf{X}-{\B \theta})\| \right) \right),$$
where $H_p$ is the distribution function of a $\sqrt{\chi_p^2}$ distribution.
From \citet[Eq.~(C.60) in Appendix]{deb2021efficiency}:
$$\begin{aligned}
{\rm ARE}(\mathbf{W}_n,\bar{\mathbf{X}}_n;{\B \xi})&=\frac{\left\|\Sigma_{\rm ERD}^{-1/2}\mathbb{E}_{\textrm{H}_0} \left[ \nabla\psi(\mathbf{X})\frac{{\B \xi}^\top \nabla f (\mathbf{X})}{f(\mathbf{X})}\right]\right\|^2}{\| \Sigma_\mathbf{X}^{-1/2}{\B \xi} \|^2}\\
&=\frac{\E \|\bar{\mathbf{X}}\|^2}{p^3} \cdot
\left\{\E \left[ \frac{\bar{h}\left(\|\bar{\mathbf{X}}\|\right)}{h_p(H_p^{-1}(\bar{H}(\|\bar{\mathbf{X}}\|)))} \right]+ (p-1) \E\left[ \frac{H_d^{-1}\circ\bar{H}\left(\|\bar{\mathbf{X}}\|\right)}{\|\bar{\mathbf{X}}\|}\right]\right\}^2\\
&\geq 1.
\end{aligned}$$
The equality is attained if and only if the infimum in \citet[Eq.~(C.60) in Appendix]{deb2021efficiency} is attained, i.e., when $f(\cdot)$ is a multivariate normal distribution \citep[Theorem 3.5]{deb2021efficiency}.
\qed

\subsection{Proof of \cref{thm:blind}}\label{sec:pf_blind}
Suppose $\tilde{f}_i(\cdot)$ has variance $\sigma_i^2$. Then $\mathbf{W}$ has covariance matrix ${\rm diag}(\sigma_1^2,\ldots,\sigma_p^2)$, and
$\Sigma_{\mathbf{X}}={\rm Cov}(A\mathbf{W})=A \, {\rm diag}(\sigma_1^2,\ldots,\sigma_p^2)A^\top$, $\| \Sigma_\mathbf{X}^{-1/2}{\B \xi} \|^2 = {\B \xi}^\top A \, {\rm diag}\left(\frac{1}{\sigma_1^2},\ldots,\frac{1}{\sigma_p^2}\right)A^\top {\B \xi}$.
Observe that under the null, with $A=(a_{i,j})$, for $\mathbf{x} =(x_1,\ldots, x_p) \in \R^p$,
$$f(\mathbf{x}) = \prod_{i=1}^p \tilde{f}_i\left(\sum_{j=1}^p a_{j,i}x_j \right).$$
Hence, with $w_i = \sum_{j=1}^p a_{j,i}x_j $ (or $\mathbf{w}=A^\top \mathbf{x}$),
$$\frac{{\B \xi}^\top\nabla f(\mathbf{X})}{f(\mathbf{X})}=\sum_{i,j=1}^p \xi_i a_{i,j}\frac{\tilde{f}_j'}{\tilde{f}_j}(w_j).$$
Let $\tilde{F}_i$ be the distribution function corresponding to $\tilde{f}_i$.
Let $R_\mathbf{W}(\mathbf{w})=\left(\Phi^{-1}\circ \tilde{F}_1(w_1),\ldots, \Phi^{-1}\circ \tilde{F}_p(w_p)\right)$.
By \citet[Lemma A.8]{ghosal2019multivariate},
$\nabla \psi(\mathbf{X}) = AR_\mathbf{W}(A^\top \mathbf{X})$. Hence,
$$\begin{aligned}
\left\|\Sigma_{\rm ERD}^{-1/2}\mathbb{E}_{\textrm{H}_0} \left[ \nabla\psi(\mathbf{X})\frac{{\B \xi}^\top \nabla f (\mathbf{X})}{f(\mathbf{X})}\right]\right\|^2
&=\left\|\mathbb{E}_{\textrm{H}_0} \left[ AR_\mathbf{W}(A^\top \mathbf{X})\sum_{i,j=1}^p \xi_i a_{i,j}\frac{\tilde{f}_j'}{\tilde{f}_j}(W_j)\right]\right\|^2\\
&=\left\|\mathbb{E}_{\textrm{H}_0} \left[ R_\mathbf{W}(\mathbf{W})\sum_{i,j=1}^p \xi_i a_{i,j}\frac{\tilde{f}_j'}{\tilde{f}_j}(W_j)\right]\right\|^2\\
&=\sum_{k=1}^p \left(\E_{\textrm{H}_0}\left[\Phi^{-1}\circ \tilde{F}_k(W_k) \sum_{i,j=1}^p \xi_i a_{i,j}\frac{\tilde{f}_j'}{\tilde{f}_j}(W_j) \right] \right)^2.
\end{aligned}$$
Under the null, $W_j$ is independent of $W_k$ when $j\neq k$ and thus $\E_{\textrm{H}_0}\left[\Phi^{-1}\circ \tilde{F}_k(W_k) \frac{\tilde{f}_j'}{\tilde{f}_j}(W_j) \right]=0$. Hence
$$\begin{aligned}
\E_{\textrm{H}_0}\left[\Phi^{-1}\circ \tilde{F}_k(W_k) \sum_{i,j=1}^p \xi_i a_{i,j}\frac{\tilde{f}_j'}{\tilde{f}_j}(W_j) \right]&=
\sum_{i=1}^p \xi_i a_{i,k}\cdot \E_{\textrm{H}_0}\left[\Phi^{-1}\circ \tilde{F}_k(W_k) \frac{\tilde{f}_k'}{\tilde{f}_k}(W_k) \right]\\
&\geq \sum_{i=1}^p \xi_i a_{i,k}\cdot \frac{1}{\sigma_k} = \frac{(A^\top {\B \xi})_k}{\sigma_k},
\end{aligned}$$
where the inequality follows from \eqref{eq:normal_bound} and \eqref{eq:normal_bound_inf}. Hence,
$$\begin{aligned}
{\rm ARE}(\mathbf{W}_n,\bar{\mathbf{X}}_n;{\B \xi})&=\frac{\left\|\Sigma_{\rm ERD}^{-1/2}\mathbb{E}_{\textrm{H}_0} \left[ \nabla\psi(\mathbf{X})\frac{{\B \xi}^\top \nabla f (\mathbf{X})}{f(\mathbf{X})}\right]\right\|^2}{\| \Sigma_\mathbf{X}^{-1/2}{\B \xi} \|^2}\geq \frac{\sum_{k=1}^p \left( \frac{(A^\top {\B \xi})_k}{\sigma_k}\right)^2 }{{\B \xi}^\top A\, {\rm diag}\left(\frac{1}{\sigma_1^2},\ldots,\frac{1}{\sigma_p^2}\right)A^\top {\B \xi}}=1.
\end{aligned}$$
\qed

\subsection{Proof of \cref{thm:efficiency}}\label{pf:thm:efficiency}
If the model $(P_\theta:{\B \theta}\in\Theta)$ is differentiable in quadratic mean, then the local experiments converge to a Gaussian limit \citep[Section 15.3]{vanderVaart1998}:
$$\left( P_{{\B \theta}_0+{\B \xi}/\sqrt{n}}^n:{\B \xi}\in\R^p\right)\to\left( N({\B \xi},I_{{\B \theta}_0}^{-1}):{\B \xi}\in\R^p\right),$$
where $I_{{\B \theta}_0}$ is the Fisher information at ${\B \theta}_0=\mathbf{0}_p$.

\noindent{\bf 1.} We consider testing~\eqref{eq:Hypo-Test-Local}. The asymptotic representation theorem \citep[Theorem 15.1]{vanderVaart1998} shows that if we have a sequence of level $\alpha$ tests with power function $\pi_n(\cdot)$ such that $\pi_n({\B \xi}/\sqrt{n})\to\pi({\B \xi})$ for every ${\B \xi}$ (as $n \to \infty$) and some function $\pi$, then $\pi(\cdot)$ is the power function of a level $\alpha$ test in the Gaussian limit experiment.
Now let us consider the Gaussian limit experiment:
For testing $\tilde{{\rm H}}_0:{\B \xi}=\mathbf{0}_p$ versus $\tilde{\rm H}_1:{\B \xi}\neq \mathbf{0}_p$ based on an observation $\mathbf{Y}\sim N({\B \xi}, I_{{\B \theta}_0}^{-1})$, the level $\alpha$ test that maximizes the minimum power over $\{{\B \xi}:\|I^{1/2}_{{\B \theta}_0} {\B \xi}\|=c\}$ has the critical region $\{\mathbf{y} \in \R^p:\|I^{1/2}_{{\B \theta}_0}\mathbf{y}\|^2\geq \chi^2_p(\alpha)\}$, for all $c>0$ \citep[Problem 15.4]{vanderVaart1998}. This implies
\begin{equation}\label{eq:efficiency_up_bound}
\inf_{{\B \xi}: \|I^{1/2}_{{\B \theta}_0} {\B \xi}\|=c}\lim_{n\to\infty} \pi_n \left(\frac{{\B \xi}}{\sqrt{n}} \right)\leq \mathbb{P}\left(\|N\big({\B \xi}_c,I_p\big)\|^2\geq \chi_p^2(\alpha)\right),
\end{equation}
for any ${\B \xi}_c\in\R^p$ with $\|{\B \xi}_c\|=c$.

Next, we show that by choosing a proper score function $J(\cdot)$, we can achieve the above optimal power (on the right side of~\eqref{eq:efficiency_up_bound}).
From \cref{thm:asymptotics_local_alternatives}, under ${\rm H}_1$,
$$\mathbf{W}_n\overset{d}{\longrightarrow}  N\left({\B \gamma},\Sigma_{\rm ERD}\right),$$
where ${\B \gamma} = -\mathbb{E}_{\textrm{H}_0} \left[J(\nabla\psi(\mathbf{X}))\frac{{\B \xi}^\top \nabla f (\mathbf{X})}{f(\mathbf{X})}\right]\in\R^p$.
In \cref{thm:population_rank}, we show that $\nabla \psi(\mathbf{x})$ is the optimal transport map under square loss that pushes $P_{{\B \theta}_0}$ to $\nu_S$.
If $\nu_S$ has a Lebesgue density, then McCann's theorem \citep[Remark 16]{McCann1995} states that $\nabla \psi(\cdot )$ admits an almost sure inverse $\nabla\psi^*:\R^p\to\R^p$ such that
$\nabla\psi^*(\nabla\psi(\mathbf{x}))=\mathbf{x}$ ($P_{{\B \theta}_0}$ a.e.-$\mathbf{x}$) and $\nabla\psi \circ \nabla\psi^* (\mathbf{x}) = \mathbf{x}$ ($\nu_S$ a.e.-$\mathbf{x}$).

With $J(\cdot)=-\frac{\nabla f(\nabla\psi^*(\cdot))}{f(\nabla\psi^*(\cdot))}$,
$J(\nabla\psi(\mathbf{x}))=-\frac{\nabla f(\mathbf{x})}{f(\mathbf{x})}$ ($P_{{\B \theta}_0}$ a.e.-$\mathbf{x}$).
Observe that,
$$J(Q \nabla\psi(\mathbf{x})) = J(\nabla\psi(Q\mathbf{x})) = -\frac{\nabla f(Q\mathbf{x})}{f(Q\mathbf{x})},$$
for any $Q\in \G$. Since $f$ is $\G$-symmetric, $f(Q\mathbf{x}) = f(\mathbf{x})$ for all $Q\in \G$ and $\mathbf{x}\in\R^p$, which further implies $\nabla f(Q\mathbf{x}) = Q\nabla f(\mathbf{x})$. Hence,
$J(Q\psi(\mathbf{x}))=Q J(\psi(\mathbf{x}))$, showing that $J(\cdot)$ is $\G$-equivariant ($\nu_S$ a.e.).
With this choice of $J$,
$$\begin{aligned}
{\B \gamma} &=  \mathbb{E}_{\textrm{H}_0} \left[ \frac{\nabla f(\mathbf{X})}{f(\mathbf{X})} \frac{\nabla f (\mathbf{X})^\top}{f(\mathbf{X})}\right]{\B \xi}= I_{{\B \theta}_0}{\B \xi},\\
\Sigma_{\rm ERD} &={\rm Var}_{\rm H_0}\left(J(\nabla\psi(\mathbf{X}))\right) =  {\rm Var}_{\rm H_0}\left(\frac{\nabla f(\mathbf{X})}{f(\mathbf{X})}\right) = I_{{\B \theta}_0}.
\end{aligned}$$
Hence, under ${\rm H}_1$,
$\mathbf{W}_n\overset{d}{\longrightarrow}  N\left(I_{{\B \theta}_0}{\B \xi},I_{{\B \theta}_0}\right)$, and therefore
$$\mathbb{P}_{{\B \theta}_0+ \frac{{\B \xi}}{\sqrt{n}}}\left(\mathbf{W}_n^\top \Sigma_{\rm ERD}^{-1} \mathbf{W}_n \geq \chi_p^2(\alpha)\right)\to 
\mathbb{P}\left(\|N(I_{{\B \theta}_0}^{1/2}{\B \xi}, I_p)\|^2 \geq \chi_p^2(\alpha)  \right).$$
Combining with \eqref{eq:efficiency_up_bound}, we have
$$\inf_{{\B \xi}: \|I^{1/2}_{{\B \theta}_0} {\B \xi}\|=c}\lim_{n\to\infty} \pi_n \left(\frac{{\B \xi}}{\sqrt{n}} \right)\leq \inf_{{\B \xi}: \|I^{1/2}_{{\B \theta}_0} {\B \xi}\|=c} \lim_{n\to\infty} \mathbb{P}_{{\B \theta}_0 + \frac{{\B \xi}}{\sqrt{n}}}\left(\mathbf{W}_n^\top \Sigma_{\rm ERD}^{-1} \mathbf{W}_n \geq \chi_p^2(\alpha)\right),$$ which yields the desired result. \newline

\noindent{\bf 2.} Recall that $\pi_n(\cdot)$ is the power function of a sequence of level $\alpha$ tests for~\eqref{eq:one-side-test}. By \citet[Theorem 15.4]{vanderVaart1998},
\begin{equation}\label{eq:Upper-Bd-one-sided}
\limsup_{n\to\infty}\pi_n\left( {\B \theta}_0 + \frac{{\B \xi}}{\sqrt{n}}\right) \leq 1-\Phi\left( z_\alpha - \frac{{\B \xi}^\top \nabla \zeta ({\B \theta}_0)}{\sqrt{\nabla \zeta ({\B \theta}_0)^\top I_{{{\B \theta}_0}}^{-1}\nabla \zeta ({\B \theta}_0)}}\right),
\end{equation}
where $\Phi(\cdot)$ is the cdf of a standard normal distribution, and $z_\alpha$ satisfies $\Phi(z_\alpha)=1-\alpha$.

From the proof of part-1., $\mathbf{W}_n\overset{d}{\longrightarrow}  N\left(I_{{\B \theta}_0}{\B \xi},I_{{\B \theta}_0}\right)$ under ${\B \theta} = {{\B \theta}_0 + \frac{{\B \xi}}{\sqrt{n}}}$, and therefore
$$\nabla \zeta ({\B \theta}_0)^\top I_{{\B \theta}_0}^{-1} \mathbf{W}_n\overset{d}{\longrightarrow}N\big({\B \xi}^\top \nabla \zeta ({\B \theta}_0), \nabla \zeta ({\B \theta}_0)^\top I_{{{\B \theta}_0}}^{-1}\nabla \zeta ({\B \theta}_0)\big).$$
This implies
$$\lim_{n\to\infty}\mathbb{P}_{{\B \theta}_0 + \frac{{\B \xi}}{\sqrt{n}}}\left(\frac{\nabla \zeta ({\B \theta}_0)^\top I_{{\B \theta}_0}^{-1} \mathbf{W}_n}{\sqrt{\nabla \zeta ({\B \theta}_0)^\top I_{{{\B \theta}_0}}^{-1}\nabla \zeta ({\B \theta}_0)}} \geq z_{\alpha} \right)=1-\Phi\left( z_\alpha - \frac{{\B \xi}^\top \nabla \zeta ({\B \theta}_0)}{\sqrt{\nabla \zeta ({\B \theta}_0)^\top I_{{{\B \theta}_0}}^{-1}\nabla \zeta ({\B \theta}_0)}}\right),$$ which attains the upper bound in~\eqref{eq:Upper-Bd-one-sided}.

\qed

\section{Some Auxiliary Results}\label{sec:Appendix-C}
In the following, we state and prove some auxiliary results. Results of a similar form can be found in~\citet[Lemma 9]{McCann1995} and~\citet[Theorem 5.10]{villani2009OT}.
\begin{prop}\label{prop:existence_c_cyc}
    If $c(\cdot,\cdot)$ is continuous, then there exists $\pi \in \Pi(\mu,\nu)$ with a $c$-cyclically monotone support (see \cref{def:c-cyc}).
\end{prop}
\begin{proof}
Let $\mu_n := \frac{1}{n}\sum_{i=1}^n \delta_{\mathbf{x}_i}$ and $\nu_n := \frac{1}{n}\sum_{i=1}^n \delta_{\mathbf{y}_i}$ be such that $\mu_n \overset{w}{\longrightarrow}\mu$ and $\nu_n \overset{w}{\longrightarrow}\nu$ (for example, we can draw $\mathbf{x}_i$ i.i.d.\ from $\mu$ and $\mathbf{y}_i$ i.i.d.\ from $\nu$).
For each $n$, solve
\begin{equation}\label{eq:LP}
    \min_{a_{ij}\geq 0}\sum_{i,j=1}^n a_{ij}c(\mathbf{x}_i,\mathbf{y}_j),\qquad {\rm\ s.t.\ }\sum_{j=1}^n a_{ij} = 1,i=1,\ldots,n,{\rm\ and\ }\sum_{i=1}^n a_{ij} = 1,j=1,\ldots,n
\end{equation}
to obtain a minimizer $(a_{ij}^n)$. In fact, this is the assignment problem \citep{Munkres1957assign,Bertsekas1988auction} and there exists an optimal solution such that $a_{ij}^n\in\{0,1\}$ for all $i,j$.
Let $\pi_n := \sum_{i,j=1}^na_{ij}^n\delta_{(\mathbf{x}_i,\mathbf{y}_j)}$.
Then $\pi_n$ is supported on $S$, the set of all couples $(\mathbf{x}_i,\mathbf{y}_j)$ such that $a_{ij}^n > 0$.
To show that $\pi_n \in \Pi(\mu_n,\nu_n)$ has a $c$-cyclically monotone support, we argue by contradiction. Suppose that there exists $N \in \mathbb{N}$ and $(\mathbf{x}_{i_1},\mathbf{y}_{i_1}),\ldots,(\mathbf{x}_{i_N},\mathbf{y}_{i_N})$ in $S$ such that (with $\mathbf{y}_{i_{N+1}} = \mathbf{y}_{i_{1}}$)
$$c(\mathbf{x}_{i_1},\mathbf{y}_{i_2}) + c(\mathbf{x}_{i_2},\mathbf{y}_{i_3})+\ldots +c(\mathbf{x}_{i_N},\mathbf{y}_{i_{N+1}})<c(\mathbf{x}_{i_1},\mathbf{y}_{i_1})+\ldots +c(\mathbf{x}_{i_N},\mathbf{y}_{i_N}).$$
Define
$$\tilde{\pi}_n := \pi_n + \varepsilon \sum_{l=1}^N\left(\delta_{(\mathbf{x}_{i_l},\mathbf{y}_{i_{l+1}}) } - \delta_{(\mathbf{x}_{i_l},\mathbf{y}_{i_{l}}) }\right).$$
Suppose $\tilde{\pi}_{n} = \sum_{i,j=1}^n \tilde{a}_{ij}^n\delta_{(\mathbf{x}_i,\mathbf{y}_j)}$. Then, for a sufficiently small $\varepsilon>0$,
$\tilde{a}_{ij}^n \geq 0$, and $(\tilde{a}_{ij}^n)$ has a smaller cost than $({a}_{ij}^n)$ in \eqref{eq:LP}, which yields a contradiction.

Hence, $\pi_n \in \Pi(\mu_n,\nu_n)$ has a $c$-cyclically monotone support. The weak convergence of $\mu_n$ and $\nu_n$ implies that $\{\mu_n\}$ and $\{\nu_n\}$ are tight, and thus $\{\pi_n\}$ must be tight (\citet[Lemma 4.4]{villani2009OT}). By Prokhorov's theorem, there is a subsequence $\{\pi_{n_k}\}_{k \ge 1}$ which converges weakly to some probability measure $\pi$.
Since $(\mathbf{x},\mathbf{y})\mapsto \mathbf{x}$ and $(\mathbf{x},\mathbf{y})\mapsto \mathbf{y}$ are continuous, by continuous mapping theorem, $\pi$ has marginals $\mu$ and $\nu$. For each $n$, the $c$-cyclically monotonicity of $\pi_n$ implies that for all $N \in \mathbb{N}$ and $\pi_n^{\otimes N}$-almost all $(\mathbf{x}_1,\mathbf{y}_1),\ldots,(\mathbf{x}_N,\mathbf{y}_N)$, we have (with $\mathbf{y}_{N+1} = \mathbf{y}_{1}$)
\begin{equation}\label{eq:c-cyc-N}
\sum_{i=1}^Nc(\mathbf{x}_i,\mathbf{y}_i)\leq \sum_{i=1}^Nc(\mathbf{x}_i,\mathbf{y}_{i+1}).
\end{equation}
In other words, $\pi_n^{\otimes N}$ is concentrated on the set $\mathcal{C}(N)$ of all $((\mathbf{x}_1,\mathbf{y}_1),\ldots,(\mathbf{x}_N,\mathbf{y}_N))\in (\R^p\times \R^p)^N$ satisfying \eqref{eq:c-cyc-N}. Since $c(\cdot,
\cdot)$ is continuous, $\mathcal{C}(N)$ is a closed set, so the weak limit
$\pi^{\otimes N} $ of $\pi^{\otimes N}_{n_k} $ is also concentrated on $\mathcal{C}(N)$ by Portmanteau theorem. Let $\Gamma$ be the support of $\pi$. Then $\Gamma^N = ({\rm supp}(\pi))^N={\rm supp} (\pi^{\otimes N})\subset \mathcal{C}(N)$. Since this holds true for all $N$, $\Gamma$ is $c$-cyclically monotone.
\end{proof}

\begin{prop}\label{prop:existence_c_cvx}
    If $c(\cdot,\cdot)$ is continuous and $\pi\in \Gamma(\mu,\nu)$ has $c$-cyclically monotone support, then there exists a $c$-convex function $\psi$ such that $\partial_c\psi$ (see \cref{def_cconvex}) contains the support of $\pi$.
\end{prop}
\begin{proof}
Let $\Gamma$ denote the support of $\pi$. 
Fix $(\mathbf{x}_0,\mathbf{y}_0)\in \Gamma$. Define
$$\begin{aligned}\psi(\mathbf{x})&:=\sup_{m\in\mathbb{N}}\sup \big\{[c(\mathbf{x}_0,\mathbf{y}_0) - c(\mathbf{x}_1,\mathbf{y}_0)] +  [c(\mathbf{x}_1,\mathbf{y}_1) - c(\mathbf{x}_2,\mathbf{y}_1)]\\
&\quad +\ldots
+[c(\mathbf{x}_m,\mathbf{y}_m) - c(\mathbf{x},\mathbf{y}_m)]:(\mathbf{x}_1,\mathbf{y}_1),\ldots,(\mathbf{x}_m,\mathbf{y}_m)\in\Gamma\big\}.
\end{aligned}$$
From the definition, with $m=1$ and $(\mathbf{x}_1,\mathbf{y}_1) = (\mathbf{x}_0,\mathbf{y}_0)$, we have
$\psi(\mathbf{x}_0)\geq 0$. On the other hand, the $c$-cyclical monotonicity of $\Gamma$ implies that
$\psi(\mathbf{x}_0)\leq 0$, and hence $\psi(\mathbf{x}_0)= 0$.
By renaming $\mathbf{y}_m$ as $\mathbf{y}$,
$$\begin{aligned}
\psi(\mathbf{x})&=\sup_{\mathbf{y}\in\R^p}\sup_{m\in\mathbb{N}}\sup_{(\mathbf{x}_1,\mathbf{y}_1),\ldots,(\mathbf{x}_{m-1},\mathbf{y}_{m-1}),\mathbf{x}_m }\big\{[c(\mathbf{x}_0,\mathbf{y}_0) - c(\mathbf{x}_1,\mathbf{y}_0)] +  [c(\mathbf{x}_1,\mathbf{y}_1) - c(\mathbf{x}_2,\mathbf{y}_1)]\\
&\quad +\ldots
+[c(\mathbf{x}_m,\mathbf{y}) - c(\mathbf{x},\mathbf{y})]:(\mathbf{x}_1,\mathbf{y}_1),\ldots,(\mathbf{x}_m,\mathbf{y})\in\Gamma\big\}\\
&=\sup_{\mathbf{y}\in\R^p} [\zeta(\mathbf{y}) - c(\mathbf{x},\mathbf{y})],
\end{aligned}$$
where $\zeta(\mathbf{y}):=-\infty$ if there does not exist $\mathbf{x}_m$ such that $(\mathbf{x}_m,\mathbf{y})\in \Gamma$, and otherwise
$$\begin{aligned}
\zeta(\mathbf{y})&:=\sup_{m\in\mathbb{N}}\sup_{(\mathbf{x}_1,\mathbf{y}_1),\ldots,(\mathbf{x}_{m-1},\mathbf{y}_{m-1}),\mathbf{x}_m }\big\{[c(\mathbf{x}_0,\mathbf{y}_0) - c(\mathbf{x}_1,\mathbf{y}_0)] +  [c(\mathbf{x}_1,\mathbf{y}_1) - c(\mathbf{x}_2,\mathbf{y}_1)]\\
&\quad +\ldots
+[c(\mathbf{x}_m,\mathbf{y}) - c(\mathbf{x},\mathbf{y})]:(\mathbf{x}_1,\mathbf{y}_1),\ldots,(\mathbf{x}_m,\mathbf{y})\in\Gamma\big\}.
\end{aligned}$$
Thus $\psi(\cdot)$ is a $c$-convex function.
For any $(\bar{\mathbf{x}},\bar{\mathbf{y}})\in \Gamma$, with $(\mathbf{x}_m,\mathbf{y}_m)=(\bar{\mathbf{x}},\bar{\mathbf{y}})$ in the definition of $\psi$,
$$\begin{aligned}\psi(\mathbf{x})&\geq \sup_{m\geq 2}\sup_{(\mathbf{x}_1,\mathbf{y}_1),\ldots,(\mathbf{x}_{m-1},\mathbf{y}_{m-1})\in \Gamma} \big\{[c(\mathbf{x}_0,\mathbf{y}_0) - c(\mathbf{x}_1,\mathbf{y}_0)] +  [c(\mathbf{x}_1,\mathbf{y}_1) - c(\mathbf{x}_2,\mathbf{y}_1)]\\
&\quad +\ldots
+[c(\mathbf{x}_{m-1},\mathbf{y}_{m-1}) - c(\bar{\mathbf{x}},\mathbf{y}_{m-1})] + [c(\bar{\mathbf{x}},\bar{\mathbf{y}}) - c(\mathbf{x},\bar{\mathbf{y}})]\big\}\\
&=\psi(\bar{\mathbf{x}}) + [c(\bar{\mathbf{x}},\bar{\mathbf{y}}) - c(\mathbf{x},\bar{\mathbf{y}})].
\end{aligned}$$
Hence, $\psi(\mathbf{x}) + c(\mathbf{x},\bar{\mathbf{y}})\geq \psi(\bar{\mathbf{x}}) + c(\bar{\mathbf{x}},\bar{\mathbf{y}})$. Taking the infimum over $\mathbf{x}\in \R^p$, we deduce that
$\psi^c(\bar{\mathbf{y}})\geq \psi(\bar{\mathbf{x}}) + c(\bar{\mathbf{x}},\bar{\mathbf{y}})$. Since the reverse inequality always holds from the definition of $\psi^c$ (\cref{def_cconvex}), we have $\psi^c(\bar{\mathbf{y}})= \psi(\bar{\mathbf{x}}) + c(\bar{\mathbf{x}},\bar{\mathbf{y}})$. This means that $(\bar{\mathbf{x}},\bar{\mathbf{y}})\in \partial_c\psi$, and so $\Gamma\subset \partial_c\psi$.
\end{proof}

\bibliography{symmetry}
\bibliographystyle{apalike}
\end{document}